%% file: main.tex
\documentclass[manuscript,screen]{acmart}

\usepackage{bm}
\usepackage{algorithm, algpseudocode}
\algnewcommand{\LineComment}[1]{\State \(\triangleright\) #1}
\usepackage{wrapfig}
\usepackage{courier}
\usepackage{ragged2e}
\usepackage{amsmath}
\usepackage{url}
\usepackage{graphicx}
\usepackage{subcaption}
\usepackage{enumitem}
\usepackage{varwidth}
\usepackage{xcolor}

\input{macros}

\AtBeginDocument{%
  }

\setcopyright{acmlicensed}
\copyrightyear{2018}
\acmYear{2018}
\acmDOI{XXXXXXX.XXXXXXX}
\acmConference[Conference acronym 'XX]{Make sure to enter the correct
  conference title from your rights confirmation email}{June 03--05,
  2018}{Woodstock, NY}

\acmISBN{978-1-4503-XXXX-X/2018/06}

\begin{document}

\title[\shortlogic~ Logic for Non-Markovian Reasoning]{Lattice Annotated Temporal (\shortlogic) Logic for Non-Markovian Reasoning}

\author{Kaustuv Mukherji}
\email{kmukherj@syr.edu}
\orcid{0000-0001-8044-1110}

\author{Jaikrishna Manojkumar Patil}
\email{jpatil01@syr.edu}
\orcid{0009-0000-3745-9147}

\author{Dyuman Aditya}
\email{daditya@syr.edu}
\orcid{0000-0002-4889-3499}

\author{\\Paulo Shakarian}
\email{pashakar@syr.edu}
\orcid{0000-0002-3159-4660}

\affiliation{%
  \institution{Syracuse University}
  \city{Syracuse}
  \state{New York}
  \country{USA}
  }

\author{Devendra Parkar}
\email{dparkar1@asu.edu}
\orcid{0009-0009-0133-8875}

\author{Lahari Pokala}
\email{lpokala@asu.edu}
\orcid{0009-0007-4199-3255}

\affiliation{%
  \institution{Arizona State University}
  \city{Tempe}
  \state{Arizona}
  \country{USA}
  }

\author{Clark Dorman}
\email{clark.dorman@ssci.com}
\affiliation{%
  \institution{Scientific Systems Company, Inc.}
  \city{Woburn}
  \state{Massachusetts}
  \country{USA}
  }

\author{Gerardo I. Simari}
\email{gis@cs.uns.edu.ar}
\orcid{0000-0003-3185-4992}
\affiliation{%
  \institution{Department of Computer Science and Engineering, Universidad Nacional del Sur (UNS) \&
  Institute for Computer Science and Engineering (ICIC UNS-CONICET)}
  \city{Bahia Blanca}
  \country{Argentina}
  }

\renewcommand{\shortauthors}{K. Mukherji et al.}

\begin{abstract}
  We introduce Lattice Annotated Temporal (\shortlogic) Logic, an extension of Generalized Annotated Logic Programs (GAPs) that incorporates temporal reasoning and supports open-world semantics through the use of a lower lattice structure. This logic combines an efficient deduction process with temporal logic programming to support non-Markovian relationships and open-world reasoning capabilities.  The open-world aspect, a by-product of the use of the lower-lattice annotation structure, allows for efficient grounding through a Skolemization process, even in domains with infinite or highly diverse constants. 
  We provide a suite of theoretical results that bound the computational complexity of the grounding process, in addition to showing that many of the results on GAPs (using an upper lattice) still hold with the lower lattice and temporal extensions (though different proof techniques are required). Our open-source implementation, PyReason, features modular design, machine-level optimizations, and direct integration with reinforcement learning environments. Empirical evaluations across multi-agent simulations and knowledge graph tasks demonstrate up to three orders of magnitude speedup and up to five orders of magnitude memory reduction while maintaining or improving task performance. Additionally, we evaluate \logic’s value in reinforcement learning environments as a non-Markovian simulator, achieving up to three orders of magnitude faster simulation with improved agent performance, including a 26\% increase in win rate due to capturing richer temporal dependencies. These results highlight \logic’s potential as a unified, extensible framework for open world temporal reasoning in dynamic and uncertain environments. Our implementation is available at: \textbf{\url{pyreason.syracuse.edu}}.
\end{abstract}




\maketitle

\section{Introduction \label{sec:intro}}

Recent work in Artificial Intelligence (AI) has looked at the use of logic programming to bring explainability and robustness to modern AI systems~\cite{bueff2023deep, dai2019bridging, cropper2020turning}. 
Temporal logic programming~\cite{dekhtyar99,APTL,martiny2016pdt,doder2024probabilistic} provides the ability to directly model non-Markovian temporal relationships (i.e., the agent's behavior can be dependent on multiple previous time steps) due to its ability to reason about if-then statements that span multiple time steps (see Figure~\ref{fig:intro-geo-program} for an example logic program). This capability would make temporal logic programming a suitable candidate to replace standard Markov Decision Processes (MDPs) prevalently found in various reinforcement learning systems. However, temporal logic programming, such as APT logic~\cite{APTL}, is intractable not only in terms of time but also space due to grounding, thereby limiting its applicability to such use-cases.  
In this paper, we overcome the problem with \textbf{L}attice \textbf{A}nnotated \textbf{T}emporal Logic or \logic. 
This logic is created by adding temporal extensions to a fragment of Generalized Annotated Programs (GAPs)~\cite{ks92}, thereby permitting direct modeling of non-Markovian temporal relationships, which allows exact but tractable inference in a variety of cases.  We retain the capabilities to represent uncertainty by virtue of the annotations, but leverage a lower-lattice structure (as opposed to an upper-lattice structure in GAPs), which not only enables open-world reasoning but also affords an efficient Skolemization process, enabling scalable grounding.  We provide a theoretical foundation for \logic~and a suite of experimental results, including a demonstration of effectiveness as a replacement for MDPs in reinforcement learning (RL) applications.

To make some of our ideas more concrete, we introduce our first running example,  modeled after the use cases explored in our experiments. 
In Figure~\ref{fig:intro-geo-diagram}, we consider a simple example where two agents, one on foot patrol (shown with an icon of a person) and another in a patrol car (car icon) move in a geospatial area (marked by the square). 
The speed of the car is double that of the agent on foot. Figure~\ref{fig:intro-geo-program} shows an excerpt from the logic program that governs the agents' movements;  
note that in the syntax of annotated logic, truth values follow the atoms after a colon.  
While we will describe lower-lattice based truth values in the technical preliminaries, we note that $[1,1]$ denotes truth, $[0,0]$ denotes falsehood, and $[0,1]$ denotes total uncertainty---any subset of the unit interval can be a truth value, allowing for expressions of various levels of first and second-order uncertainty. 
In Figure~\ref{fig:intro-geo-program}, the first two rules show how agents may move in different directions with different speeds, creating new constants in space when rules are fired. The last two rules help update the truth values to false when an agent moves away from the old location. Here, variable $A$ is used for agents, while $L_1, L_2$ can be grounded with constants in space. The $\delay$ is used here to capture the differing speeds between agents, showing how the temporal component can be used to capture the dynamics of a non-Markovian environment. Note that the $\delay$ between antecedent and consequent can be heterogeneous within the logic program.  For this example, at $t=0$, the patrol car chooses to move left, and the foot patrol decides to go right. Following these action choices, the aforementioned rules are grounded and in effect two new constants are grounded in space, as shown in Figure~\ref{fig:intro-geo-diagram} with red pins, after one and two timesteps.


\begin{figure}[tb]
    \begin{center}
        \fbox{
            \parbox{0.65\linewidth}{
            \begin{align}
                \Pi_{geo}=\{~&at(A,L_2):[1,1] \xleftarrow[\delay = 1]~~ at(A,L_1):[1,1] \wedge moveLeft(A):[1,1] \wedge speed(A,fast):[1,1] \wedge left(L_1,L_2):[1,1]~,  \notag \\
                & \text{If fast-moving agent $A$ is at $L_1$ and moves left, it will be at $L_2$ (left of $L_1$) after one time unit.}  \notag \\
                &at(A,L_2):[1,1] \xleftarrow[\delay = 2]~~ at(A,L_1):[1,1] \wedge moveRight(A):[1,1] \wedge speed(A,slow):[1,1] \wedge right(L_1,L_2):[1,1]~,  \notag \\
                & \text{If slow-moving agent $A$ is at $L_1$ and moves right, it will be at $L_2$ (right of $L_1$) after two time units.}  \notag \\
                &at(A,L_1):[0,0] \xleftarrow[\delay = 1]~~ at(A,L_1):[1,1] \wedge moveLeft(A):[1,1]~, \notag \\
                & \text{If agent $A$ is at $L_1$ and moves left, it will no longer be at $L_1$ after one time unit.}  \notag \\
                &at(A,L_1):[0,0] \xleftarrow[\delay = 1]~~ at(A,L_1):[1,1] \wedge moveRight(A):[1,1]~\}\notag \\
                & \text{If agent $A$ is at $L_1$ and moves right, it will no longer be at $L_1$ after one time unit.}  \notag
            \end{align}
            }
          }
    \end{center}
    \caption{Excerpt of logic program $\Pi_{geo}$ for the geospatial example shown in Figure~\ref{fig:intro-geo-diagram}. English translations for each rule are also provided. \label{fig:intro-geo-program}}
    \Description{Excerpt of logic program $\Pi_{geo}$ for the geospatial example shown in Figure~\ref{fig:intro-geo-diagram}. English translations for each rule are also provided.}
\end{figure}

\begin{figure}[tb]
    \begin{center}
        \begin{subfigure}[t]{0.302\columnwidth}
            \vskip 0pt
            \includegraphics[width=\linewidth]{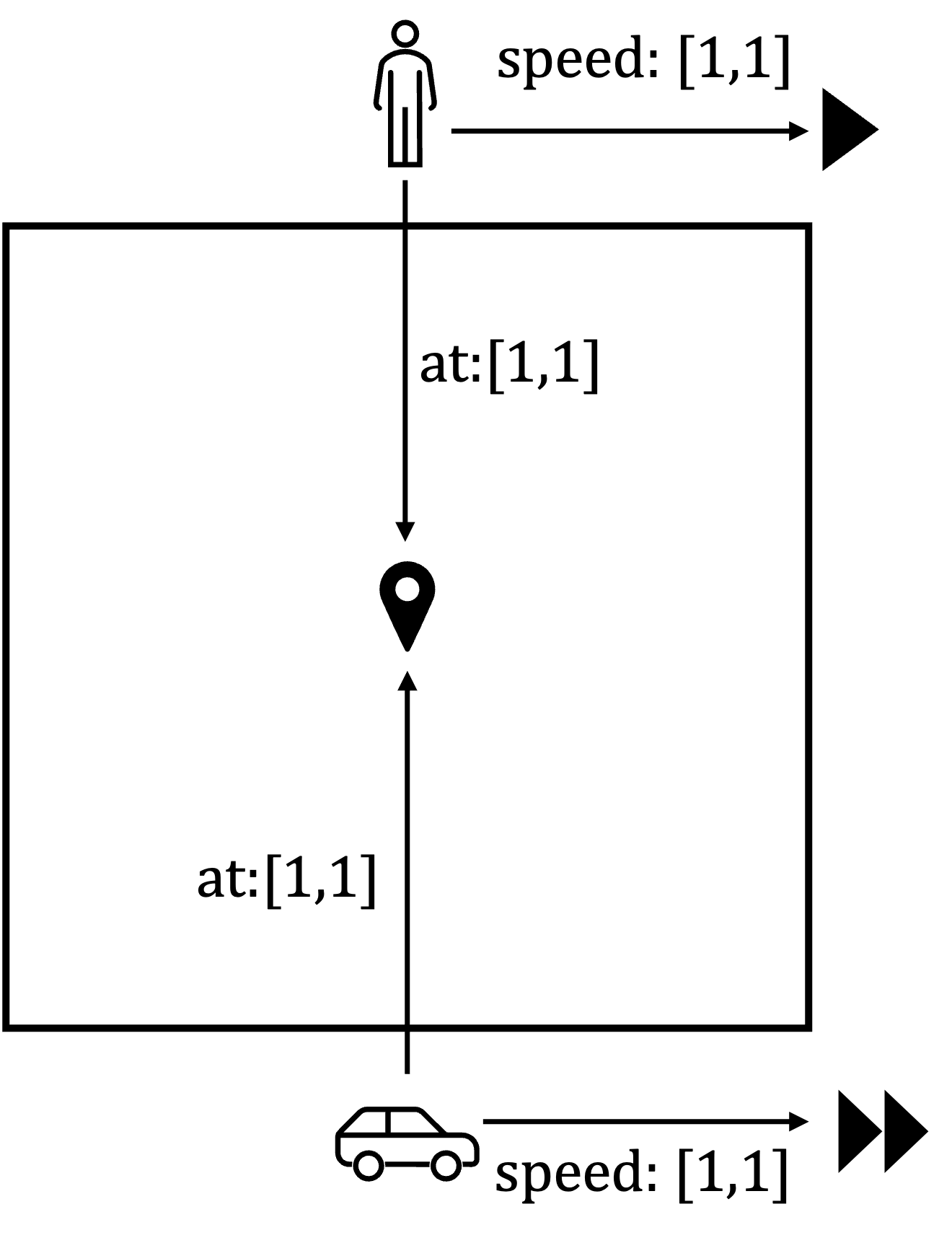}
            \caption*{t=0}
        \end{subfigure}
        \begin{subfigure}[t]{0.311\columnwidth}
            \vskip 0pt
            \includegraphics[width=\linewidth]{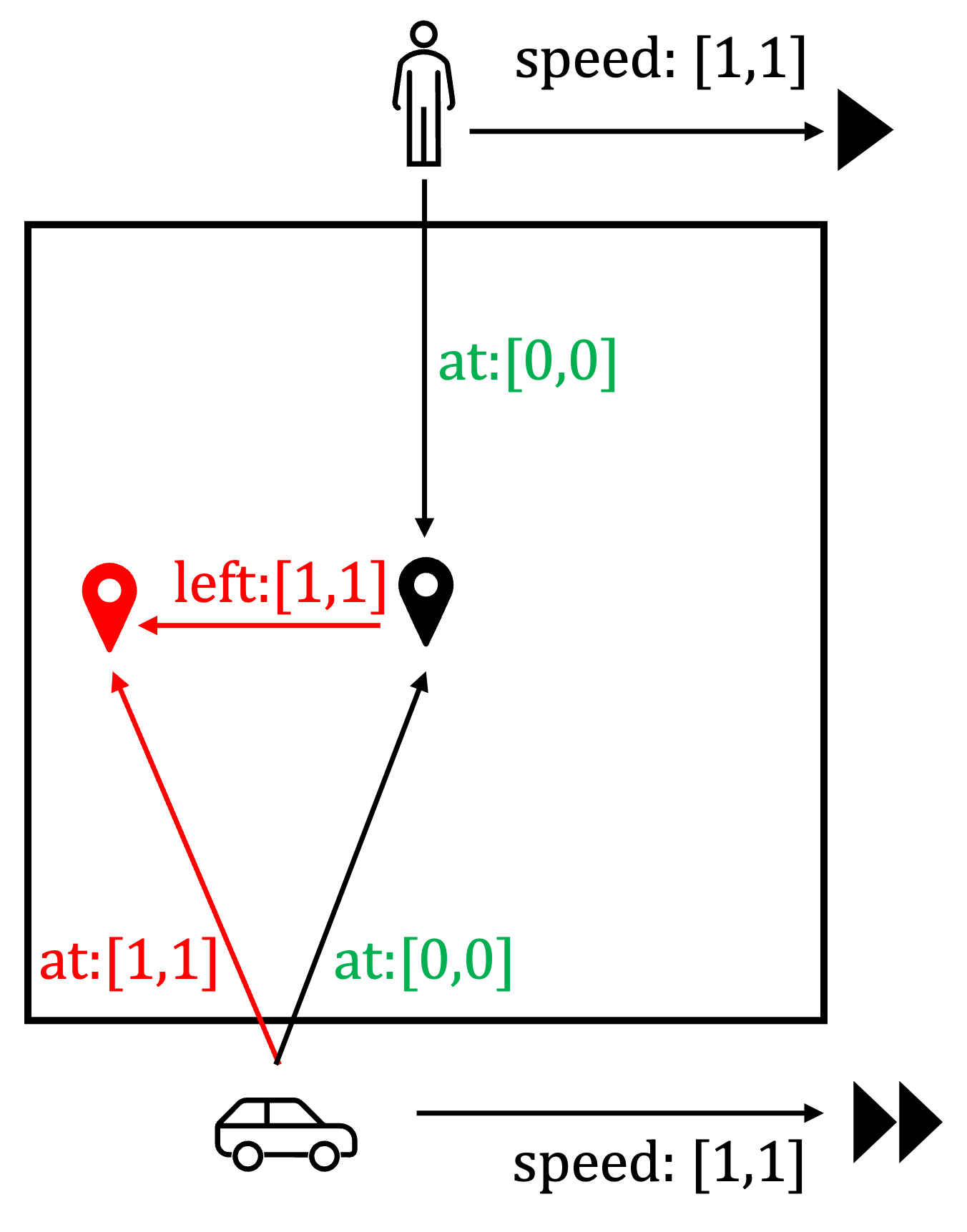}
            \caption*{t=1}
        \end{subfigure}
        \begin{subfigure}[t]{0.317\columnwidth}
            \vskip 0pt
            \includegraphics[width=\linewidth]{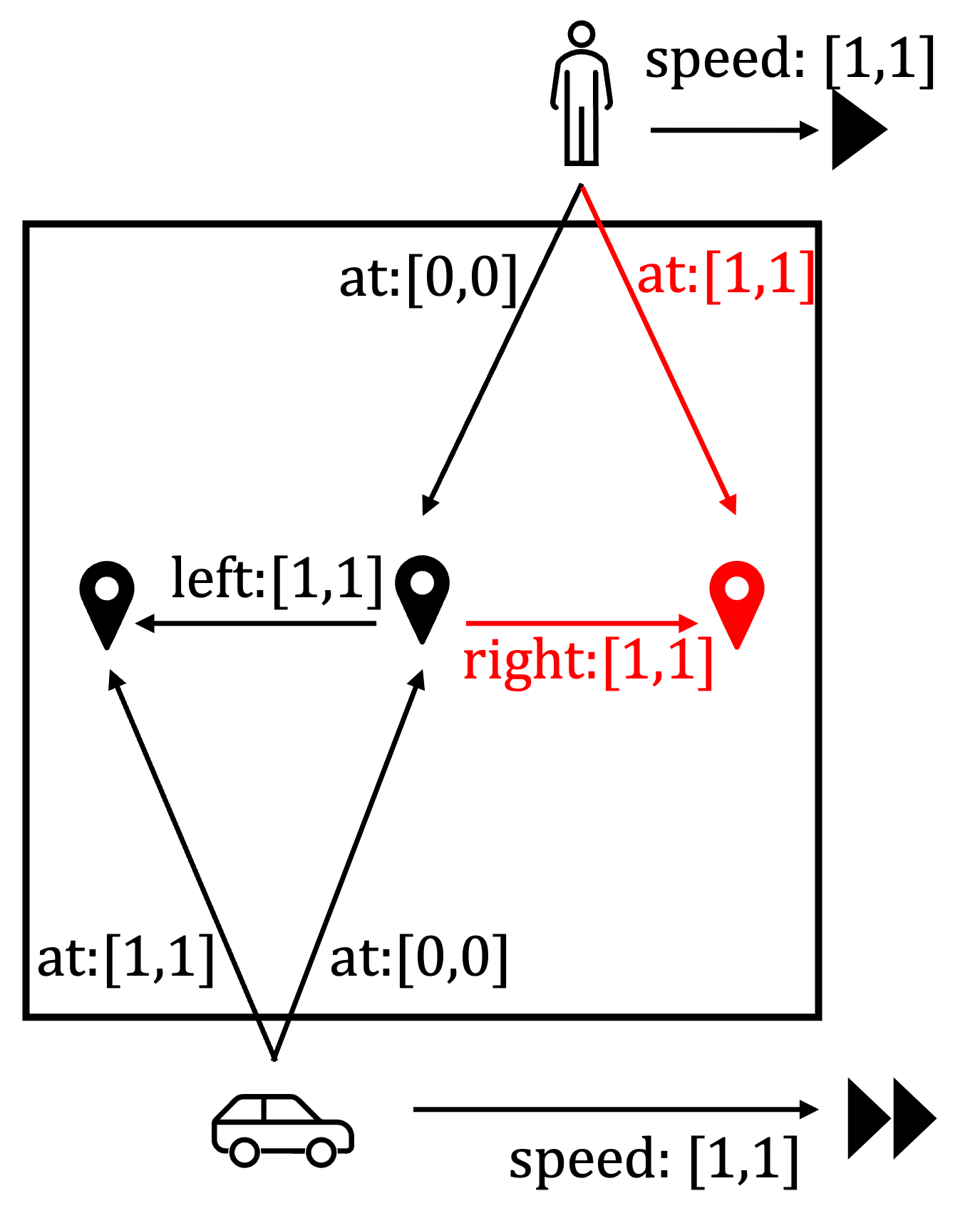}
            \caption*{t=2}
        \end{subfigure}
    \end{center}
    \caption{Geospatial example: Creation of new atoms during inference for two time steps (left to right). Newly created atoms at each time point are shown in red. Existing atoms whose annotations change are marked in green.}
    \Description{Geospatial example: Creation of new atoms during inference for two time steps (left to right). Newly created atoms at each time point are shown in red. Existing atoms whose annotations change are marked in green.}
    \label{fig:intro-geo-diagram}
\end{figure}

To build the capabilities presented throughout the paper in these examples, we provide the following contributions:
\begin{enumerate}
    \item We introduce \logic, an extension of Generalized Annotated Logic Programs (GAPs) that includes temporal extensions and leverages a lower lattice instead of an upper lattice -– a departure from the traditional GAPs on which it is based. We formalize the syntax and semantics and show how using a lower lattice for annotated logic facilitates open-world reasoning. We reprove results for satisfaction, consistency, and entailment for the extended logic with lower lattice and temporal aspects. Then we introduce the fixpoint operator that allows for deductive reasoning while mapping timepoint literal pairs to timepoint literal pairs, which readily supports non-Markovian reasoning. We show that the use of the lower lattice not only allows for open-world reasoning but also a natural form of Skolemization of constants and atoms that has applications to geospatial and knowledge graph reasoning tasks. We show that theoretical results on correctness for GAPs proven in~\cite{shakarian2022extensions} hold for our temporal extensions. In section~\ref{sec:skolemization}, we analytically bound the creation of new ground atoms during reasoning, which suggests a significant reduction in grounding resulting from this form of Skolemization, and can provide significant speed-up in certain applications. Later, in Sections~\ref{sec:geospatial} and~\ref{sec:kg-comp}, we experimentally show that this result is a loose upper bound for practical domains where constants and atoms are diverse and sparse.

    \item We describe an implementation of the logic~fixpoint operator with highly optimized data structures that can leverage the Skolemization properties. The implementation is built on modern Python and is machine code optimized to handle graph-based data structures compatible with all common graph-based databases. Grounding and inference processes are parallelized across threads to leverage computing resources optimally, thereby reducing running time. It is also modular, where new constants and rules can be added to the existing program at any point. Our temporal extensions also allow rules whose effects are observed at different intervals from when they are fired, and this directly supports building logic programs for non-Markovian dynamics. Among other applications, all of these capabilities make this implementation especially suitable for addressing common challenges seen in RL use cases. Our implementation, called PyReason, is available as an open-source project at \url{https://pyreason.syracuse.edu}. Various software and hardware acceleration techniques used are detailed in Section~\ref{sec:software-hardware}.

    \item Our experimental evaluation encompasses two distinct applications: a multi-agent geospatial simulation and knowledge graph completion. The empirical results corroborate our theoretical analysis, demonstrating reductions of up to three orders of magnitude in the size of groundings for multi-step reasoning. Notably, while our theoretical analysis provides a conservative upper bound, practical observations show significantly fewer ground atoms with higher numbers of fixpoint applications. This results in enhanced scalability in both computational speed and memory efficiency, further improving the approach’s applicability in real-world scenarios.

    \item We provide an extensive suite of experiments showing how Skolemization arising from the lower lattice of annotations provides multiple orders of magnitude of speedup when leveraged by efficient data structures present in our implementation. In Section~\ref{sec:geospatial}, we show the efficacy of our Skolemization-based approach where we observe several orders of speedup - growing significantly with the number of ground atoms. Similarly, memory reduction is observed to increase with a higher number of ground atoms, giving up to five orders of magnitude of savings for the highest setting in our experiment. Section~\ref{sec:kg-comp} shows the scaling capability of our approach on four popular knowledge bases: WN18RR and FB15k-237~\cite{wn18rr_fb15k237}, YAGO03-10~\cite{yago}, and UMLS~\cite{umls}, with varying numbers of rules. Our approach is shown to always provide a speedup and memory reduction, which slowly settles towards the base resources as large number of constants, rules, and reasoning steps slowly moves towards a fully connected network. We also do some initial exploration of the potential of multi-step reasoning by showing how two-step reasoning can provide better results on information retrieval tasks across subsets of all four datasets.

    \item Finally, in Section~\ref{sec:rl-expt} we present a suite of experimental results where we leverage both the speedups described above along with the temporal characteristics of the logic to allow for efficient non-Markovian simulation of an environment for use in training a reinforcement learning agent.  We show up to three orders of magnitude speedup compared with two simulation environments while maintaining agent performance. When non-Markovian capabilities are used, we show a notable 26\% improvement in agent win rate, compared to when limited by the Markov assumption.
\end{enumerate}

The rest of the paper is organized as follows. In Section~\ref{sec:related}, we review some related, well-established logics, discuss the importance of capturing non-Markovian dynamics for modern applications, and other relevant work. 
In Section~\ref{sec:prelims}, we describe the syntax and semantics of \logic, as well as provide a running example, which is an excerpt from domains used in our experiments. In Section~\ref{sec:formal}, we provide theoretical results on our temporal extensions to GAPs and resulting performance gains. We detail our open-source implementation of \logic~in Section~\ref{sec:implementation}. Sections~\ref{sec:geospatial}~and~\ref{sec:kg-comp} contain two sets of experiments that verify the theoretical performance gains derived in Section~\ref{sec:skolemization}, and show the scaling capability of our implementation on two diverse applications. In Section~\ref{sec:rl-expt}, we show how our implementation is amenable to be used effectively as a simulator for reinforcement learning applications. Finally, in Section~\ref{sec:conclusion}, we summarize our findings and share our thoughts on future work.

\section{Related Work}
\label{sec:related}

Logic programming emerged decades ago as a foundational paradigm in artificial intelligence, enabling the formal representation of knowledge through declarative statements and supporting automated reasoning via deductive inference~\cite{kowalski88}. This approach facilitated the systematic derivation of consequences from encoded knowledge, significantly advancing the capabilities of early AI systems. Common implementations of logic programming like Prolog~\cite{colmerauer1982prolog} and Datalog~\cite{datalog} are designed for non-temporal applications and do not support uncertainty or annotations -- While it is always possible to devise special syntax (predicates, constants, etc.) to encode some form of temporal reasoning, modeling heterogeneous non-Markovian relationships in such frameworks is not part of their design.
Another popular branch of logic inspired by Prolog is Answer Set Programming (ASP)~\cite{marek1999stable, niemelä99}, which employs stable model semantics~\cite{gl8811988stable} and SAT-inspired algorithms, supports non-monotonic reasoning, and can handle incomplete information, making it popular for knowledge representation, robotics, and bioinformatics~\cite{erdem2016applications}. However, classical ASP or its equivalent reasoning formalisms like Equilibrium logic~\cite{eqlogic} do not support temporal reasoning. Temporal Equilibrium Logic (TEL)~\cite{teqlogic} extends them by using modal temporal operators. However, inference in TEL is intractable~\cite{teqcomplexity}, making it unsuitable for large problems, thus precluding it from the large-scale temporal reasoning experiments carried out in this work.
Other notable temporal logics include Linear Temporal Logic~\cite{ltl77} and Computation Tree Logic~\cite {clarke1986automatic}, but they are not designed for deductive reasoning; rather, they are primarily used for formal verification applications. We note that the semantic structures often used for the verification of such logics—e.g., Kripke structures~\cite{cresswell2012new} and Markov Decision Processes for probabilistic variants~\cite{hansson1994logic, CLEAVELAND2005316}—inherently incorporate Markov assumptions, which we do not make in this formalism.

Similarly, in modern machine learning systems, despite numerous advances in reasoning about graphs and environments, the environment is typically modeled as a Markov Decision Process (MDP), where the next state and reward depend solely on the current state and action. However, in real-world applications, this is seldom the case~\cite{gupta2021non} and, as a result, it is impossible to make an optimal decision based only on the current state~\cite{schmidhuber1990reinforcement}. Recent work on Reinforcement Learning (RL)~\cite{chandak2024reinforcement} illustrates the errors introduced when applying standard RL algorithms like Q-learning to non-Markovian environments. RL algorithms are well-known for their substantial data and training time requirements. These demands become even more pronounced when dealing with environments characterized by non-Markovian dynamics. Most works have attempted to tackle this issue by making the reward non-markovian, but keeping the underlying dynamics of the system markovian~\cite{gaon2020reinforcement, agarwal2023reinforcement}. Meanwhile, Gupta et al.~\cite{gupta2021non} approximates non-Markovian dynamics as a fractional dynamical system to reduce data demands of model-free RL. While these have shown promise, and despite the consensus for the need to incorporate non-Markovian dynamics in ML systems, significant progress has remained elusive due to prohibitive data and time requirements. An important direction for future work would be to use algorithms designed to learn non-Markovian policies on non-Markovian structures.

Classical logic programming approaches following MDPs cannot represent non-Markovian time relationships. The idea of allowing non-Markovian temporal dependencies -- specifically, rules featuring heterogeneous time lags -- was introduced in Annotated Probabilistic Temporal (APT) Logic~\cite{APTL}, which extended concepts from Temporal Probabilistic Logic Programs (TPLP)~\cite{dekhtyar1999temporal} and was subsequently developed further in Probabilistic Doxastic Temporal (PDT) Logic~\cite{martiny2016pdt} and by Doder et al.~\cite{doder2024probabilistic}. APT logic demonstrated that such rule structures enable expressive capabilities not achievable by Markov Decision Processes (MDPs). However, the deductive inference presented in these works is intractable. In this paper, we retain the expressive strengths of these non-Markovian constructs within our logical framework, while providing tractable semantics that facilitate efficient reasoning. We also investigate the use of non-Markovian policy learning with annotated logic, a direction that, to our knowledge, has not been previously pursued in the literature, likely due to the intractable nature of APT logic inference.

Finally, in this paper we leverage a Skolemization procedure for efficiency that is possible due to our use of lower-lattice semantics. Skolemization~\cite{Ginsberg} is a technique traditionally used for automatic theorem proving, and has been employed in predicate calculus for the substitution of existentially quantified variables with Skolem function applications for several decades~\cite{morgan1974symbolic,loveland1978automated}. This process has subsequently been applied in the conversion of dynamic semantics into constants and functions~\cite{schubert1999dynamic}, and the generation of normal forms for logical formulae~\cite{akama2011meaning}. More recently, it has been used for translation in Answer Set Programming~\cite{diller-etal-2019-making}. However, we are not aware of any work that studies Skolemization with respect to annotated logic, including with respect to lower-lattice semantics.
A major advantage of deductive or exact reasoning is that it ensures consistency and is often considered explainable~\cite{marques2024logic}. However, black-box models, approximate or inexact reasoners, have become prevalent in real-world applications due to their flexibility and efficiency across domains involving diverse data structures and types~\cite{he2016deep,vaswani2017attention}. The proposed need-based grounding technique in \logic, based on Skolemization, significantly reduces the footprint of a program, allowing for a speed-up, which enables more extensive applications.

\section{Technical Preliminaries}
\label{sec:prelims}

\logic~extends Generalized Annotated Logic programs (GAPs)~\cite{ks92} by incorporating lower-lattice and temporal components.  The lower lattice semantics were introduced in our prior work~\cite{shakarian2022extensions}, which did not include temporal extensions, implementation, or many of the theoretical results in this paper. This framework employs GAPs with a lower lattice to model open-world scenarios, allowing atoms to be associated with a range of values beyond just ``true'' or ``false''. The temporal extensions facilitate the representation of dynamic environments, while the underlying semantic structures and fixpoint approach enhance explainability in complex systems.

In subsequent subsections, we introduce the syntax and semantics of this logic. We also provide a running example that builds on the geospatial example introduced in Section~\ref{sec:intro}.


\subsection{Syntax \label{sec:syntax}}

We consider a first-order logical language with finite sets $\constantSet$, $\predicateSet$, and $\variableSet$ of constant, predicate, and variable symbols, respectively.
Following convention, we use uppercase letters for variables and lowercase for constants.

If a predicate symbol is directly applied to a list comprised of elements from the set of variables and constants ($\variableSet \cup \constantSet$), it forms an atom.

\begin{definition}[Atom]
    If $E_1,..., E_n\in \variableSet \cup \constantSet$, and $p \in \predicateSet$, then $p(E_1, \ldots, E_n)$ is an atom.
\end{definition}


\begin{example}[Atom \label{ex:atom}]
    Consider the geospatial example shown in Figure~\ref{fig:intro-geo-diagram}. From the instance at $t=0$, we can have the following atoms:
    \begin{enumerate}
        \item Unary atoms: $\texttt{agent}(A)$, $\texttt{location}(L)$
        \item Binary atoms: $\texttt{at}(A, L)$, $\texttt{speed}(A, S)$
    \end{enumerate}
    Here, $A$, $L$, and $S$ are variables, and \texttt{agent}, \texttt{location}, \texttt{at}, and \texttt{speed} are predicate symbols.
\end{example}

An atom, made up of only variables, is called a non-ground atom. If it contains both variables and constants, it is a partially ground atom. A fully instantiated atom, without any variables, is called a ground atom.

\begin{definition}[Ground Atom]
    If $c_1,..., c_n \in \constantSet$ and $p \in \predicateSet$, then $p(c_1,..., c_n)$ is a ground atom.
\end{definition}

\begin{example}[Ground Atom \label{ex:ground-atom}]
As mentioned previously in Section~\ref{sec:intro}, the geospatial example domain has the following constants for typed variables
$A, L,$ and $S$:

\begin{center}
\begin{tabular}{lll}
  \textbf{\underline{$A$}} & \textbf{\underline{$L$}} & \textbf{\underline{$S$}} \\
  $footPatrol$ & $locMid$ & $slow$ \\
  $patrolCar$ & $locLeft$ & $fast$ \\
            & $locRight$ & \\
\end{tabular}
\end{center}

Grounding the atoms in Example~\ref{ex:atom} using these constants, we get:
\begin{enumerate}
    \item Unary ground atoms: \texttt{agent$(footPatrol),$ agent$(patrolCar),$ location$(locMid),$ location$(locLeft),$ location$(locRight)$}.
    \item Binary ground atoms: \texttt{at$(footPatrol, locMid),$ at$(patrolCar, locMid),$ speed$(footPatrol, slow),$ \\speed$(patrolCar, fast),$ \ldots}
\end{enumerate}
\end{example}


Following~\cite{ks92}, we define a lattice structure $\semiLattice$ where elements consist of subsets of the real unit interval, where $[0,1]$ (representing total uncertainty) is the lowest element of the lattice while the upper elements are all intervals $[\ell,u]$ where $\ell=u$; such elements include $[1,1]$ (total truth) and $[0,0]$ (total falsehood). Figure~\ref{fig:lowerLattice} illustrates an example of such a structure. One specific function we define is ``$\neg$'' for negation, which is used in the semantics of \cite{ks92}.  For a given $[l,u]$, $\neg([l,u])=[1-u, 1-l]$.  Note that we also use the symbol  $\neg$ as a connector in our first-order language (following the formalism of \cite{ks92}). Now, we define annotations for atoms. We assume the existence of a set $\avar$ of variable symbols ranging over $\semiLattice$ and a set $\calf$ of function symbols, each of which has an associated arity.

\begin{definition}[Annotation]
\phantom{...}
\begin{enumerate}
    \item Any member of $\semiLattice \cup \avar$ is an annotation.
    \item If $f$ is an $n$-ary function symbol over $\semiLattice$ and $m_1,\ldots,m_n$ are annotations, then $f(m_1,\ldots,m_n)$ is an annotation.
\end{enumerate}
\end{definition}

In annotated logic, atoms are associated with elements of the lattice structure, which enables open-world reasoning. We thus define, as in~\cite{mancalog13,shakarian2022extensions}, \textit{annotated atoms} as $a:\mu$, where $a$ is an atom and $\mu$ is an element of the lattice. An \textit{annotated literal} is either an annotated atom or its negation. Functions and variables are also permitted in the annotations (see above references for further details).

\begin{figure}[tb]
    \begin{center}
        \includegraphics[width=0.75\linewidth, trim = 20 8 20 8, clip]{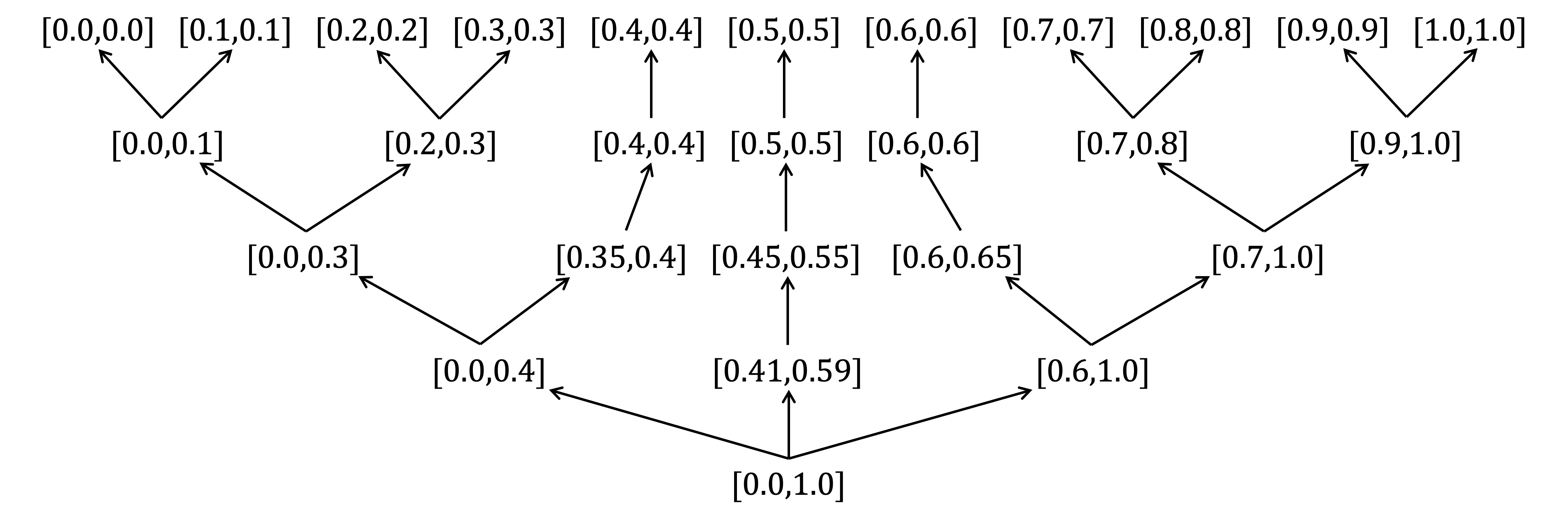}
    \end{center}
    \caption{Example of a lower semi-lattice structure where the elements are intervals in $[0,1]$.}
    \Description{Example of a lower semi-lattice structure where the elements are intervals in $[0,1]$.}
    \label{fig:lowerLattice}
\end{figure}

\begin{example}[Annotation\label{ex:annotation}]

Continuing with our example and, as shown in Figure~\ref{fig:intro-geo-diagram}, an example of an annotated ground atom is \texttt{speed$(footPatrol, slow):[1,1]$}. Its equivalent negation is \texttt{$\neg$speed$(footPatrol, slow):[0,0]$}.
\end{example}


\noindent
In order to extend the logic to make statements about time, and we follow the convention of~\cite{shakarian2012annotated} and define \textit{temporal annotated facts (TAFs)}. 

\begin{definition}[Temporal Annotated Facts (TAFs)]
For a literal $a$, an annotation $\mu$, and a time point $t$, $a:(\mu,t)$ is a temporal annotated fact (TAF).
\end{definition}

We choose to use the notation $a:(\mu,t)$ for TAFs as it makes it quite distinct from annotated atoms without a temporal component (typically denoted by $a:\mu$). We note that, in a previous work~\cite{bavikadisea}, $a:\mu_t$ was used to denote TAFs; however, we have chosen to move away from that notation as subscripts are also used sometimes to differentiate between different literals, constants, components of a vector, etc.

\begin{example}[Temporal Annotated Facts (TAFs) \label{ex:tafs}]
    In our running example, we have three time points, $t = 0, 1, 2$. As the two agents move, their locations are represented by TAFs:
    \begin{center}
    \begin{tabular}{ll}
    $t=0$ & \texttt{at$(footPatrol, locMid):([1,1],0)$} \\
    & \texttt{at$(patrolCar, locMid):([1,1],0)$} \\
    $t=1$ & \texttt{at$(footPatrol, locMid):([0,0],1)$} \\
    & \texttt{at$(patrolCar, locMid):([0,0],1)$} \\
    & \texttt{at$(patrolCar, locLeft):([1,1],1)$} \\
    $t=2$ & \texttt{at$(footPatrol, locMid):([0,0],2)$} \\
    & \texttt{at$(patrolCar, locMid):([0,0],2)$} \\
    & \texttt{at$(patrolCar, locLeft):([1,1],2)$} \\
    & \texttt{at$(footPatrol, locRight):([1,1],2)$} \\
    \end{tabular}
    \end{center}
Note how the ``truth'' about the location of ``$footPatrol$'' is uncertain at $t=1$. However, open-world reasoning supports complete uncertainty, and hence if we were to ground any arbitrary point between the starting and final geo-locations of ``$footPatrol$'' at $t=1$, say, $locMidRight$,  then we can represent it with the TAF:
    \begin{center}
        \texttt{at$(footPatrol, locMidRight):([0,1],1)$}.
    \end{center}
    
\end{example}

\noindent
Annotated literals serve as the building blocks of a GAP rule. We propose the following definition of a modified version of the GAP rules defined in~\cite{shakarian2022extensions}:  

\begin{definition}[GAP Rule]
If $\nonGroundedLiteral_0:\mu_0, \nonGroundedLiteral_1:\mu_1,\ldots,\nonGroundedLiteral_m:\mu_m$ are annotated literals s.t.\ for all $i,j \in 1,...,m$, $\nonGroundedLiteral_i\not\equiv \nonGroundedLiteral_j$, then:
\begin{align}
r\equiv \nonGroundedLiteral_0:\mu_0 &\xleftarrow[\delay]{} \nonGroundedLiteral_1:\mu_1\,\wedge\,\ldots\wedge\, \nonGroundedLiteral_m:\mu_m, \text{ with } \delay \geq 0
\end{align}
is called a \emph{GAP rule}.
We will use the notations $\textit{head}(r)$, $\textit{delay}(r)$, and $\textit{body}(r)$ to denote $\nonGroundedLiteral_0$, $\delay$, 
and $\{\nonGroundedLiteral_1,\ldots,\nonGroundedLiteral_m\}$, resp. 
When $m=0$ ($\textit{body}(r)=\emptyset$), the rule is called a \emph{fact}. 
A GAP rule is \emph{ground} iff there are no occurrences of variables in it. 
\end{definition}
Intuitively, $\delay$ is the temporal gap between when the rule is fired and when its effects hold. 
If $\textit{body}(r)$ is satisfied at time $t$, then the annotation of $\nonGroundedLiteral_0$ changes to $\mu_0$ at time $t+\delay$.

\begin{example}[GAP Rule\label{ex:gap-rule}]
Examples of GAP rules, and their English language equivalent, for the geospatial example are shown in Figure~\ref{fig:intro-geo-program}.
\end{example}

\label{sec:prelim-immediate}
We call rules with $\delay = 0$ immediate rules, which are applied as soon as the body is satisfied. Immediate rules relax the need for interdependent ground rules to be separated by not only applications of fixpoint operators but also actual time points. We use $\delay = 0$ to approximate infinitesimal time intervals. We elaborate on how the implementation achieves this in Section~\ref{sec:immediate}.

A temporal logic program $\Pi$ is a finite set of GAP rules that can be used to capture expert knowledge of an environment's dynamics or learned from data. Incorporating temporal constructs into literals and GAP rules enables explicit representation and reasoning over temporal dependencies, facilitating non-Markovian reasoning by capturing historical states and temporal relations beyond the current state. 

\subsection{Semantics \label{sec:semantics}}

Annotated logic programs are associated with interpretations that map literal-time point pairs to annotations.
Given a program, the intuition is that this structure (which, as shown below, can be produced as output of deductive inference) can directly describe changes in the environment or knowledge. 
Interpretations are symbolic, and hence support explainability based on the underlying logical language.
We now provide a formal definition of interpretations and the associated satisfaction relationship.

\begin{definition}[Interpretation]
\label{def:interp}
Let $\Pi$ be a program, $\groundLiteralSet$ the set of all ground literals, and 
$T = t_1, ..., t_{\textit{max}}$ a sequence of time points. 
An interpretation $\interpretation$ is a mapping  $\groundLiteralSet \times T \to \semiLattice$ such that for all literals~$l$, we have 
$\interpretation(l, t) = \neg(\interpretation(\neg l, t))$.
\end{definition}

\begin{example}[Interpretation]
\label{ex:interpretation}
    Consider the TAF \texttt{at$(patrolCar, locMid):([1,1],0)$} from Example~\ref{ex:tafs}. The interpretation can be represented as: $\interpretation(\texttt{at}(patrolCar, locMid), 0) = [1,1]$; 
    its negation is: $(\interpretation(\neg \texttt{at}(patrolCar, locMid), 0)) = [0,0]$.
\end{example}

The set~$\interpretationSet$ of all interpretations can be partially ordered via the ordering:
$\interpretation_1\preceq \interpretation_2$ iff for all ground literals $\groundedLiteral \in \groundLiteralSet$ and time $t$, 
$\interpretation_1(\groundedLiteral, t)\sqsubseteq \interpretation_2(\groundedLiteral, t)$. 
This set forms a complete lattice under the $\preceq$ ordering.
An interpretation is said to satisfy an annotated ground literal at time $t$ if its annotation is contained in the sub-lattice of its assigned value.

\begin{definition}[Satisfaction]
An interpretation $\interpretation$ \emph{satisfies} an annotated ground literal $\groundedLiteral:\mu$ at time $t$, denoted 
$\interpretation \satisfactionAtTime{t} \groundedLiteral:\mu$, iff $\mu \sqsubseteq \interpretation( \groundedLiteral, t)$. 
\end{definition}

We can then extend the definition of satisfaction from an annotated literal to a complete GAP rule as:

\begin{definition}[Satisfaction of GAP rule]
$\interpretation$ satisfies the ground rule 
\begin{align}
    r\equiv \groundedLiteral_0: \mu_0 \leftarrow_{\delay} \groundedLiteral_1:\mu_1\wedge\,\ldots\,\wedge\, \groundedLiteral_m:\mu_m
\end{align}
denoted $\interpretation \models r$, iff for $t \leq t_{max} - \delay$, where for all $\groundedLiteral_i: \mu_i \in \textit{body}(r)$, 
if $\interpretation \satisfactionAtTime{t} \groundedLiteral_i:\mu_i$ then $\interpretation \satisfactionAtTime{t + \delay} \textit{head}(r)$. 

We say that $\interpretation$ satisfies a non-ground literal or rule iff $\interpretation$ satisfies all of its ground instances.
\end{definition}

A program in our logical language is made up of both TAFs and rules.

\begin{definition}[Generalized Annotated Program (GAP)]
GAP $\program$ is a set of Temporal annotated facts ($\program_{TAFs}$) and GAP rules ($\program_{Rules}$).  Given an Interpretation $\interpretation$ and a program $\program$, $\interpretation \satisfaction \program$ iff 
\begin{eqnarray*}
    \forall x \in \program, \interpretation \satisfaction x
\end{eqnarray*}
\end{definition}


The next definition captures the central concept of consistency; intuitively, 
if a Generalized Annotated Program's rules and annotations yield stable, non-contradictory, logically sound outcomes during its evaluation, it is considered to be consistent.

\begin{definition}[Consistency]
    A GAP $\program$ is \emph{consistent} if there exists some $\interpretation$ that satisfies all the rules in $\program$.
\end{definition}

A GAP is said to entail a TAF if it logically implies or derives the TAF from its rules and annotations under its semantic framework.

\begin{definition}[Entailment]
    We say GAP $\program$ \emph{entails} TAF $a:(\mu, t)$, denoted $\program \models_\mathit{ent} a:(\mu, t)$, iff for every interpretation $\interpretation$ s.t.\ $\interpretation \satisfaction \program$ we have that $\interpretation \satisfaction_t a:\mu$.
\end{definition}

A model of the GAP that assigns annotations to TAFs in a way that logically satisfies the program's rules and is minimal with respect to the lattice ordering of annotations is called the minimal model of the program. Minimal models thus represent the most ``conservative'' solutions (having the tightest annotation bounds) consistent with the program.

\begin{definition}[Minimal model]
    Given program $\program$, the \textit{minimal model} of $\program$ is an interpretation $\interpretation$ s.t.\ $\interpretation \satisfaction \program$ and for all interpretation $\interpretation'$ s.t.\ $\interpretation' \satisfaction \program$,
we have that $\interpretation' \preceq \interpretation$.
\end{definition}

As shown by \cite{ks92}, we can associate a fixpoint operator with any GAP $\program$ that maps interpretations to interpretations. In~\cite{mancalog13,shakarian2022extensions}, the authors present a fixpoint operator for identifying the logical outcome of a logic program. Intuitively, this operator performs a simulation while recording changes. Under the assumption of consistency, this operator produces an exact result in polynomial time (see Theorems~3.2 and~3.4 in \cite{APTL}), and our implementation provides practical speedups and consistency checking while maintaining these guarantees. We next define this operator formally:

\begin{definition}[Fixpoint Operator]
Let $\program$ be a program and $\interpretation$ an interpretation. 
The fixpoint operator $\fixpointOperator$ is a mapping defined as follows:
\[
\fixpointOperator(I)(\groundedLiteral_0, t) = \mathbf{sup}(\textit{annoSet}_{\program,\interpretation}(\groundedLiteral_0, t)),
\]
where $\textit{annoSet}_{\program,\interpretation}(\groundedLiteral_0, t) = \{\interpretation(\groundedLiteral_0, t)\}\cup
\{\mu_0$ such that for all ground rules $r \in \program$, where 
$\textit{head}(r)=\groundedLiteral_0:\mu_0$, for all $\groundedLiteral_i: \mu_i \in \textit{body}(r)$, $\textit{delay}(r) \leq t$, and $\interpretation \satisfactionAtTime{{t-\textit{delay}(r)}} \groundedLiteral_i:\mu_i\}$. 
\end{definition}

Note that the operator maps \textit{all} timepoint-literal pairs to timepoint-literal pairs, essentially revising the entire sequence of timepoints at once.
This contrasts with approaches such as Markov Decision Processes, which model transitions to a new state at each timepoint, and 
allows for direct modeling of non-Markovian dynamics.

\begin{definition}[Iterative Applications of $\fixpointOperator$]
Given natural number $i > 0$, interpretation $\interpretation$, and program $\program$, we define multiple applications of fixpoint operator
$\Gamma$ as follows:
$\fixpointOperator^i(I) = \fixpointOperator(I)$ if $i=1$, and $\fixpointOperator^i(I) = \fixpointOperator(\fixpointOperator^{i-1}(I))$ otherwise.
\end{definition}

\begin{figure}[t]
  \centering
  \fbox{%
    \begin{varwidth}{\linewidth}
      \[
      \begin{aligned}
        \Pi_{simple} = \{~& b(X):[1,1] \xleftarrow[\delay=1]~~ a(X):[1,1], \\
                        & c(X):[1,1] \xleftarrow[\delay=0]~~ b(X):[1,1] \}
      \end{aligned}
      \]
    \end{varwidth}
  }
  \caption{A simple logic program $\Pi_{simple}$ used to illustrate the application of fixpoint operator in Example~\ref{ex:fpo}.}
  \Description{A simple logic program $\Pi_{simple}$ used to illustrate the application of fixpoint operator in Example~\ref{ex:fpo}.}
  \label{fig:fixpoint-program}
\end{figure}

\begin{example}[Fixpoint Operator]
\label{ex:fpo}

Consider the simple program $\Pi_{simple}$ as shown in Figure~\ref{fig:fixpoint-program}.
If the fixpoint operator is applied twice, notice how the interpretations change in Table~\ref{tab:fpo-example}. We assume a grounding of $X=x$ for this example. Further applications of the fixpoint applications do not lead to any further change in the set of interpretations. In such a scenario, we say that the fixpoint operation has converged.

\begin{table}[t]
\caption{Evolution of interpretations as $\fixpointOperator$ is applied to $\Pi_{simple}$ in Figure~\ref{fig:fixpoint-program}}
\label{tab:fpo-example}
\begin{tabular}{clll}
\toprule
$t$ & $I$ & $\fixpointOperator(I)$ & $\fixpointOperator^2(I)$ \\
\midrule
1 & \texttt{a}(x):[1,1] & \texttt{a}(x):[1,1] & \texttt{a}(x):[1,1] \\
2 & & \texttt{b}(x):[1,1] & \texttt{b}(x):[1,1], \texttt{c}(x):[1,1] \\
3 & \texttt{a}(x):[1,1] & \texttt{a}(x):[1,1] & \texttt{a}(x):[1,1]\\
4 & & \texttt{b}(x):[1,1] & \texttt{b}(x):[1,1], \texttt{c}(x):[1,1] \\
\bottomrule
\end{tabular}
\end{table}

\end{example}

Finally, \logic\ additionally supports fuzzy logic~\cite{hohle1978probabilistic, alsina1983some, vojt01} and other non-classical approaches by enabling arbitrary functions that can be used over real values or intervals of reals. 
This provides a key advantage to reasoning about constructs learned with neuro-symbolic approaches such as those explored in works such as~\cite{lnn2020,deepMinIlp2018,nttIlp21,lnnInduction22,shakarian2022extensions}.

\section{Theoretical Analysis}
\label{sec:formal}

This section provides a formal analysis of our framework, beginning with a detailed examination of the theoretical correctness of the fixpoint operator when applied to Generalized Annotated Programs (GAPs) extended with temporal constructs. Following this, we investigate the theoretical aspects of performance, emphasizing how the incorporation of a lower semi-lattice structure facilitates dynamic, on-demand creation of symbols via Skolemization, thereby enabling efficient reasoning in domains with potentially infinite constants. Together, these results establish both the soundness of our formalism and its practical scalability, laying a foundation for subsequent experimental evaluation.

\subsection{Theoretical Results on Correctness\protect\footnote{A version of the results in Section~\ref{sec:formal-correct}  presented in an earlier conference paper from the authors in~\cite{shakarian2022extensions}; however, here we expand on them to include GAPs with temporal structures.}}

\label{sec:formal-correct}

We first show the fundamental properties of the fixpoint operator when a GAP is consistent.

\begin{theorem}
\label{thm:lgap_fp}
If GAP $\Pi$ is consistent, then:
\begin{enumerate}
\item\label{lowerLatticeThmPt1} $\fixpointOperator$ is monotonic,
\item\label{lowerLatticeThmPt2} $\fixpointOperator$ has a least fixpoint $\mathit{lfp}(\fixpointOperator)$, and
\item\label{lowerLatticeThmPt3} $\Pi$ entails TAF $a:\mu$ iff $\mu\leq \mathit{lfp}(\fixpointOperator)(a)$.
\end{enumerate}
\end{theorem}
\begin{proof}
(\ref{lowerLatticeThmPt1} and \ref{lowerLatticeThmPt2}) By creating an interpretation that maps atoms to annotations for time t $\in \{t_1,\ldots,t_{max}\}$, the monotonicity of $\fixpointOperator$ is trivial even in the case where $\Pi$ is inconsistent and it also has a least fixpoint by definition of the $\fixpointOperator$ operator.

(\ref{lowerLatticeThmPt3}) Suppose BWOC that $\Pi$ entails $a:\mu$ and $\mu> \mathit{lfp}(\fixpointOperator)(a)$.  However, this would imply there is a series of logical constructs that allow us to derive $a:\mu$ at some time t, and this would trivially be reflected in the iterative applications of the $\fixpointOperator$ operator.  Going the other way, BWOC if $\mu\leq \mathit{lfp}(\fixpointOperator)(a)$ but $\Pi$ does not entail $a:\mu$ would imply that there is no application of the constructs in $\Pi$ that lead to the deductive conclusion of $a:\mu$ at any time t; however this is again contradicted by the fact that $\fixpointOperator$ directly leverages the elements of $\Pi$.
\end{proof}

We can also show that for GAPs, we can bound the number of applications of $\fixpointOperator$ until convergence. This means that the computation process is guaranteed to terminate after a finite number of steps, establishes that the semantics are well-defined, and allows for effective inconsistency checking. This is crucial for practical use because it ensures that inference based on these programs completes in a predictable and finite time.

\begin{theorem}
\label{thm:appsOfT}
If GAP $\Pi$ is consistent, then $\mathit{lfp}(\fixpointOperator) \equiv \fixpointOperator^x$ where $x = height(\semiLattice)*|\mathcal{A}|*t_{max}$.
\end{theorem}
\begin{proof}
We know, by the definition of $ \fixpointOperator$, for any $i \leq x$, that for all $a \in \mathcal{A}$, $\fixpointOperator^{i(a)}\sqsubseteq  \fixpointOperator^{x(a)}$.  Hence, we just need to consider the case where $i > x$ and $\mathit{lfp}(\fixpointOperator) \equiv \fixpointOperator^i$ and $\mathit{lfp}(\fixpointOperator) \not\equiv \fixpointOperator^x$.  However, at each iteration the annotation of at least one literal must change.  The bound on the number of changes in annotation is $height(\semiLattice)$ (as the annotations must stay the same or increase monotonically, as $\Pi$ is consistent by the statement).  Hence, we have a contradiction.
\end{proof}

We can also leverage the $\fixpointOperator$ operator to identify inconsistencies. Since it tracks how logical statements are derived step-by-step, it can flag when new inferences contradict previously derived facts, revealing inconsistencies in the program's knowledge base or rules. Furthermore, it can provide an explainable trace of where and why inconsistencies occur, enabling users to pinpoint the exact rules or data causing the conflict.

\begin{theorem}
\label{them:incon}
GAP $\Pi$ is inconsistent if and only if for value $i$, and ground atom $a$, there exist $\mu, \mu' \in annoSet_{\Pi,\fixpointOperator ^ i}(a, t)$ where $\mu \not\sqsubseteq \mu'$ and $\mu' \not\sqsubseteq \mu$.
\end{theorem}
\begin{proof}

Claim 1: If there exist $i,a$ such that the statement holds, then $\Pi$ is inconsistent.  
Suppose, BWOC, that such an $i,a$ pair exist and $\Pi$ is consistent.  We know, by the definition of $\fixpointOperator$, that $\fixpointOperator(\fixpointOperator ^ i)$ must be an interpretation.  However, as there is no element above both $\mu,\mu'$, $\fixpointOperator$ that $\fixpointOperator(\fixpointOperator ^ i)$ cannot be a valid interpretation.

Claim 2: If $\Pi$ is inconsistent, then there exist $i,a$ such that the statement holds.  Suppose, BWOC, $\Pi$ is inconsistent and there does not exist such an $i,a$ pair.  Then, this implies that for all $a \in \mathcal{A}$ there exists some $i'$ where $\fixpointOperator(\fixpointOperator ^{i'}) = \fixpointOperator ^{i'}$ which means for any $i'' >i'$ at time t, we have $\fixpointOperator ^ {i'} = \fixpointOperator ^ {i''}$.  Therefore, by the definition of satisfaction, $\fixpointOperator ^ {i'}$ must satisfy $\Pi$, which is a contradiction.
\end{proof}

Note that this theorem does not depend on Theorem~\ref{thm:lgap_fp} (which has consistency as a requirement). Application of $\fixpointOperator$ can find an atom where the lower bound exceeds the upper bound (causing an inconsistency) if and only if $\Pi$ is inconsistent, and this will always happen within a finite, polynomial number of applications of $\fixpointOperator$. 
On the other hand, $\fixpointOperator$ is guaranteed to converge within a finite, polynomial number of applications if $\Pi$ is consistent. As long as $\fixpointOperator$ has not converged, we may not make any conclusion about the consistency of $\Pi$.

These results rigorously establish the correctness, convergence, and inconsistency detection properties of the fixpoint operator for GAPs with temporal extensions, ensuring sound and tractable reasoning. Building on this solid foundation, the next subsection examines how the underlying lattice structure enables the dynamic creation of atoms and constants, offering significant performance improvements. This analysis sets the stage for practical improvements in scalability and efficiency when reasoning over complex or infinite domains.

\subsection{Theoretical Results on Performance}
\label{sec:skolemization}

The use of the lower lattice structure enables the creation of atoms and constants in an ad-hoc manner.  In this section, we provide new formal arguments as to how such such ad-hoc symbol creation can provide significant performance improvements.  These results enable \logic~to reason about temporal relationships in settings with potentially infinite constants without sacrificing performance.

We begin by providing an example for a non-temporal use case designed to provide an intuition on the use of a lower lattice for ad-hoc symbol creation.
Figure~\ref{fig:intro-kg-program} shows a logic program for reasoning about changes in a knowledge graph, and Figure~\ref{fig:intro-kg-diagram} illustrates the knowledge graph before and after inference. 
This example is an excerpt from the YAGO03-10~\cite{yago} dataset, which we have used in our experiments in Section~\ref{sec:experiments}. 
The program has a single rule stating that ``\textit{X is a citizen of country Y if X is born in Z and Z is a city in Y}''. 
Figure~\ref{fig:intro-kg-diagram} shows an example knowledge graph with three constants: $ben, miami,$ and $usa$, which form two binary ground atoms: \texttt{bornIn$(ben, miami)$} and \texttt{cityIn$(miami, usa)$}, which are true. 
When the rule is grounded with $X = ben$, $Y = miami$, and $Z = usa$, the body is satisfied and the rule fires, creating a new binary ground atom $\texttt{citizenOf}(ben, usa)$ as shown in the figure. 

\begin{figure}[t]
  \centering
  \fbox{%
    \begin{varwidth}{\linewidth}
      \[
      \begin{aligned}
        \Pi_{kg} = \{& citizenOf(X,Y):[1,1] \leftarrow bornIn(X,Z):[1,1] \wedge cityIn(Z,Y):[1,1] \} \\
                    & \text{If X is born in city Z and Z is in country Y, then X is a citizen of Y.}
      \end{aligned}
      \]
    \end{varwidth}
  }
  \caption{Example of a logic program $\Pi_{kg}$ for the knowledge graph shown in Figure~\ref{fig:intro-kg-diagram}. English translation for each rule is provided.}
  \Description{Example of a logic program $\Pi_{kg}$ for the knowledge graph shown in Figure~\ref{fig:intro-kg-diagram}. English translation for each rule is provided.}
  \label{fig:intro-kg-program}
\end{figure}

\begin{figure}[t]
    \begin{center}
        \begin{subfigure}[t]{0.4\columnwidth}
            \vskip 0pt
            \includegraphics[width=\linewidth]{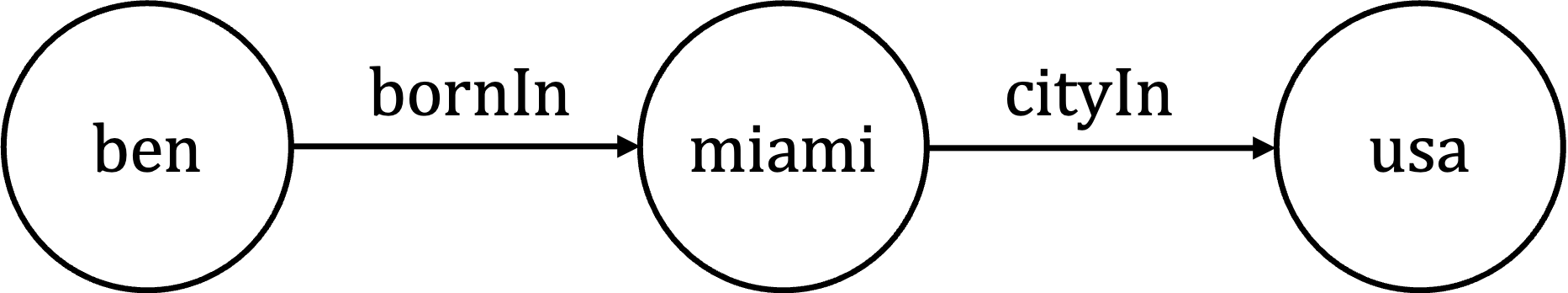}
        \end{subfigure}
        \hspace{15pt}
        \begin{subfigure}[t]{0.4\columnwidth}
            \vskip 0pt
            \includegraphics[width=\linewidth]{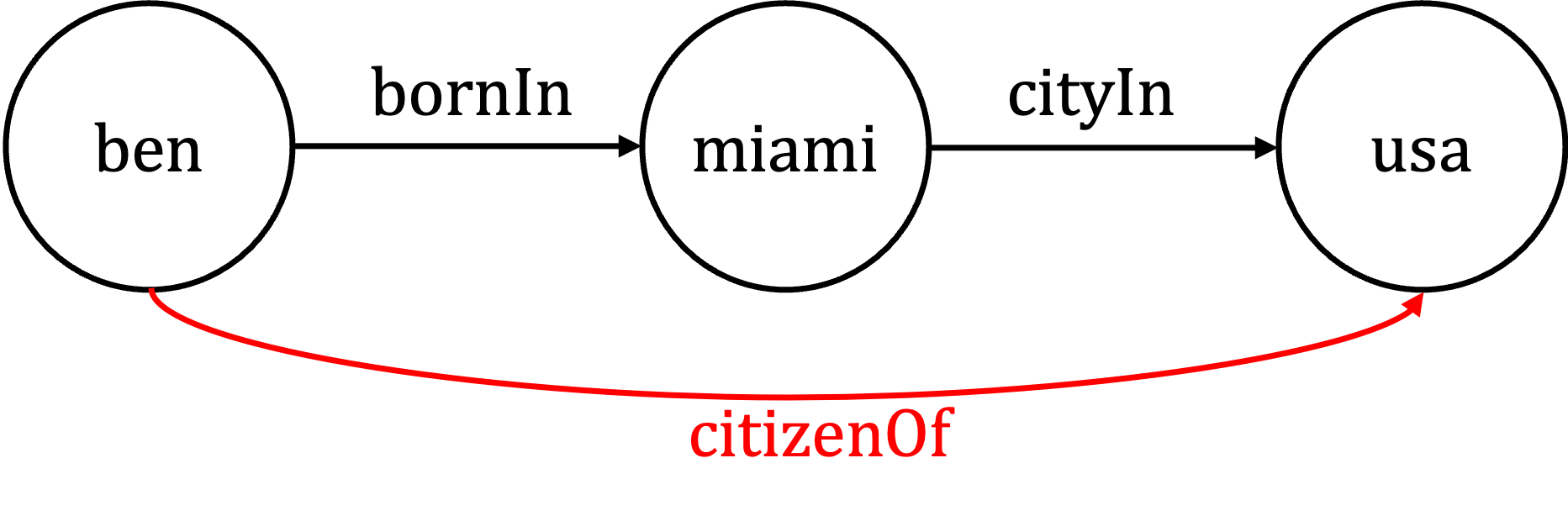}
        \end{subfigure}
    \end{center}
    \caption{Knowledge graph before (left) and after (right) inference. The newly created binary atom is shown in red.}
    \Description{Knowledge graph before (left) and after (right) inference. The newly created binary atom is shown in red.}
    \label{fig:intro-kg-diagram}
\end{figure}

Similarly, consider the geospatial example introduced in Section~\ref{sec:intro} with $\Pi_{geo}$, illustrated in Figure~\ref{fig:intro-geo-diagram}. When the first rule, which says, ``\textit{If fast-moving agent $A$ is at $L_1$ and moves left, it will be at $L_2$ (left of $L_1$) after one time unit}'' is fired, with the grounding $A = patrolCar$ and $L_1 = locMid$, a new constant $locLeft$ is created to the left of $locMid$ in space after one time point.
In the same manner, firing of the second rule creates another constant $locRight$ after two time points. These create new unary ground atoms: \texttt{location$(locLeft),$ location$(locRight)$}, and binary ground atoms: \texttt{at$(patrolCar, locLeft),$} \texttt{ at$(footPatrol, locRight),$} \texttt{ left$(locMid, locLeft),$ right$(locMid, locRight)$}.

This illustrates the creation of a new constant in space dynamically and on demand, and consequently corresponding unary and binary atoms are grounded. 
Note the creation of a binary ground atom, and that \textit{all} binary ground atoms in the language are implicitly assumed to exist even if not explicitly represented as facts.  This is because of our use of a lower-lattice structure (i.e., all atoms are originally assigned a truth value at the bottom of the lattice). 
As uncertainty regarding the truth of atoms is reduced with the progression of inference (which moves up the lattice), {\em we only need to represent in memory those atoms that are not assigned a truth value associated with the bottom element of the lattice.} 

In our open world formalism, anything that is not known to be true or false is uncertain---in practice, this means that everything that is uncertain (which could make up a vast portion of a practical KB) does not need to be allocated memory. 
In previous versions of our implementation, as well as other established software, the majority of symbols need to be grounded as the first step before carrying out inference. 
This limitation restricts the software to merely updating truth values or annotations to existing ground atoms based on available rules and facts at runtime. The increased memory consumption and computational overhead due to these factors are compounded due to the inherent sparsity of real-world knowledge bases and datasets. To mitigate these limitations, we introduce a Skolemization feature that allows our implementation to ground symbols only when certain rules are fired. This innovation eliminates the need for a complete grounding of symbols as a first step, allowing dynamic addition of new constants and ground atoms during inference. Consequently, we can now operate on incomplete knowledge bases while substantially reducing memory usage and running time. Empirical evidence supporting these improvements is presented in Section~\ref{sec:experiments}. 

We now examine the theoretical impact of the Skolemization feature; its empirical impact will be studied later.
In the following, we will refer to constant {\em types} as a way to split the set of constants $\constantSet$ for modeling purposes to best represent the different properties (such as attributes in a graph) of the objects being modeled.
Similarly, we can divide the set of variables $\variableSet$ and predicates $\predicateSet$ into multiple subsets based on the domain being modeled.
Finally, we assume a standard vector representation for each of these elements, and thus refer to their {\em components} in our analyses.
The first result establishes an upper bound on the number of possible ground atoms:

\begin{proposition}
Consider an environment with $t$ types of constants $c_1, c_2, \ldots, c_t$. 
If each constant of type $c_i$ has $m$ components $c_{i,1}, c_{i,2}, \ldots, c_{i,m}$ and the $j^{th}$ component $c_{i,j}$ can take $n_{i,j}$ values, then the maximum possible number of constants is:
    \begin{equation}
        \sum_{i=1}^{t} \prod_{j=1}^{m} n_{i,j}
    \end{equation}
If the number of attributes for constant of type $c_i$ is $a_i$, then the maximum possible number of ground atoms is:
    \begin{equation}
        \sum_{i=1}^{t} a_i \prod_{j=1}^{m} n_{i,j}
        \label{eq:ga-non-skol}
    \end{equation}
\end{proposition}
This result gives us two corollaries establishing bounds on the space required to store these atoms.
\begin{corollary}
Suppose that a constant is grounded in a vector with $m$ components, and each component can be represented using $b$ bits.
Then, the maximum possible number of constants is $2^{bm}$. 
Additionally, the amount of memory required to store these constants is $2^{bm}*(bm)$ bits.
\end{corollary}

\begin{corollary}
Let $n_{hv}$ denote the upper bound on predicate arity, and $|P|$ denote the number of distinct predicates, where $P$ is the set of all predicates.
Then, the maximum number of ground atoms formed from $2^{bm}$ constants is $|P| * ((2^b)^m)^{n_{hv}}$. 
Additionally, if each ground atom takes $b_a$ bits of memory, the amount of memory required to store all ground atoms is 
$|P| * ((2^b)^m)^{n_{hv}} * b_a$ bits.
\end{corollary}

\noindent
Using these results, we now arrive at our main result regarding performance improvements; to improve readability, we
summarize the notation used in Table~\ref{tab:notation-perf-theorem}.

\begin{table}[t]
\caption{Notation used in Theorem~\ref{theorem:ga-theorem}.}
\label{tab:notation-perf-theorem}
\begin{tabular}{ll}
\toprule
$g_i$ & Set of ground atoms after $i$ applications of the fixpoint operator $\Gamma$. \\
$g_0$ & Initial set of ground atoms. \\
$P$ & Set of all possible predicates in the domain. \\
$g_i(p)$ & Subset of $g_i$ having predicate $p$, where $p \in P$. \\
$\textit{pred}(\textit{head}(r))$ & Predicate in the atom in the head of rule~$r$; $\textit{pred}(\textit{head}(r)) \in P$. \\
$\textit{pred}(\textit{body}(r), j)$ & Predicate in the $j^{\textit{th}}$ clause in body of rule~$r$; $\textit{pred}(\textit{body}(r), j) \in P$. \\
\bottomrule
\end{tabular}
\end{table}

\begin{theorem}
\label{theorem:ga-theorem}
Let $\program_{\textit{Rules}}$ be a GAP. 
The number of new ground atoms produced after the $i^{\textit{th}}$ application of the fixpoint operator cannot exceed:
    \begin{equation}
        \sum_{r \in \Pi_{\textit{rules}}} ~ \prod_j \big|g_{i-1}(\textit{pred}(\textit{body}(r), j))\big|
        \label{eq:ga-theorem-skol}
    \end{equation}
\end{theorem}
A proof sketch is provided here. The full proof can be found in Appendix~\ref{app:proofs}.

\begin{proof}[Proof Sketch]
Let $P_{\program} \subseteq P$ be the set of predicates containing only predicates present in the head of at least one rule in $\program_{\textit{Rules}}$.
\begin{equation}
    |g_i| = \sum_{p \in P_{\Pi}} |g_i(p)| + \sum_{p \notin P_{\Pi}} |g_i(p)| \nonumber \\   
\end{equation}
As application of fixpoint operators can only create or modify atoms having a predicate in the head of a rule,
\begin{equation}
    |g_i| = \sum_{p \in P_{\Pi}} |g_i(p)| + \sum_{p \notin P_{\Pi}} |g_0(p)| \label{eq:gi_main}
\end{equation}
Let $\fixpointOperator_{r}(g)$ denote the set of ground atoms produced when a single fixpoint operator is applied to a single rule $r$ with the set of ground atoms $g$. 
\begin{flalign}
    |g_i(p)| = & |g_{i-1}(p) \cup \bigcup_{r \in \Pi_{\textit{rules}}~\wedge~ \textit{pred}(\textit{head}(r))= p} \fixpointOperator_{r}(g_{i-1})| \notag \\
    = & |g_{i-1}(p)| + \textit{newF}_{p,i} \times \textit{uniqueF}_{p,i} \times \sum_{r \in \Pi_{\textit{rules}}~\wedge~ \textit{pred}(\textit{head}(r))= p} |\fixpointOperator_{r}(g_{i-1})| \label{eq:gip_main}
\end{flalign}
Here, $\textit{newF}_{p,i} \in [0,1]$ denotes the fraction of ground atoms produced, with predicate $p$ and at the $i^{\textit{th}}$ $\Gamma$ application, which did not exist after the $(i-1)^{\textit{th}}$ application. 
Similarly, $\textit{uniqueF}_{p,i} \in [0,1]$ is the fraction of ground atoms produced across rules, with predicate $p$ in the head, which are unique.
We now try to estimate the last term in Equation~\ref{eq:gip_main}
\begin{equation}
    |\fixpointOperator_{r}(g_{i-1})| = \textit{validF}_{r,i} \times \prod_j |g_{i-1}(\textit{pred}(\textit{body}(r), j))| \label{eq:gamma_r_main}
\end{equation}

Here, $\textit{validF}_{r,i} \in [0,1]$ denotes the fraction of valid groundings that leads to firing of non-ground rule $r$, within the cross-product of possible groundings for each body clause.

\noindent
From Equations~\ref{eq:gip_main} and~\ref{eq:gamma_r_main}, and considering the maximum value ($=1$) for all three fractions we get:
\begin{equation}
    \Delta |g_i(p)| = |g_i(p)| - |g_{i-1}(p)| \leq \sum_{\substack{r \in \Pi_{\textit{rules}} \\ \textit{pred}(\textit{head}(r))= p}} ~~~ \prod_j |g_{i-1}(\textit{pred}(\textit{body}(r), j))| \label{eq:delta_gip_limit}
\end{equation}

\noindent
Substituting Equation~\ref{eq:delta_gip_limit} into Equation~\ref{eq:gi_main} we obtain:
\begin{equation}
    |g_i| \leq \sum_{p \in P_{\Pi}} |g_{i-1}(p)| + \sum_{p \in P_{\Pi}} \sum_{\substack{r \in \Pi_{\textit{rules}} \\ \textit{pred}(\textit{head}(r))= p}} \prod_j |g_{i-1}(\textit{pred}(\textit{body}(r), j))| + \sum_{p \notin P_{\Pi}} |g_0(p)|
\end{equation}
We get $|g_{i-1}|$ by adding the first and last term on the right-hand side. The two summations can also be combined to give us our final expression:
\begin{equation}
    \Delta |g_i| = |g_i| - |g_{i-1}| \leq \sum_{r \in \Pi_{\textit{rules}}} ~ \prod_j |g_{i-1}(\textit{pred}(\textit{body}(r), j))| \label{eq:delta_gi_bound}
\end{equation}
\end{proof}

These results show the potential impact of reducing unnecessary groundings.
In Section~\ref{sec:experiments}, we experimentally evaluate the impact of Skolemization and how it compares to this theoretical analysis.

\section{Implementation}
\label{sec:implementation}
In this section, we elaborate on our implementation of \logic; specifically, we discuss our PyReason software and our efficient grounding process, and then provide details on using logic as a simulator, and how it can be interfaced readily with any reinforcement learning agent. 
Though earlier versions of the implementation were introduced in~\cite{pyreason2023, mukherji2024scalable}, since then there have been multiple improvements. PyReason offers a comprehensive and flexible framework for reasoning based on generalized annotated logic; it supports various extensions, including temporal, graphical, and uncertainty-related features, which allow to capture a wide range of logics, such as fuzzy, real-valued, interval, and temporal logics. 
The open source codebase, tutorials, application case studies, and various other materials can be found at \url{https://pyreason.syracuse.edu}.


\subsection{Knowledge Graphs as Data Structures Supporting Efficient Implementations}

Built on modern Python, PyReason is specifically designed to handle graph-based data structures efficiently, and to support scalable yet accurate reasoning. We allow graphical input via the convenient \textit{Graphml} format, making it compatible with data exported from popular graph databases like Neo4j and GraphML. The Python library Networkx is used to load and interact with the graph data. 

A knowledge graph can be modeled using first-order predicate logic, where entities correspond to constants, relations correspond to predicates, and facts (usually represented as \textit{(entity1, relation, entity2)} triples in popular datasets) correspond to ground predicates (predicates with only constants as arguments). Constants correspond to nodes in the graph, while edges correspond to pairs of constants. Our approach, consistent with related literature~\cite{deepMinIlp2018,thomasNsr21}, restricts predicate arity to~1 or~2.

Algorithm~\ref{alg:flow} illustrates the structure of our implementation. Our software stores interpretations in a nested dictionary. Each constant forms a node, and edges correspond to pairs of constants that appear together in at least one grounded atom. For computational efficiency and ease of use, our software allows specification of a range of time-points $T = {t_1, t_2, \ldots}$ instead of a single time-point $t$, for which an interpretation $I$ remains valid. To reduce memory requirements, only one set of interpretations (current) is stored at any point in time. Once the inference engine is initialized, the program enters the ``Main loop'', which includes all the necessary subroutines to cycle through each point in time. The fixpoint operator ($\fixpointOperator$) is applied iteratively until convergence at each time point, and after every application, subroutines for logical consistency checking, followed by inconsistency resolution (if required), are run. Our optimized rule grounding process, key to efficient application of $\fixpointOperator$ is detailed in Section~\ref{sec:impl-grounding}. In addition, the $persistent\_flag$ shown in Algorithm~\ref{alg:flow} provides users with the option of preserving annotations from inferences made at the previous time step (when the flag is set) or resetting them to the bottom element of the lattice, $[0,1]$, (when the flag is reset).

\begin{algorithm}
\caption{Open-World Annotated Temporal Logic implemented in PyReason}\label{alg:flow}
\begin{algorithmic}[1]
\Statex \textbf{\underline{Data Structures}}
\State Nested Dictionary $\bm{I}=[Node/Edge,[Predicate, [Lower, Upper, Static]]]$ to store current interpretations only.
\State List $\bm{L}=[(Node/Edge, Predicate, Lower, Upper, Static, at\_t)]$ to store facts and inferences, before it is used to update the dictionary.

\State List $\textbf{IPL}=[(Predicate_1, Predicate_2)]$ containing pairs of predicates that cannot hold simultaneously (the bounds must be pairwise complementary). In the propositional case, if one of the predicates is $true$, the other must be $false$. We call this  ``inconsistent predicate list (IPL)''.
\State List $\bm{E}=[(Node/Edge, Predicate)]$ containing a list of predicates that becomes inconsistent in the course of program execution.

\Algphase{Initialization}
\State Initialize $\bm{I}$ as follows:
\Statex \quad For each nodes/edges, use $type\_checking$ to initialize valid predicates only.
\Statex \quad All bounds are initialized to [0,1].
\State $\bm{L} \gets [~]$ 
\Statex Facts (including initial interpretations) are then copied into $\bm{L}$
\State $t \gets 0$
\State $\bm{E} \gets [~]$
\State Input: Number of diffusion time-steps $T$, Set of rules $\bm{R}$

\Algphase{Main loop}
\While{$t \leq T$}

\For{$i$ in $I$, if (persistent\_flag == false)}
    \State reset bounds to [0,1] \Comment{Annotations returned to bottom of the lattice.}
\EndFor

\State $update\_req \gets 0$
\For{$l$ in $\bm{L}$, where ($l(at\_t) == t$)}
    \If{check\_consistency($l \in \bm{L}$,$l \in  \bm{I}$)}
        \State update\_req += update\_interp($l \in \bm{L}$,$l \in  \bm{I}$)
    \Else
        \State resolve\_inconsistency($l \in  \bm{I}$)
        \If{$(l, l') \in \textbf{IPL}$, $\forall l'$}
            \State resolve\_inconsistency($l' \in  \bm{I}$)
        \EndIf
    \EndIf
\EndFor

\If{$update\_req$}
    \State Apply fixpoint operator($\fixpointOperator$) once. \Comment{Grounding process elaborated in Algorithm~\ref{alg:grounding}}
    \For{each resulting interpretation}
        \If{$Static$ is $false$ in $I$}
            \State Add to $\bm{L}$
        \EndIf
    \EndFor
    \State Go to line $14$. \Comment{Check if another $\fixpointOperator$ application is needed.}
\Else
    \State $t \gets t+1$.
\EndIf

\EndWhile
\end{algorithmic}
\end{algorithm}

The core of PyReason is its rule-based reasoning, which enables handling uncertainty, open-world novelty, non-ground rules, quantification, and other diverse requirements seamlessly. The system remains agnostic to the selection of t-norm\footnote{A t-norm is a kind of binary operation that is typically used in fuzzy logic; they can be interpreted as a generalization of conjunction in classical logic.}, providing flexibility in using different logical connectives. 
Our description of the world as a knowledge graph (KG) adds support to applications where a policy must be learned via reasoning over context-related KGs such as~\cite{XR4DRAMA23}. Additionally, recent progress in developing Knowledge Graphs for probabilistic reasoning, as demonstrated by studies such as~\cite{vidKG23, ontologyKG23, DocSemMap22}, highlights the potential role of our framework in a wide range of practical applications.

\subsubsection{Special Case: Static Ground Atoms}
\label{sec:static}

In analyzing real-world use cases, we identified domain- or application-specific facts that remain constant over time. For instance, in the geospatial example introduced in Section~\ref{sec:intro}, an agent’s top speed is an inherent property unaffected by time or inference. To model such cases in our logic, we introduce an optional $Static$ flag. If a ground atom is declared with $Static=True$, its annotation remains fixed at the specified value across all time points. Algorithm~\ref{alg:flow} illustrates how the Static flag is incorporated into the data structures and checked before annotation updates. Ground atoms can be set to \emph{static} through rules as well, preventing other rules from being able to update the same ground atom later.

\subsection{Rule Grounding and Use of Skolemization\label{sec:impl-grounding}}

Many implementations of exact reasoning in different computational logic-based approaches struggle with the inherent complexity of the variable grounding process. The rule grounding process in \logic~leverages a novel form of Skolemization enabled by the use of a lower lattice structure for annotations, which significantly reduces the creation of new ground atoms. This approach not only pursues tractable reasoning by bounding the size of groundings, but also facilitates scalable inference in complex temporal and non-Markovian environments. First, we optimize the grounding process through efficient constant search using a predicate hash map, and by using CPU parallelization methods detailed in Section~\ref{sec:software-hardware}. Then, by efficiently managing the introduction of new constants during grounding, Skolemization enhances both speed and memory usage, making it a key factor in the practical applicability of \logic~for large-scale logic programs.

Algorithm~\ref{alg:grounding} details the variable grounding process for a single non-ground rule in our implementation. After initializing lists for possible groundings, the variable dependency graph captures the inter-dependency between variables that co-occur in different clauses. This proves to be especially useful later in the process to rapidly trim down the combinations of possible groundings. This, in turn, significantly reduces the grounding search space by reducing the branching factor. First, we loop through all the body clauses to generate all possible candidate combinations of constants that simultaneously satisfy the complete set of annotations in the ground rule. After each clause is reviewed, the dependency graph is used to prune the list of candidates. Once this process is complete, if we have one or more satisfiable groundings, we apply the annotation function to compute the bounds for the ground atom in the head of the fired rule. We then check if the ground atom in the head of every ground rule fired exists in the graph and, if they do not, then the new constants are created at runtime. We then go back to the ``Main loop'' in Algorithm~\ref{alg:flow}, which subsequently updates the interpretations.

\begin{algorithm}
\caption{Grounding a Non‐Ground Rule\label{alg:grounding}}
\begin{algorithmic}[1]
\Require 
  A rule \(r\in \Pi\) with
  \(\mathrm{head}(r)\), \(\mathrm{body}(r)=\{c_1,\dots,c_m\}\),  
  thresholds \(\Theta=(\Theta_1,\dots,\Theta_m)\),  
  annotation function \(f_{\mathrm{ann}}\),  
  interpretation \(\mathcal{I}\),  
  node set \(V\), edge set \(E\),  
  and predicate‐to‐constant maps \(\mathrm{PredMap}\).  
\Ensure 
  Two lists of grounded instances:  
  \(\mathit{app\_nodes}\) if \(\mathrm{head}(r)\) is unary,  
  \(\mathit{app\_edges}\) if it is binary.
\Statex
\State Extract variables \(\mathcal{V}\) and (if binary) head‐edge pattern \((X_h,Y_h)\)
\State Initialize
  \(\mathit{groundings}[X]\gets\varnothing\) for each \(X\in\mathcal{V}\),
  \(\mathit{groundings\_e}[(X,Y)]\gets\varnothing\) for each \((X,Y)\) in binary clauses
\State Build variable dependency graph \(D\) from all binary clauses
\For{\(i=1\) to \(m\)}
  \State Let \(c_i\) be
    \(\bigl(X:p:[\ell,u]\bigr)\) if unary, or
    \(\bigl(X,Y:p:[\ell,u]\bigr)\) if binary
  \If{\(c_i\) is unary}
    \State 
      \(S \gets \begin{cases}
         V\cap\mathrm{PredMap}[p],&X\text{ unseen}\\
         \mathit{groundings}[X],&\text{otherwise}
       \end{cases}\)
    \State 
      \(Q \gets \{\,v\in S\mid \mathcal{I}(v,p)\sqsubseteq[\ell,u]\,\}\)
    \State \(\mathit{groundings}[X]\gets Q\)
    \State Prune any \(\mathit{groundings\_e}\) entries involving \(X\)
    \State \(\textbf{if }|Q|\) fails \(\Theta_i\) \textbf{then return} empty lists
  \ElsIf{\(c_i\) is binary}
    \State Determine candidate edges
    \[
      S \;\gets\;
      \begin{cases}
        E\cap\mathrm{PredMap}[p], & X,Y\text{ both unseen}\\
        \{(v,w)\mid v\in\mathit{groundings}[X],\,w\in\operatorname{Nbr}(v)\}, & Y\text{ unseen}\\
        \{(v,w)\mid w\in\mathit{groundings}[Y],\,v\in\operatorname{Rnbr}(w)\}, & X\text{ unseen}\\
        \mathit{groundings\_e}[(X,Y)], &\text{otherwise}
      \end{cases}
    \]
    \State 
      \(Q\gets\{(v,w)\in S\mid \mathcal{I}((v,w),p)\sqsubseteq[\ell,u]\}\)
    \State \(\mathit{groundings\_e}[(X,Y)]\gets Q\)
    \State Update
      \(\mathit{groundings}[X]\gets\{v\mid\exists\,w:(v,w)\in Q\},\)
      \(\mathit{groundings}[Y]\gets\{w\mid\exists\,v:(v,w)\in Q\}\)
    \State Propagate refinements along \(D\)
    \State \(\textbf{if }|Q|\) fails \(\Theta_i\) \textbf{then return} empty lists
  \EndIf
  \State Propagate any changed \(\mathit{groundings}\) via dependency‐graph pruning
\EndFor
\If{all clauses satisfied}
  \State Initialize \(\mathit{app\_nodes},\;\mathit{app\_edges}\)
  \ForAll{each tuple of groundings for head‐variables}
    \State Locally re‐refine and re‐check all \(\Theta_i\)
    \If{satisfied}
      \State Assemble annotation inputs from each clause’s matches
      \State Compute \([\ell',u']\gets f_{\mathrm{ann}}(\dots)\)
      \State Add any new constants or edges to \((V,E)\)
      \If{\(\mathrm{head}(r)\) unary}  
         append to \(\mathit{app\_nodes}\)
      \Else{}
         append to \(\mathit{app\_edges}\)
      \EndIf
    \EndIf
  \EndFor
  \State \Return \(\mathit{app\_nodes},\;\mathit{app\_edges}\)
\Else
  \State \Return empty lists
\EndIf
\end{algorithmic}
\end{algorithm}

\subsubsection{Special Case: Immediate Rules}
\label{sec:immediate}

In section~\ref{sec:prelim-immediate}, we introduced immediate rules. They are rules with no delay, which are applied at the same time point at which the body of the rule is satisfied. Immediate rules make the program search for new applicable rules whose clauses might now be satisfied because of the immediate rule.
Practically, rules with $\delay=0$, are rules with infinitesimally small delay. This allows cascading of several rules, which are all fired within the same time point. A simple example was shown in Example~\ref{ex:fpo}. A practical use case can be found in our experiment in Section~\ref{sec:rl-setup}, when the shooting action is brought into the picture because multiple events are occurring within a single time point, but they're all interconnected. We note that this is possible without any extensions to annotated logic as the temporal extensions we use (based on~\cite{mancalog13, APTL, pyreason2023}) have no requirement that two time units be uniformly separated in actual time.

\subsection{Logic as a Simulator for Reinforcement Learning}

Deep Reinforcement Learning (RL) algorithms typically require a simulator to learn an agent policy. However, traditional simulators have several drawbacks, such as speed and data efficiency, as well as lack of explainability, modularity, and extensibility without retraining. 
We introduce {\tt PyReason-gym}, an OpenAI Gymnasium wrapper that allows easy interfacing with a grid world that uses PyReason as the simulation and dynamics engine. In our experiment in Section~\ref{sec:rl-expt}, we show the applicability of our approach in a practical RL setting and compare it to two well-established simulation environments. By integrating PyReason as the underlying engine, the grid world dynamics and agent interactions are governed by logical rules, enabling precise and interpretable state transitions. Crucially, the temporal extensions and lower-lattice annotation of \logic~means {\tt PyReason-gym} can naturally interact with a simulator with non-Markovian dynamics, where agent behavior depends on multiple previous time steps rather than just the current state. {\tt PyReason-gym} also uses a core functionality of our implementation to support the generation of detailed, time-stamped event traces, facilitating explainability in agent behavior and environment interactions. We discuss how this is particularly useful and show an example of an explainable trace in Section~\ref{sec:explainability}. With efficient memory usage and configurable settings, {\tt PyReason-gym} offers a practical, scalable solution for reinforcement learning applications, which can work with most out-of-the-box RL algorithm packages widely available. We also provide {\tt PyReason-gym} as a fully open-source codebase, enabling anyone to easily reproduce and adapt it for their own application domains. 
We refer the interested reader to the PyReason website\footnote{\url{https://pyreason.syracuse.edu/}} for more details.

\subsection{Explainability\label{sec:explainability}}

In order to support interpretability and explainability, interpretations used in past computations can be obtained using \textit{rule traces}, which retain the change history for each interpretation and the corresponding grounded logical rules that caused each change. Such rule traces pave the way towards these goals, as every inference can be traced back to the cascade of rules that led to it. This enables users and developers to understand the exact reasoning path, identify and diagnose unexpected behaviors, and validate the correctness of complex temporal dependencies. Additionally, because \logic~supports non-Markovian dynamics, rule traces are essential for capturing how past states and delayed effects contribute to current inferences, making temporal reasoning transparent and interpretable.

As an illustration of this functionality, Table~\ref{tab:rule-trace-ex} shows the rule trace for Example~\ref{ex:fpo} where the fixpoint operator was applied twice to the program $\Pi_{simple}$. Note that to denote TAFs, the ``Rule fired'' column is intentionally left blank.

\begin{table}[tb]
\caption{Rule trace produced by the PyReason software when Example~\ref{ex:fpo} was executed.}
\label{tab:rule-trace-ex}
\begin{tabular}{ccllccl}
\toprule
t & $\Gamma$ & Constant Symbols & Predicate & Old Annotation & New Annotation & Rule fired \\
\midrule
1 & 0 & x & a & {[}0.0,1.0{]} & {[}1.0,1.0{]} & -- \\

\midrule
2 & 1 & x & b & {[}0.0,1.0{]} & {[}1.0,1.0{]} & $rule_1$ \\
2 & 2 & x & c & {[}0.0,1.0{]} & {[}1.0,1.0{]} & $rule_2$ \\

\midrule
3 & 0 & x & a & {[}0.0,1.0{]} & {[}1.0,1.0{]} & -- \\

\midrule
4 & 1 & x & b & {[}0.0,1.0{]} & {[}1.0,1.0{]} & $rule_1$ \\
4 & 2 & x & c & {[}0.0,1.0{]} & {[}1.0,1.0{]} & $rule_2$ \\
\bottomrule
\end{tabular}
\end{table}

\subsection{Software- and Hardware-based Performance Improvements}
\label{sec:software-hardware}

One of the key strengths of PyReason is its speed and machine-level optimized fixpoint-based deduction approach. This ensures efficient and scalable reasoning capabilities, even when dealing with large graphs with over 30 million edges. As stated before, the variable grounding process is one of the primary bottlenecks of exact reasoning implementations. Grounding non‐ground rules is inherently expensive: to instantiate a rule with $k$ variables over a domain of size $N$, the engine must consider up to $N^k$ candidate substitutions at each iteration, and even simple two‐variable bodies can require checking on the order of $N^2$ ground pairs. In a knowledge base with $N=10^5$ constants, a rule like  
\[
h(Y) \;\longleftarrow\;p(X)\wedge q(X,Y)
\]  
would naïvely generate $10^{10}$ combinations per fixpoint step—clearly prohibitive in pure Python without additional heuristics. To accelerate the grounding process and to streamline the whole inference engine, we make several software optimizations as well as leverage multi-core CPU processors for hardware optimizations.    

\subsubsection{Software Optimizations\label{sec:software}}

We list some of the design choices that leverage the properties of \logic~to optimize our approach.
\begin{enumerate}
  \item \textit{Uncertain Predicate Filtering:}  
    For each predicate $p$, maintain: 
    \[
      \mathit{PredFiltered}_p \;=\;\{\,v\in V \mid \interpretation(v,p)\neq[0,1]\},
    \]
    i.e.\ only those constants whose current annotation for $p$ is not completely uncertain.  When grounding a clause $p(X):[l,u]$, we intersect $V$ with $\mathit{PredFiltered}_p$, often reducing the domain by orders of magnitude.

    \item \textit{Predicate Maps:}  
    In addition to \(\mathit{PredFiltered}\), we maintain for each predicate \(p\) a map:
    \[
      \mathit{PredMap}_p = \{\,v\in V\mid \exists\,t:\;p(v):\mu_t\in\mathcal{I}\},
    \]
    and similarly for edges in \(E\).  When grounding \(p(X)\) or \(p(X,Y)\), we initialize the candidate set from \(\mathit{PredMap}_p\) rather than the entire \(V\) or \(E\).  This enforces an open‐world pruning based on known TAFs and avoids scanning the full graph.

\item \textit{Clause Order Optimization:}  
    Within each rule, we reorder the body so that all unary clauses appear before binary clauses.  Since \(\lvert V\rvert \ll \lvert E\rvert\) in typical sparse graphs, grounding the unary literals first shrinks the candidate \(groundings[X]\) sets, dramatically reducing the cost of the subsequent binary clause groundings \(q(X,Y)\).

  \item \textit{Early Threshold Checks:}  
    Immediately after grounding the $i^{th}$ body literal \(c_i\) and obtaining its candidate set \(Q_i\), we verify its counting or percentage threshold \(\Theta_i\).  Concretely, for a clause:  
    \[
      \bigl|\{\,Y \mid q(X,Y):[l,u]\}\bigr|\ge k,
    \]
    we compute \(\lvert Q_i\rvert\) once—if \(\lvert Q_i\rvert<k\), the entire rule cannot fire for this \(X\), so we abort grounding the remaining clauses and move on.  This saves the cost of further joins and refinements when a single clause already fails.
    
  \item \textit{Dependency‐Graph Pruning:}  
    We build a small dependency graph \(D\) over the rule’s logic variables; each binary clause \(q(X,Y)\) adds connections \(X\!-\!Y\).  
    Whenever a clause refines \(groundings[X]\), we propagate that change along \(D\) to remove any now‐invalid values from neighboring \(groundings[Y]\) sets (and vice versa). This reduces the search spaces for future clauses and is able to stop the grounding process early if previous clauses no longer have their thresholds satisfied.
\end{enumerate}

\subsubsection{Hardware Accelerations\label{sec:hardware}}

Grounding and fixpoint iteration are governed by a few core compute kernels—rule‑grounding loops, lattice‑join updates, and annotation propagation—that repeatedly traverse large portions of the graph representation. In naïve Python these loops suffer both interpreter overhead and poor memory access patterns, rapidly becoming the bottleneck as the number of rules or graph size grows. To overcome this, we compile and cache all of these critical loops with Numba’s Python \textit{@njit(parallel=True)} decorator, yielding LLVM‐generated \cite{lattnerllvm} native code.

Furthermore, each fixpoint iteration naturally decomposes into independent tasks—one per rule—since grounding and annotation for rule \(r_i\) only reads the ``old'' interpretation and writes its own proposed updates.  We exploit this by replacing standard \texttt{for} loops over rules with Numba’s \textit{prange} (parallel range function), which:
\begin{enumerate}
  \item \textit{Distributes rules across cores:}  Each CPU thread compiles and executes a disjoint chunk of the rule list, grounding and computing annotation updates in parallel.
  \item \textit{Minimizes synchronization:}  Threads accumulate their local updates in thread‐local buffers. At the end of the iteration a single, lightweight reduction step merges each thread into the global interpretation with only~$O(R)$ overhead (where \(R\) is the number of rules).
  \item \textit{Adaptively balances load:}  Numba’s runtime uses dynamic scheduling to hand off new rules to idle threads, smoothing out any variability in grounding cost between simple and complex rules.
\end{enumerate}

The combination of targeted software optimizations and efficient hardware-level parallelization significantly accelerates PyReason’s grounding and fixpoint computation processes. This integrated approach enables scalable, high-performance reasoning even on massive graphs that would otherwise be computationally prohibitive.

\section{Experimental Evaluation}
\label{sec:experiments}

In this section, we present the results of our empirical evaluation, which complement the theoretical analysis from Section~\ref{sec:skolemization}, and then go on to explore the computational benefits of leveraging Skolemization.
We begin with an application in a geospatial domain, showcasing the creation of new constants and the resulting scalability benefits. Next, we explore Knowledge Graph completion tasks on popular datasets to illustrate the utility and scalability of our approach in knowledge extraction tasks.
Finally, we extend our evaluation to Reinforcement Learning (RL) scenarios, where we demonstrate the scalability, portability, and explainability of our logic-based simulator in PyReason for complex game environments. These RL experiments also highlight the benefits of incorporating non-Markovian dynamics—which capture dependencies beyond the current state—and logic shielding—which provides safety constraints or policy guidance—to improve learning in challenging settings.

\subsection{Creation of Constants in a Geospatial Application}
\label{sec:geospatial}

This experimental setup demonstrates the efficacy of our Skolemization technique in dynamic geospatial environments, highlighting its computational and memory advantages.
The experiment involves a series of geospatial areas (maps) defined over spaces with varying granularity, ranging from resolution~2 (two levels of nested quadrants equaling~16 points) to resolution~9. Effectively, we may visualize a fully grounded out map as a grid. 
These resolutions define the set of possible constants.
The logic program in each experiment is designed to model a two-team game scenario, where agents navigate the map according to user-provided directions, constrained by grid space and agent type. Each team has two types of agents in this game, and each type has unique movement restrictions: 
(a)~{\em Border agents}: Limited to edges of the area, with one step per time unit. 
(b)~{\em Field agents}: Unrestricted movement, capable of two steps per time unit.
Leveraging the logic's capability of handling non-Markovian properties, two differing speeds are used, which is shown to be relevant in practical deployments.  Each team's agents begin at a corner, representing the ``minimum'' scenario for ground atoms in the Skolemization approach.

Leveraging PyReason as the inference engine, the experiment compares the two approaches. 
Key distinctions in initial graph construction include:
\begin{enumerate}
    \item {\em Non-Skolemization}: All spatial points and immediate neighbor edges are defined as constant symbols (nodes), resulting in a graph size proportional to map resolution.
    \item {\em Skolemization}: Initially grounds only points that are occupied by agents, plus their immediate neighbors, dynamically grounding new constants during reasoning based on movements.
\end{enumerate}
The Skolemization graph size correlates with agent positions rather than map resolution, potentially offering significant computational advantages in sparse environments.

\subsubsection{Reduction in Groundings with Skolemization: Creation of New Constants in a Geospatial Domain}

\begin{figure*}[t]
    \begin{center}
        \begin{subfigure}{0.45\linewidth}
            \begin{center}
                \includegraphics[width=\linewidth, trim=2 3 2 2, clip]{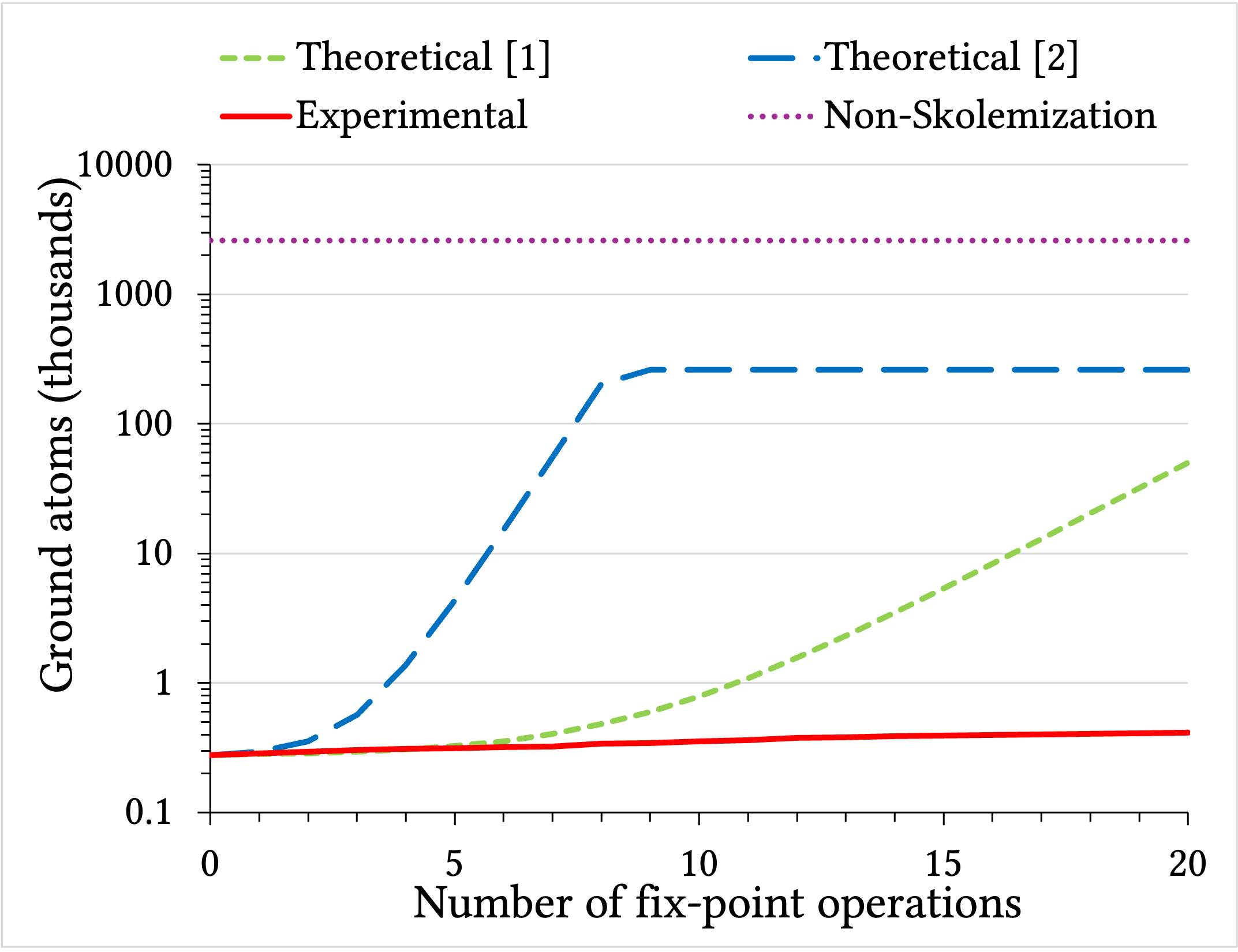}
            \end{center}
        \end{subfigure}
        \begin{subfigure}{0.45\linewidth}
            \begin{center}
                \includegraphics[width=\linewidth, trim=2 3 2 2, clip]{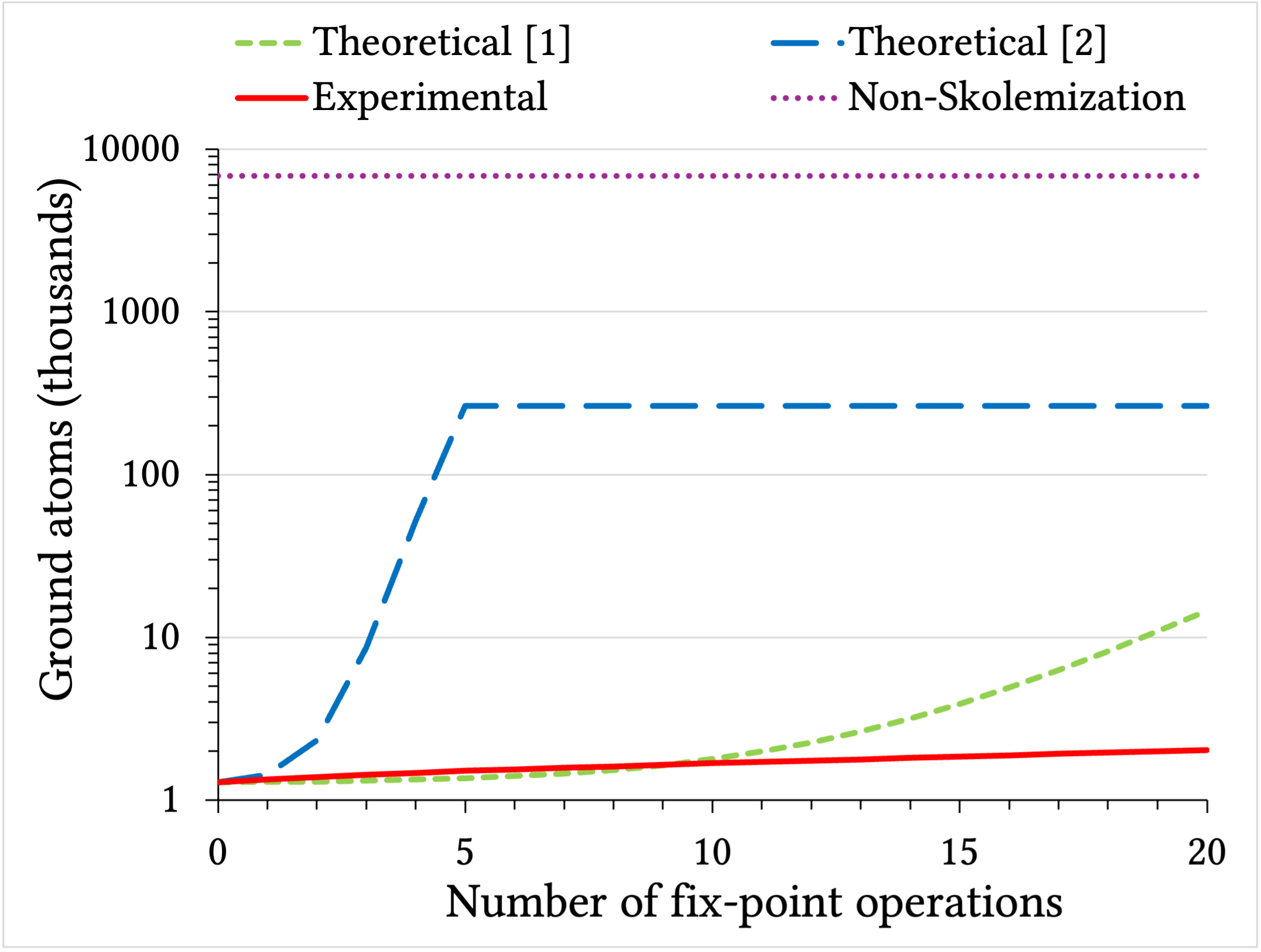}
            \end{center}
        \end{subfigure}
        \begin{subfigure}{0.45\linewidth}
            \begin{center}
                \includegraphics[width=\linewidth, trim=2 3 3 43, clip]{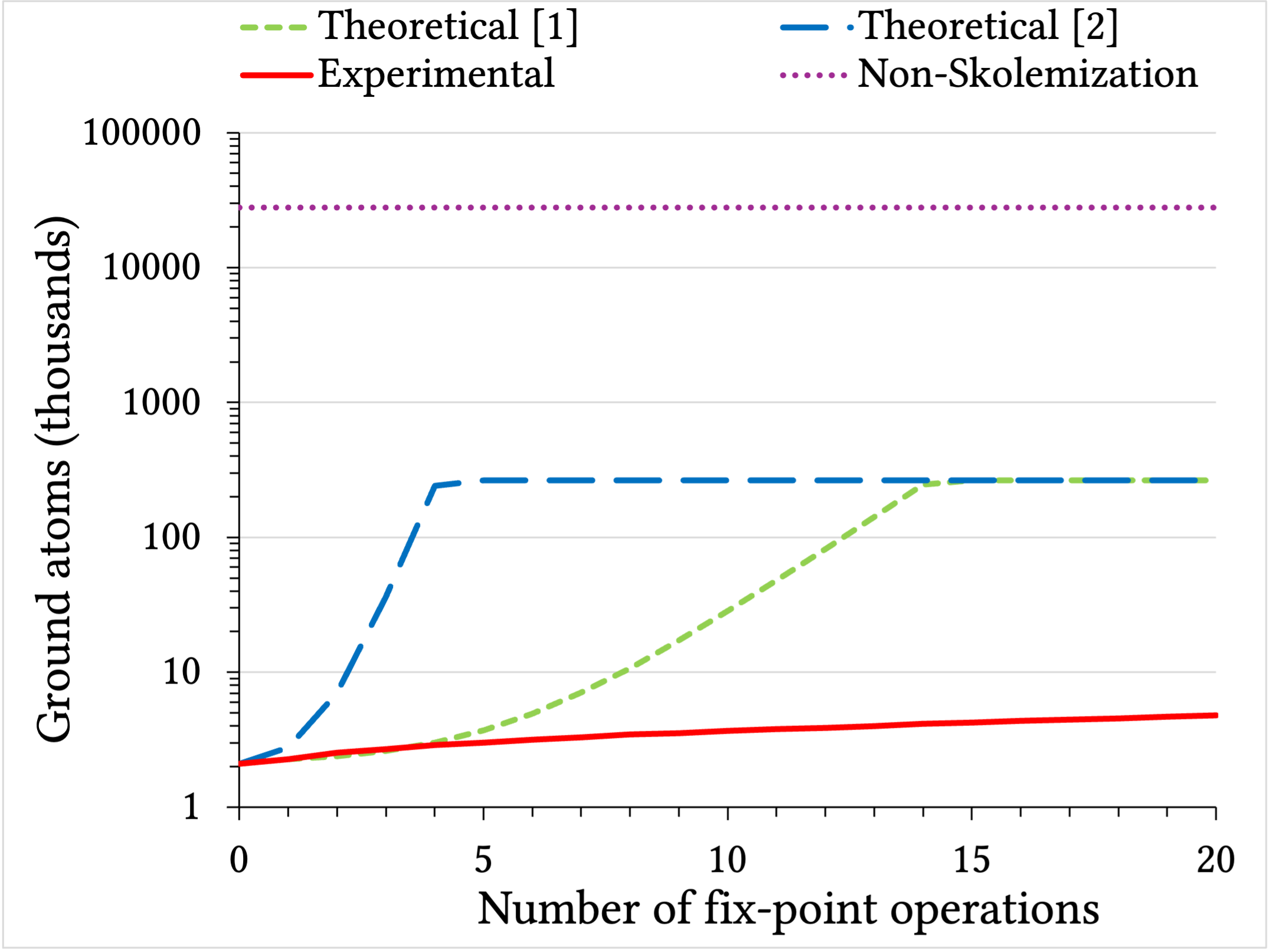}
            \end{center}
        \end{subfigure}
        \begin{subfigure}{0.45\linewidth}
            \begin{center}
                \includegraphics[width=\linewidth, trim=2 3 2 45, clip]{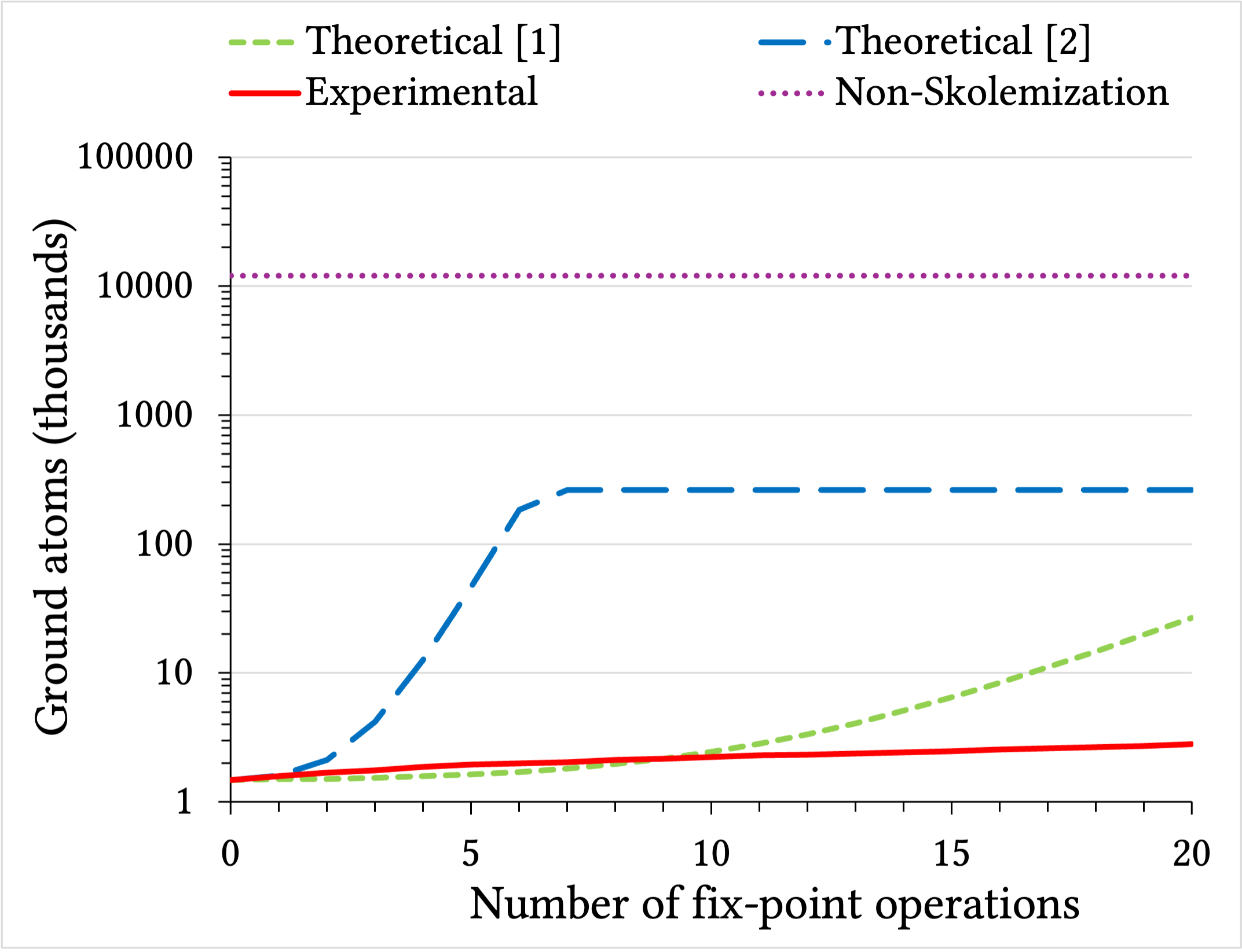}
            \end{center}
        \end{subfigure}
    \end{center}
    \caption{(Clockwise from top-left) Comparison of \#ground atoms with multiple applications of fixpoint operator ($\Pi_4, \Pi_{20}, \Pi_{40}, \Pi_{100}$).}
    \Description{Clockwise from top-left) Comparison of \#ground atoms with multiple applications of fixpoint operator ($\Pi_4, \Pi_{20}, \Pi_{40}, \Pi_{100}$).}
    \label{fig:ga-geo}
\end{figure*}

In Section~\ref{sec:skolemization}, we derived expressions for the maximum number of ground atoms with and without Skolemization (Eqs.~\ref{eq:ga-theorem-skol} and~\ref{eq:ga-non-skol}, respectively). 
Now, we experimentally aim to study how these expressions compare with empirical observations in our aforementioned use cases.
 
We create four programs, one each for when each of the two teams has~2~($\Pi_4$), 10~($\Pi_{20}$), 20~($\Pi_{40}$), and 50~($\Pi_{100}$) agents.  
We can think of minimal model computation via the fixpoint operator on these programs as essentially simulating agent movements. 
In our experiments, we allow for twenty time steps (note that time steps are the amount of time we represent, and are not necessarily related to the number of fixpoint operations).
The fully-grounded version of the graph contains 262,144 geographic constants (represented by nodes in the graph), and the largest graph (for 100 agents) is comprised of more than 27.7 million ground atoms. In comparison, the graph before simulation for the Skolemization approach only contains 2,084 ground atoms.
Results of the simulations are shown in Figure~\ref{fig:ga-geo} (note that the $y$-axes are log scale). 
We choose two parameter values for each setting to plot the theoretical bounds that closely mirror experimental values for the first reasoning step. 
We note that the theoretical bound is generally tight for lower numbers of inference steps (fixpoint operations), and the tightness of the bound varies based on the parameter.
Further, we observe that, even when theoretical results converge after certain number of fixpoint operations, it is up to a few orders of magnitude below the number of ground atoms for the Non-Skolemization case. All theoretical values converged before~50 fixpoint applications when run to convergence.

\subsubsection{Scalability: Scaling with Ground Atoms}

\begin{figure}[t]
    \begin{center}
        \begin{subfigure}{0.45\columnwidth}
            \begin{center}
                \includegraphics[width=\linewidth, trim=3 3 3 3, clip]{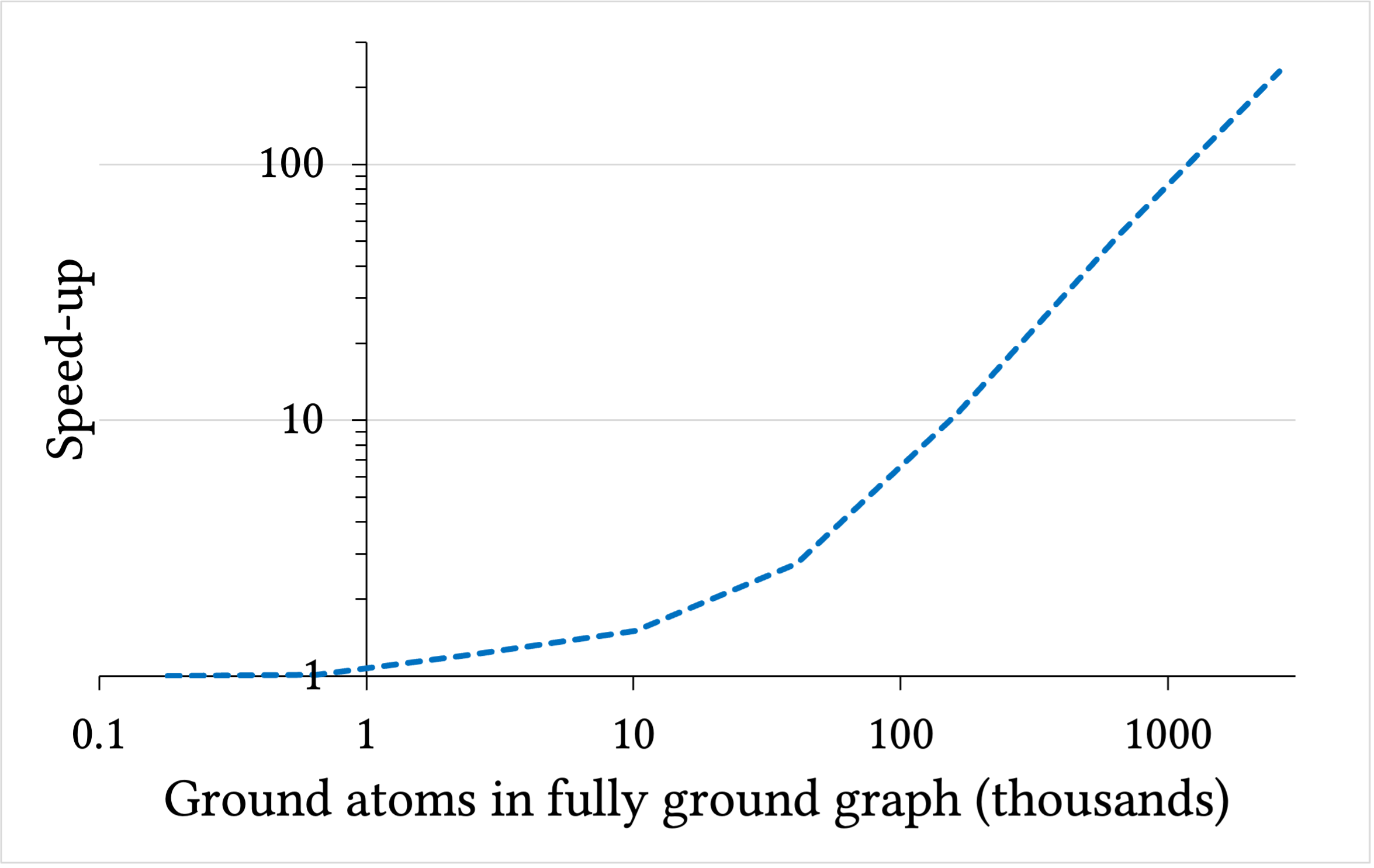}
            \end{center}
        \end{subfigure}
        \begin{subfigure}{0.45\columnwidth}
            \begin{center}
                \includegraphics[width=\linewidth, trim=3 3 3 3, clip]{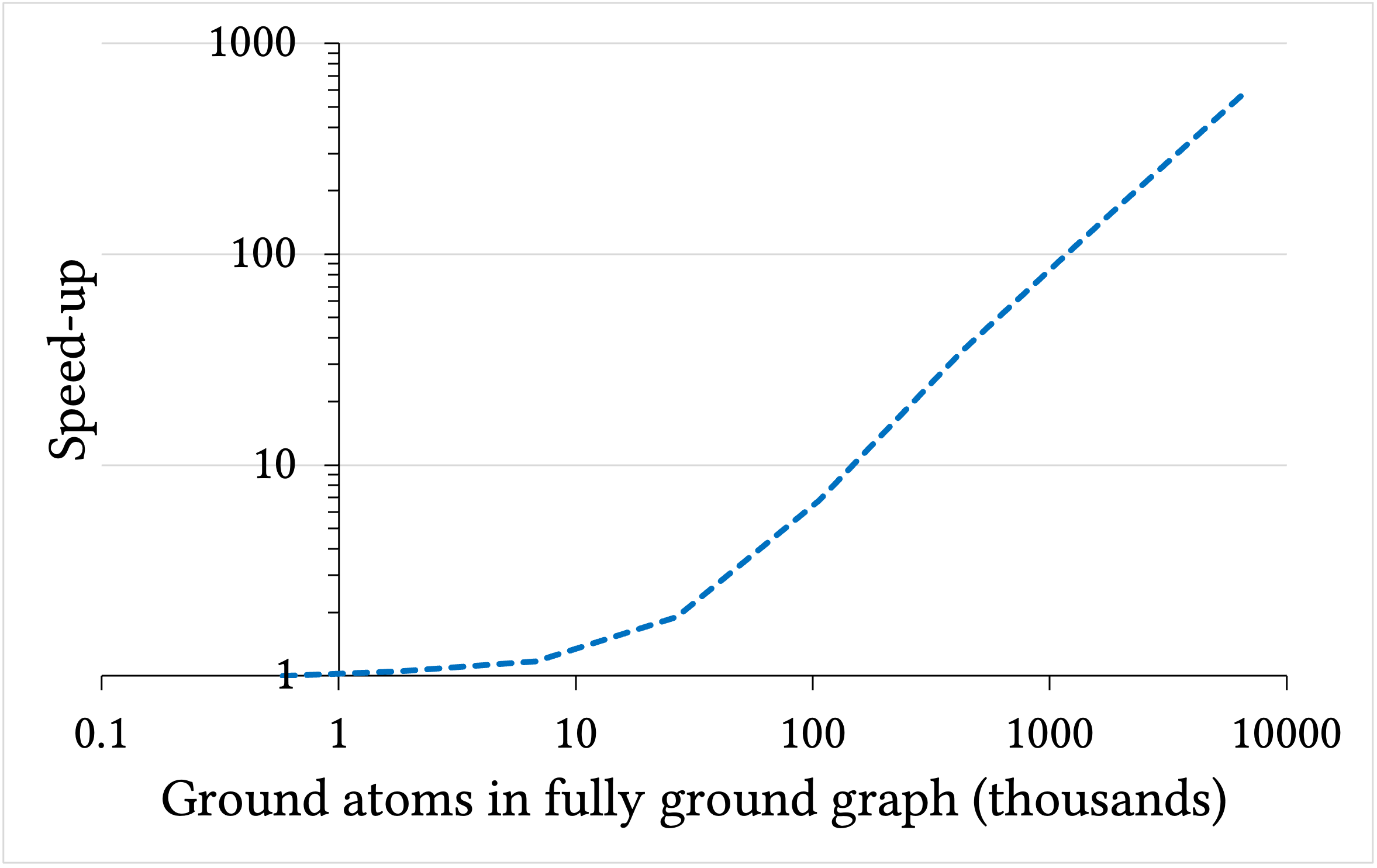}
            \end{center}
        \end{subfigure}
        \begin{subfigure}{0.45\columnwidth}
            \begin{center}
                \includegraphics[width=\linewidth, trim=3 3 3 3, clip]{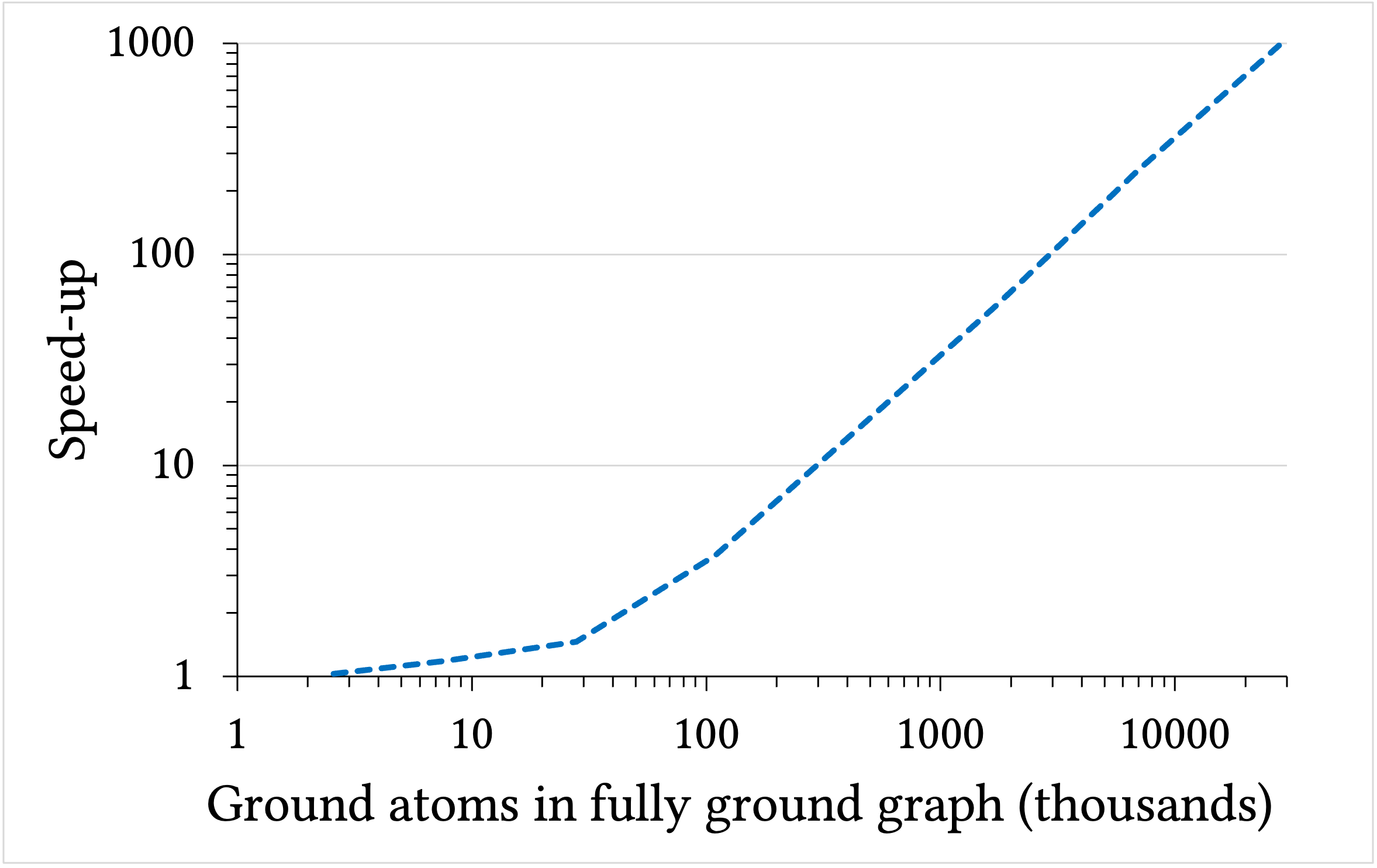}
            \end{center}
        \end{subfigure}
        \begin{subfigure}{0.45\columnwidth}
            \begin{center}
                \includegraphics[width=\linewidth, trim=3 3 3 3, clip]{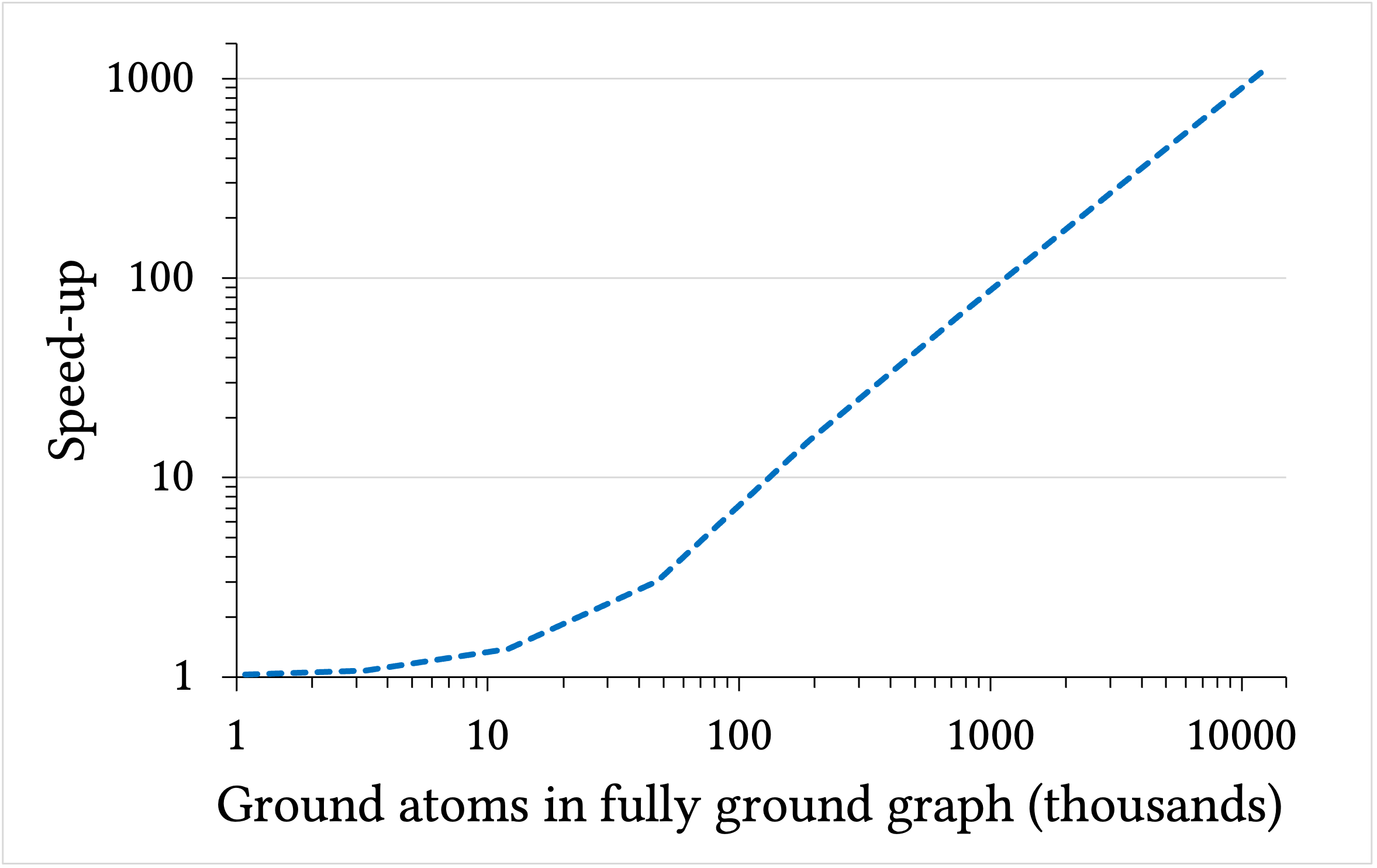}
            \end{center}
        \end{subfigure}
    \end{center}
    \caption{(Clockwise from top-left) Speed-up vs Map size for $\Pi_4$, $\Pi_{20}$, $\Pi_{40}$, $\Pi_{100}$.}
    \label{fig:speedup-geo}
\end{figure}

\begin{figure}[tb]
    \begin{center}
        \begin{subfigure}{0.45\columnwidth}
            \begin{center}
                \includegraphics[width=\linewidth, trim=3 3 3 3, clip]{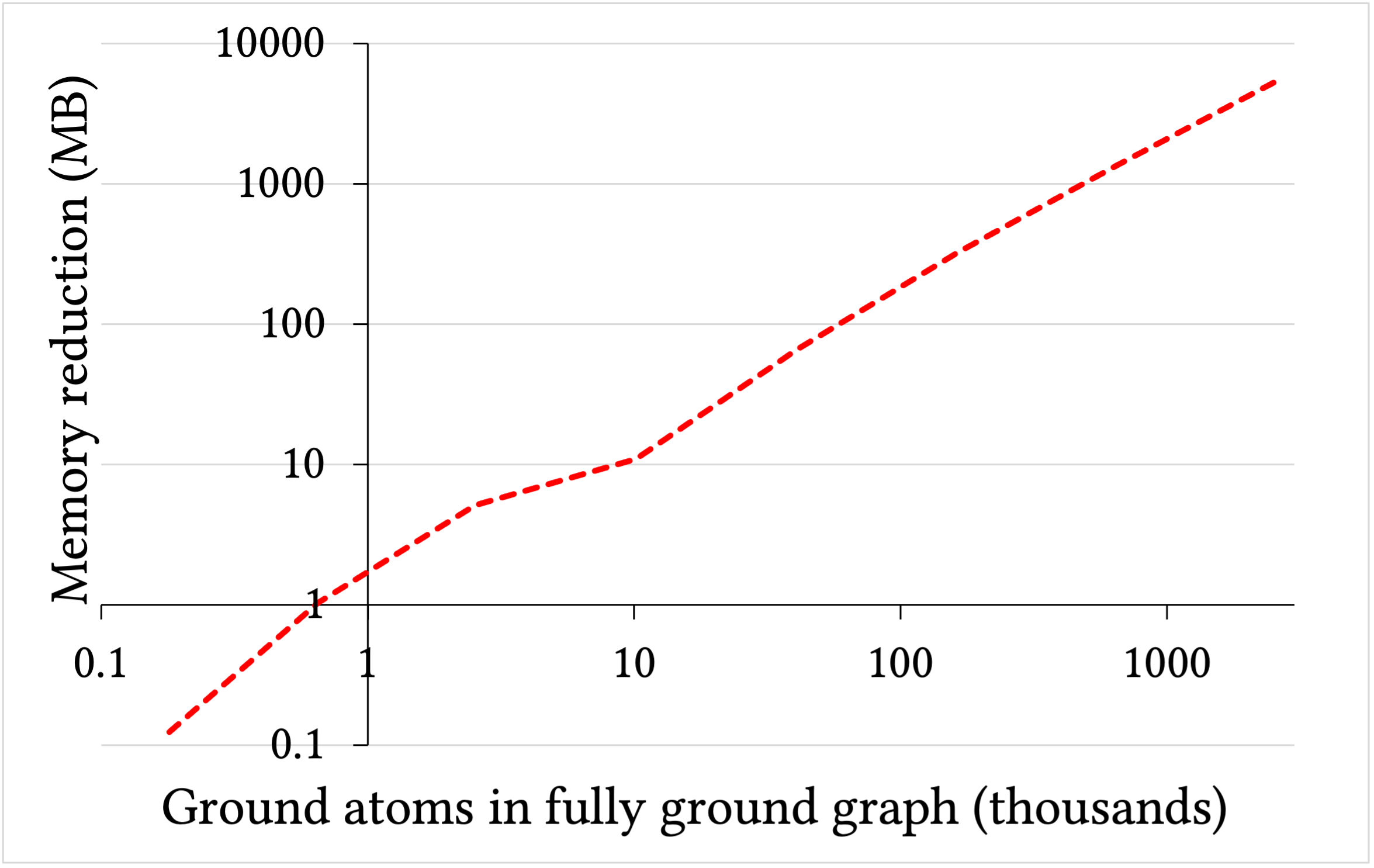}
            \end{center}
        \end{subfigure}
        \begin{subfigure}{0.45\columnwidth}
            \begin{center}
                \includegraphics[width=\linewidth, trim=3 3 4 3, clip]{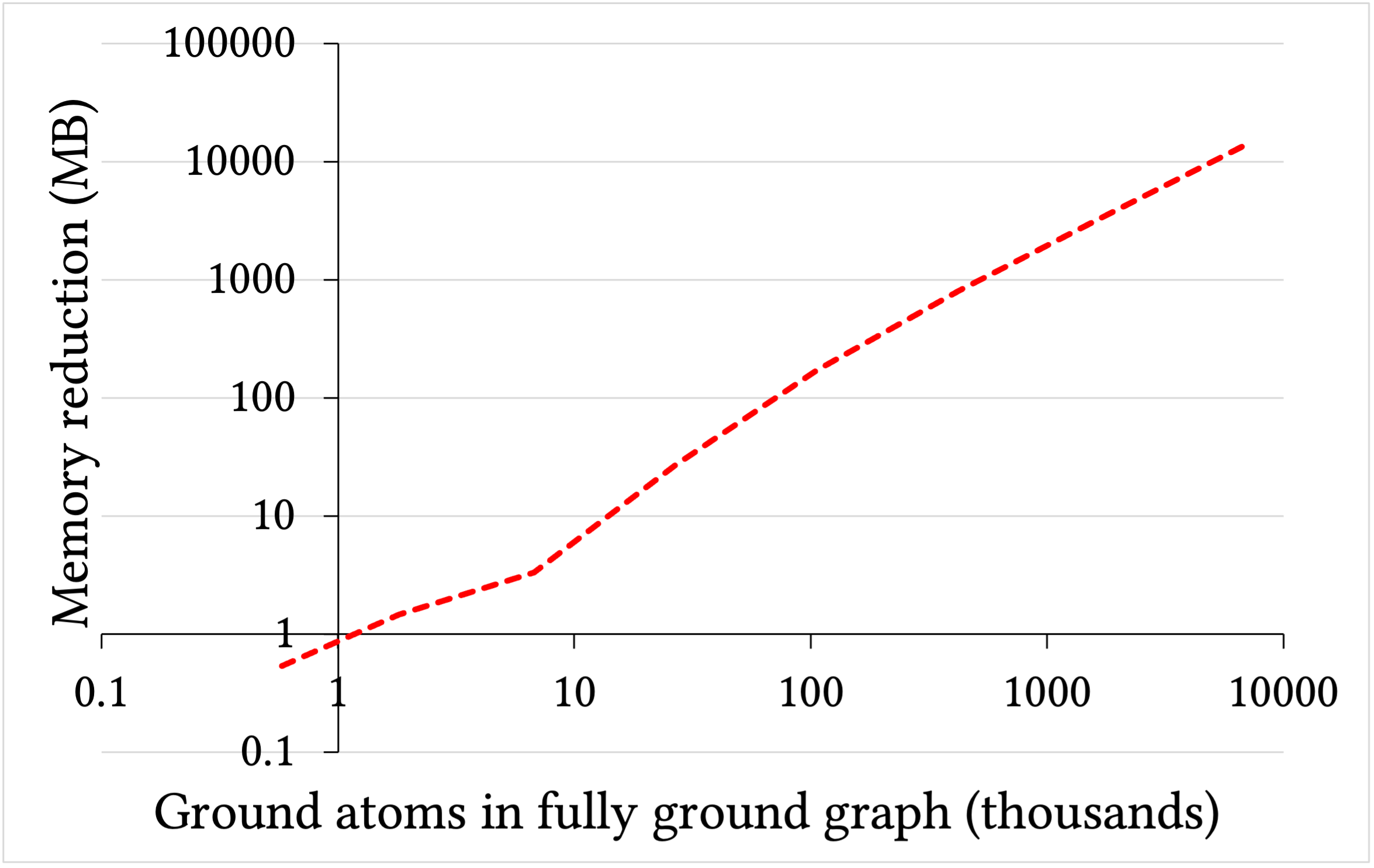}
            \end{center}
        \end{subfigure}
        \begin{subfigure}{0.45\columnwidth}
            \begin{center}
                \includegraphics[width=\linewidth, trim=3 3 3 3, clip]{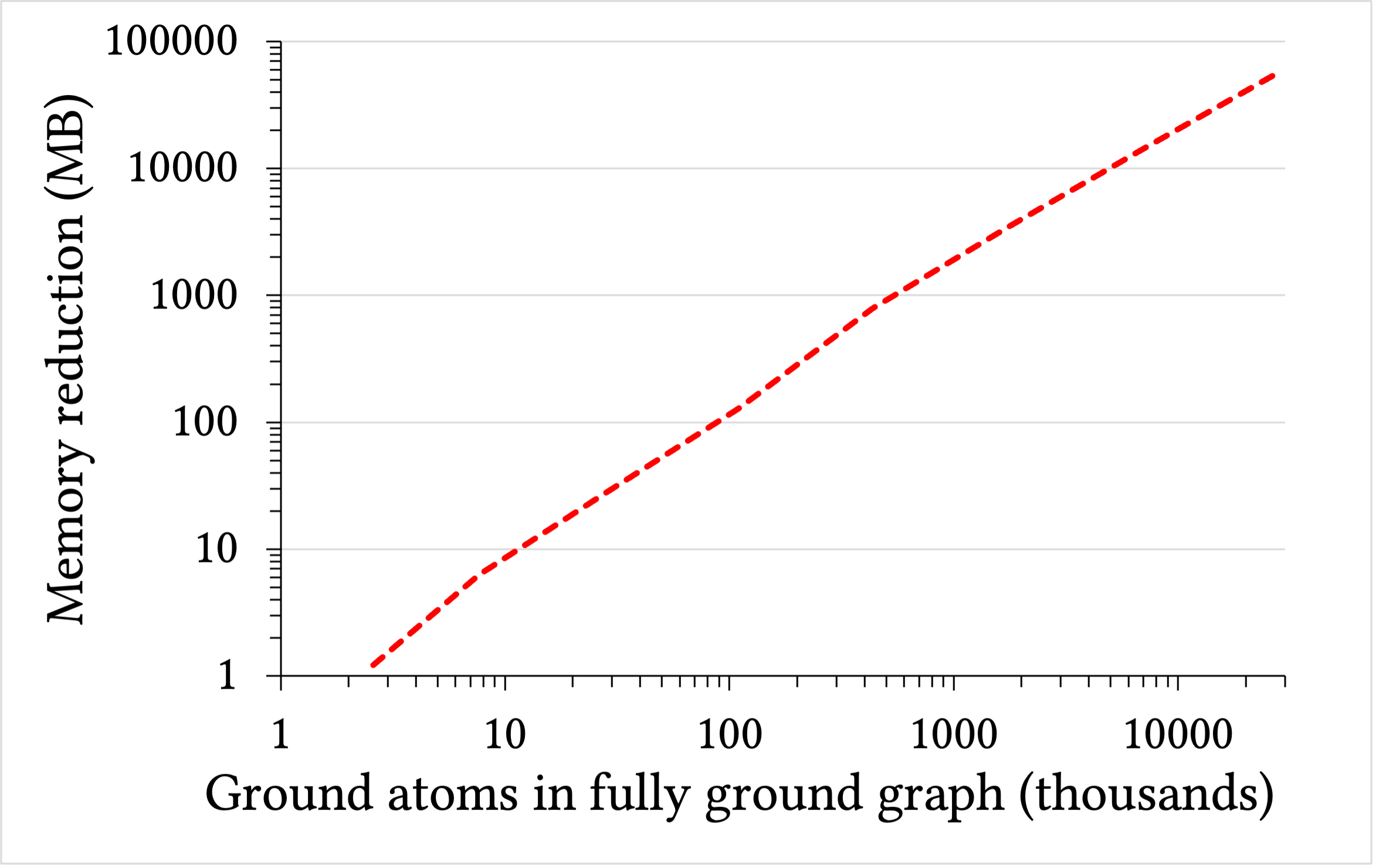}
            \end{center}
        \end{subfigure}
        \begin{subfigure}{0.45\columnwidth}
            \begin{center}
                \includegraphics[width=\linewidth, trim=3 3 3 3, clip]{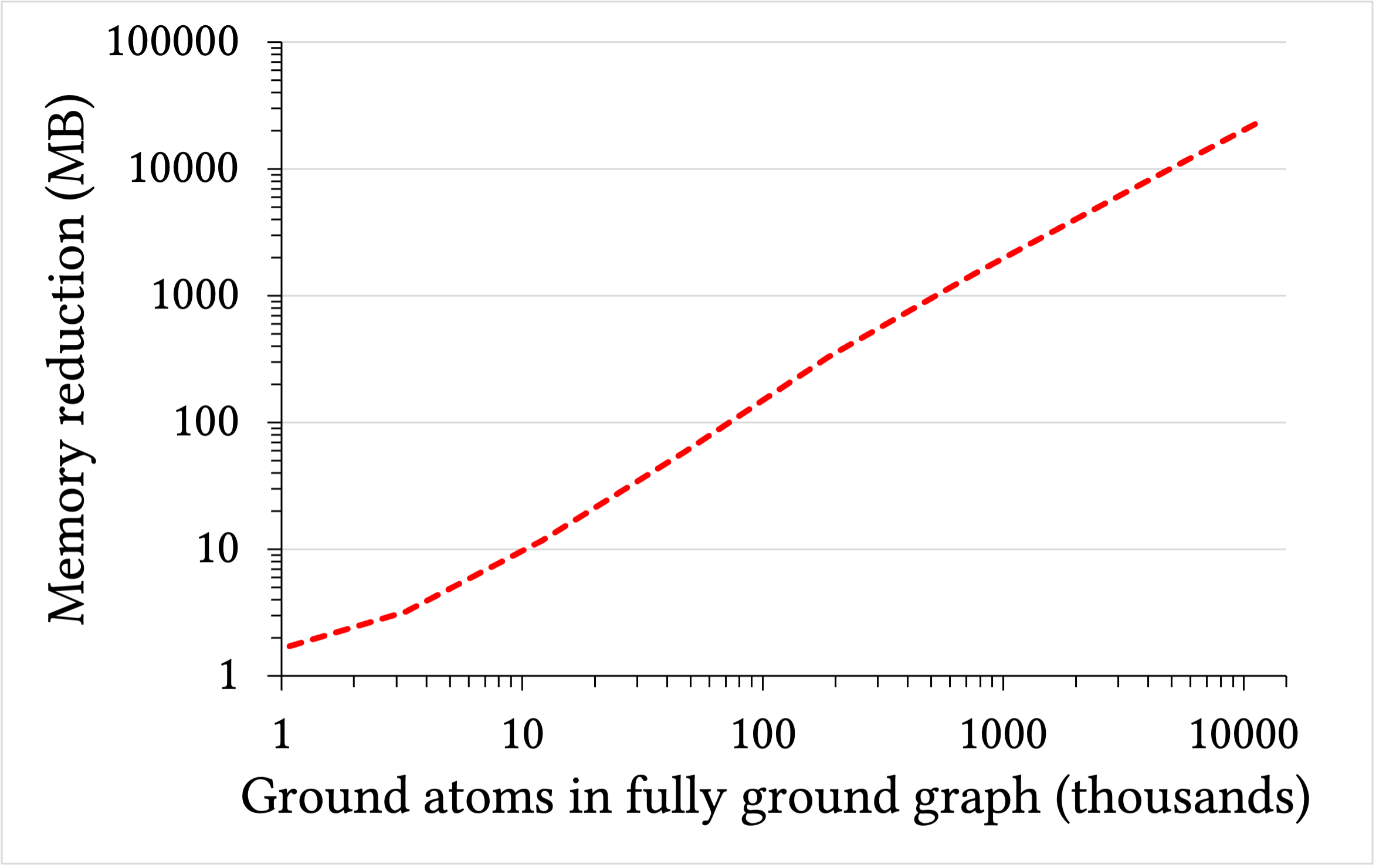}
            \end{center}
        \end{subfigure}
    \end{center}
    \caption{(Clockwise from top-left) Memory reduction (in MB) vs Map size for $\Pi_4$, $\Pi_{20}$, $\Pi_{40}$, $\Pi_{100}$.}
    \label{fig:mem-geo}
\end{figure}

In the aforementioned geospatial domain, we again make use of programs $\Pi_4$, $\Pi_{20}$, $\Pi_{40}$, and $\Pi_{100}$; for each case, we simulated agent movements for 100 actions while varying map size, thereby varying the number of ground atoms in the fully-grounded graph (Non-Skolemization case).
Figure~\ref{fig:speedup-geo} illustrates the observed speedup, defined as the ratio of running time for the non-Skolemization approach versus the Skolemization approach, across various experimental settings---note that both axes in these plots are log scale. 
A discernible pattern emerges wherein the magnitude of speedup increases substantially as the number of ground atoms escalates (corresponding to larger graph sizes). Moreover, the actual speedup demonstrates a positive correlation with the number of agents; this observation aligns with intuitive expectations, as a greater number of agents requires more groundings during the reasoning process.
Figure~\ref{fig:mem-geo} depicts the reduction in memory footprint achieved by our approach (again, both axes are log scale). 
Notably, we observe memory savings of up to~60GB for the largest graph with the highest number of agents.

\subsection{Knowledge Graph completion with multi-step reasoning}
\label{sec:kg-comp}

We conducted Knowledge Graph (KG) completion experiments on four standard datasets: UMLS~\cite{umls}, YAGO03-10~\cite{yago}, FB15k-237, and WN18RR~\cite{wn18rr_fb15k237}. 
To obtain extensive runtime and memory footprint data within a generous time limit, we created subsets of the latter three datasets using stratified sampling, preserving the underlying structure and relation types---Table~\ref{tab:kg_details} provides details of the datasets used.
The Skolemization approach graph was limited to training set triplets, while the Non-Skolemization graph included all possible edge relations, resulting in a direct relationship between graph size and the number of constants in the training set for the Non-Skolemization approach.
We used the publicly available AnyBURL rule learner~\cite{meilicke2024anytime} to generate rules, limiting the rule learning process to~20 minutes per dataset for consistency and efficiency. We filtered for rules with minimum 70\% confidence, and obtained between a few and several thousand rules for different datasets. 

KG completion was performed in PyReason using the created graph and generated rules, and evaluation results were computed using AnyBURL's evaluation script. 
To assess the utility of multi-step reasoning for KG completion tasks we use the metrics hits@k (the fraction of true predictions that appear in the top-$k$ predictions), precision (ratio of true positives to the total number of predictions), and recall (ratio of true positives to the total number of relevant instances)---we include the latter two in order to better understand KG completion performance, particularly to observe the impact of multi-step reasoning.

\begin{table}[t]
\caption{Knowledge graph datasets used \textit{(*subsets)}.}
\label{tab:kg_details}
\begin{tabular}{lccc}
\toprule
Dataset & Nodes & Edges & Unique Predicates \\
\midrule
UMLS & 135 &5,216 & 46 \\
FB15k-237* &945 & 1,108 & 237 \\
YAGO03-10* &3,029&5,020 & 37\\
WN18RR* & 8,809 & 10,007 & 11\\
\bottomrule
\end{tabular}
\end{table}

\subsubsection{Reduction in Groundings with Skolemization: Grounding in a Sparse Knowledge Graph}

\begin{figure*}[t]
    \begin{center}
        \begin{subfigure}{0.45\linewidth}
            \begin{center}
                \includegraphics[width=\linewidth, trim=2 3 3 2, clip]{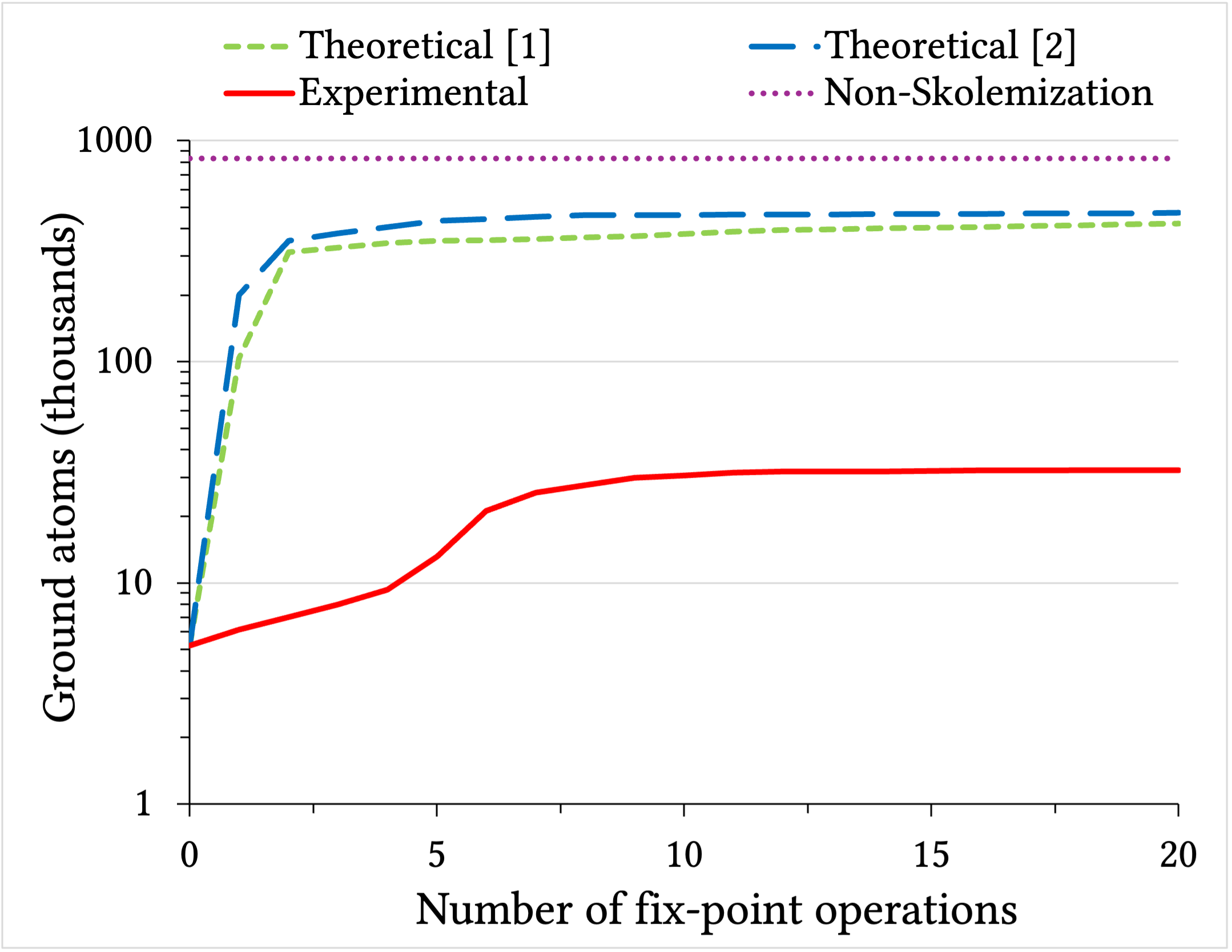}
            \end{center}
        \end{subfigure}
        \begin{subfigure}{0.45\linewidth}
            \begin{center}
                \includegraphics[width=\linewidth, trim=2 3 3 5, clip]{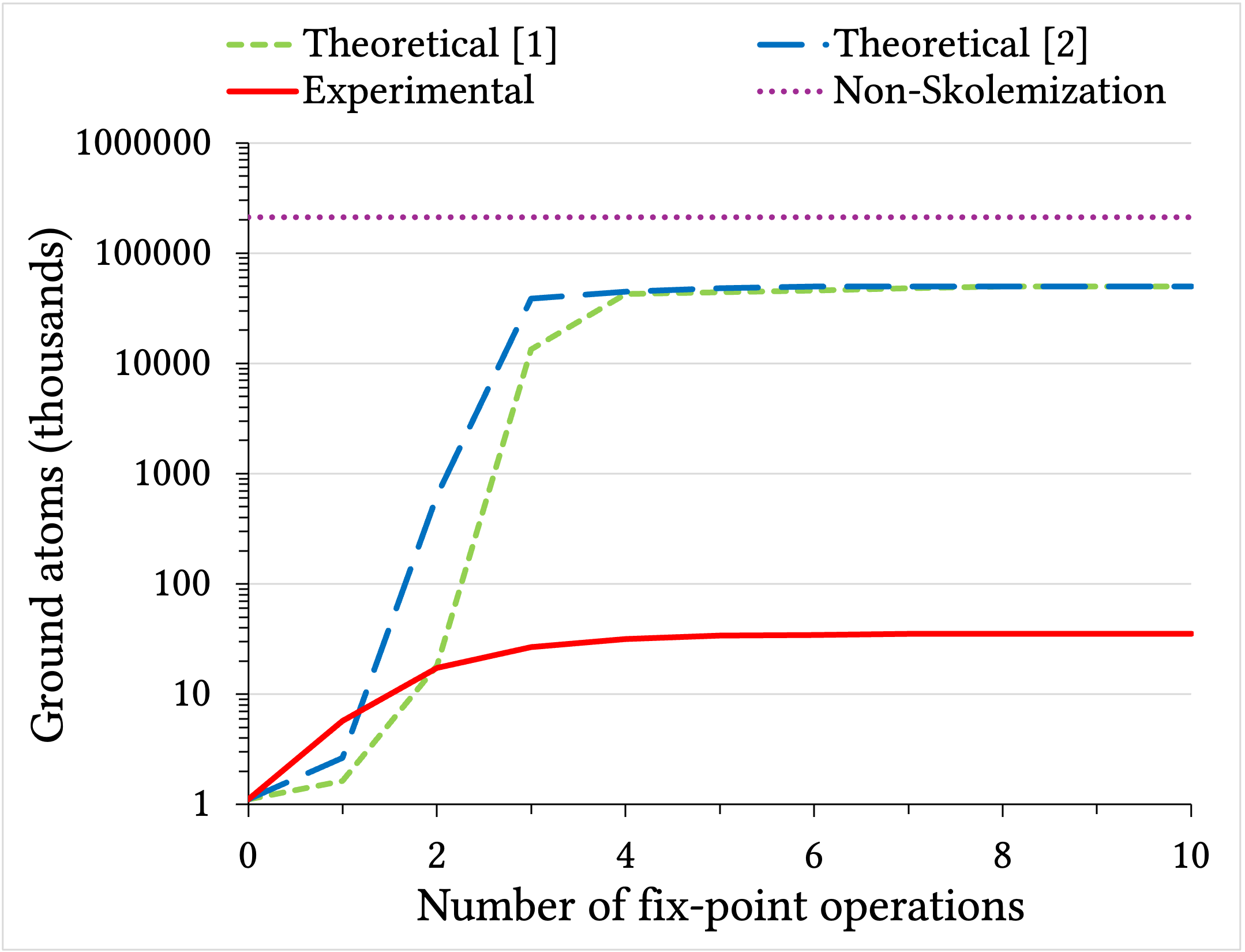}
            \end{center}
        \end{subfigure}
        \begin{subfigure}{0.45\linewidth}
            \begin{center}
                \includegraphics[width=\linewidth, trim=2 3 3 42, clip]{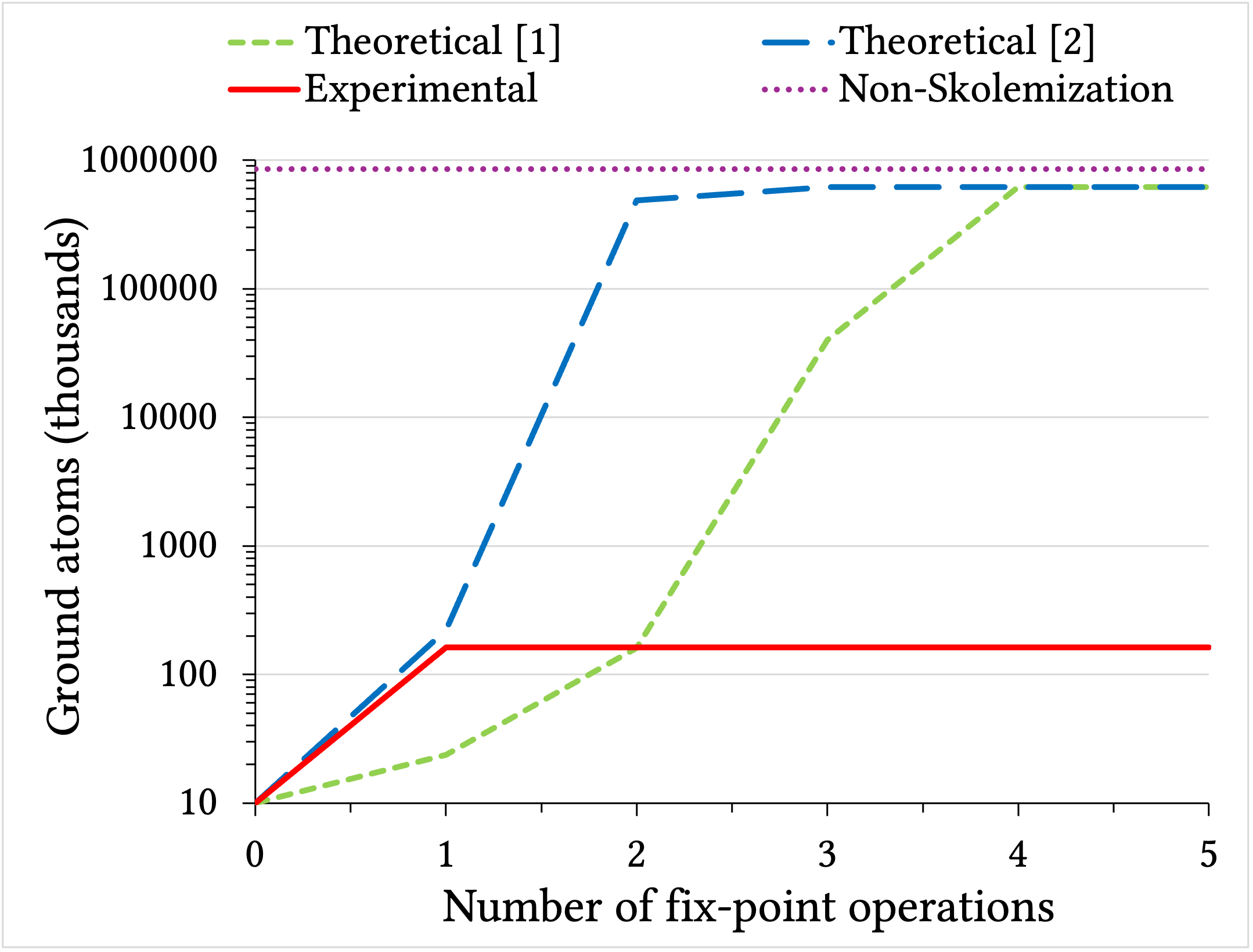}
            \end{center}
        \end{subfigure}
        \begin{subfigure}{0.45\linewidth}
            \begin{center}
                \includegraphics[width=\linewidth, trim=2 3 3 42, clip]{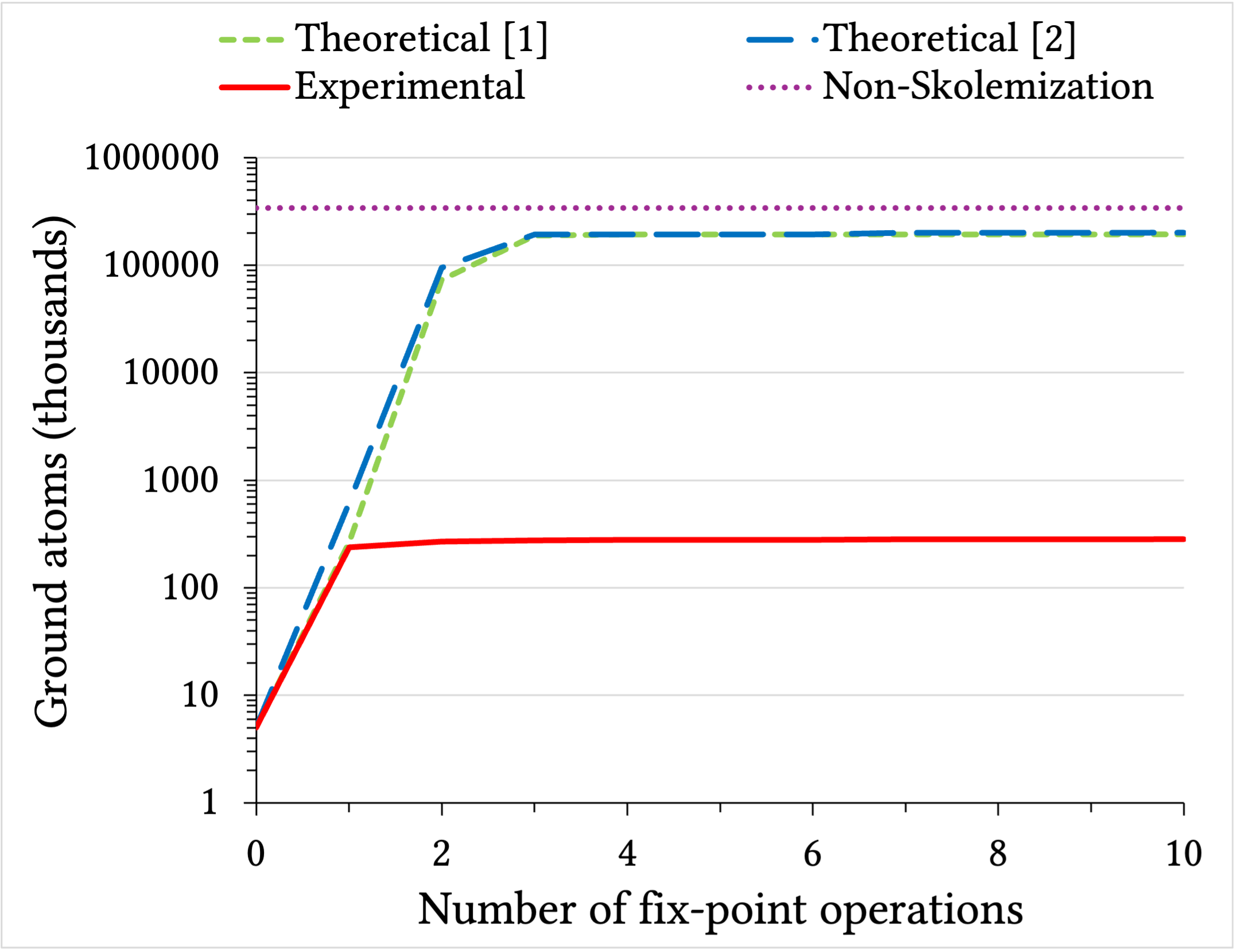}
            \end{center}
        \end{subfigure}
    \end{center}
    \caption{(Clockwise from top-left) Comparison of number of ground atoms with multiple applications of fixpoint operator for UMLS, FB15k-237, YAGO03-10, and WN18RR.}
    \label{fig:ga-kg}
\end{figure*}

For each of the four generated knowledge graphs, we perform multi-step reasoning and observe how many new groundings are made after each fixpoint operation. 
To ensure the feasibility of our comprehensive experimental suite within reasonable time constraints, we employed the following combinations of program size (quantified by ground rules) and fixpoint operations for each dataset: 
\begin{itemize}
    \item UMLS: 10,000 rules, 20 fixpoints; 
    \item FB15k-237: 1,565 rules, 10 fixpoints; 
    \item YAGO03-10: 1,000 rules, 10 fixpoints; and
    \item WN18RR: 100 rules, 5 fixpoints.
\end{itemize}
Results are plotted in Figure~\ref{fig:ga-kg} ($y$ axes on log scale), and show similar characteristics to those of the geospatial domain. 
The theoretical limit, while initially approximating experimental values at lower fixpoint operations, rapidly increases before stabilizing at a generous upper bound. We note that cases where the number of constants produced exceeds the theoretical bound (in lower inference steps) are due to the choice of parameters.  Despite this, the upper bound consistently remains significantly below the number of ground atoms for the non-Skolemization approach.
An intriguing observation emerged regarding the influence of unique predicates in the head of non-ground rules within a program. A higher diversity of predicates typically correlated with delayed convergence and a greater number of ground atoms at convergence. These findings support our hypothesis that practical applications require materializing only a small fraction of all possible groundings. 
The Skolemization approach thus facilitates efficient reasoning in logic programming by substantially reducing the number of required ground atoms.

\subsubsection{Scalability: Scaling with Rules and Fixpoint Applications}

\begin{figure}[t]
    \begin{center}
        \begin{subfigure}{0.4\columnwidth}
            \begin{center}
                \includegraphics[width=\linewidth, trim=3 3 3 3, clip]{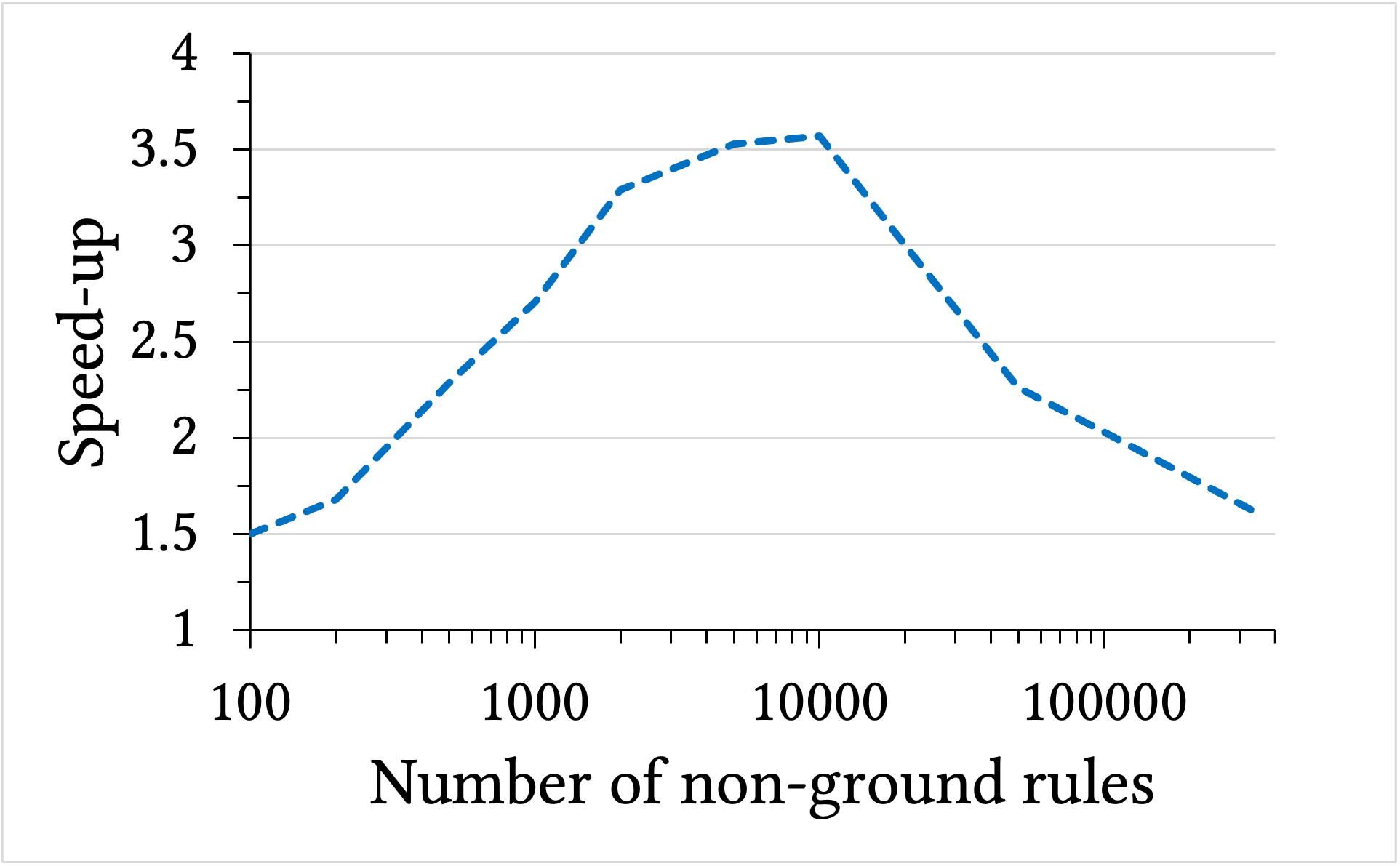}
            \end{center}
        \end{subfigure}
        \begin{subfigure}{0.4\columnwidth}
            \begin{center}
                \includegraphics[width=\linewidth, trim=3 3 3 3, clip]{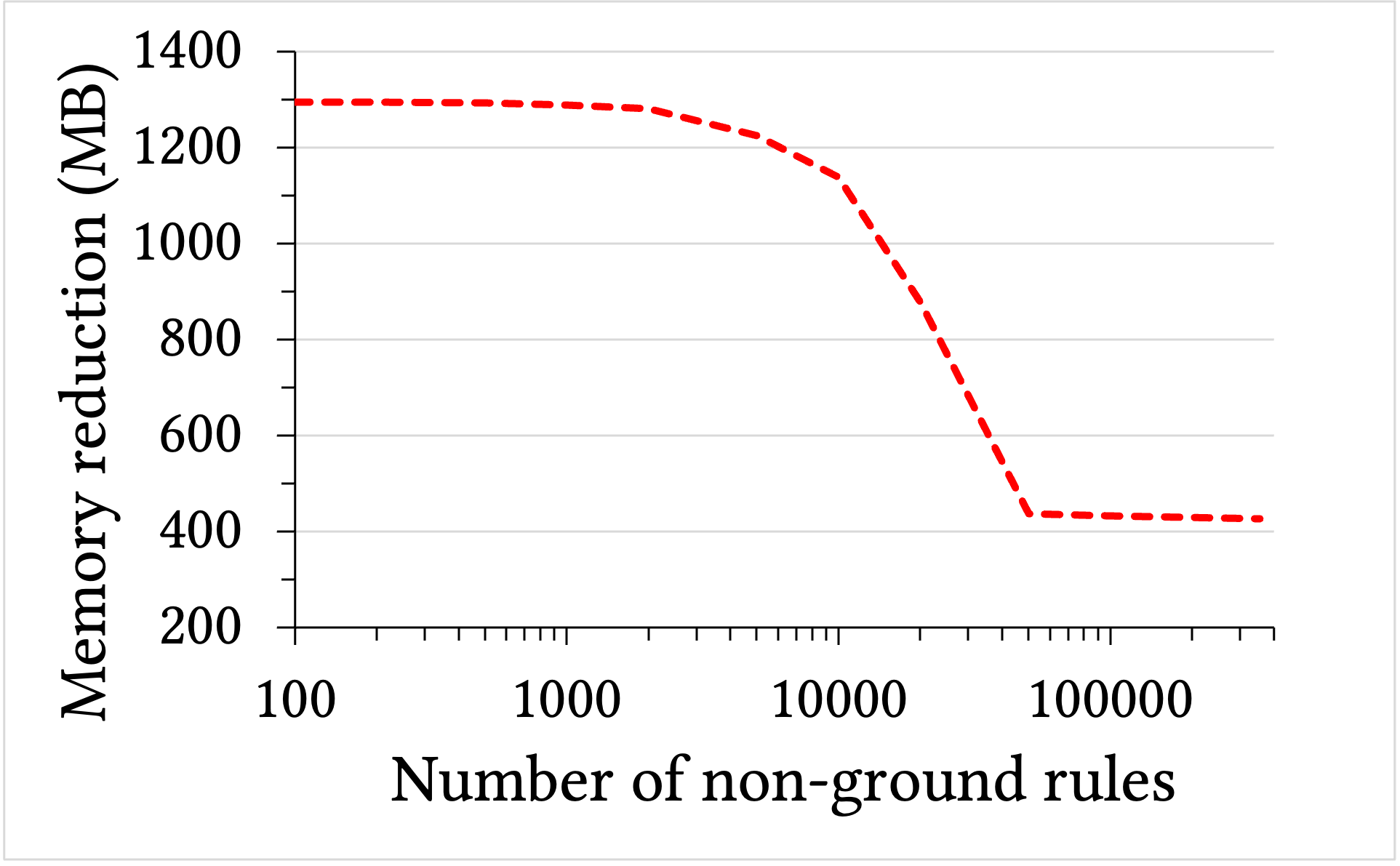}
            \end{center}
        \end{subfigure}
        \begin{subfigure}{0.4\columnwidth}
            \begin{center}
                \includegraphics[width=\linewidth, trim=3 3 3 3, clip]{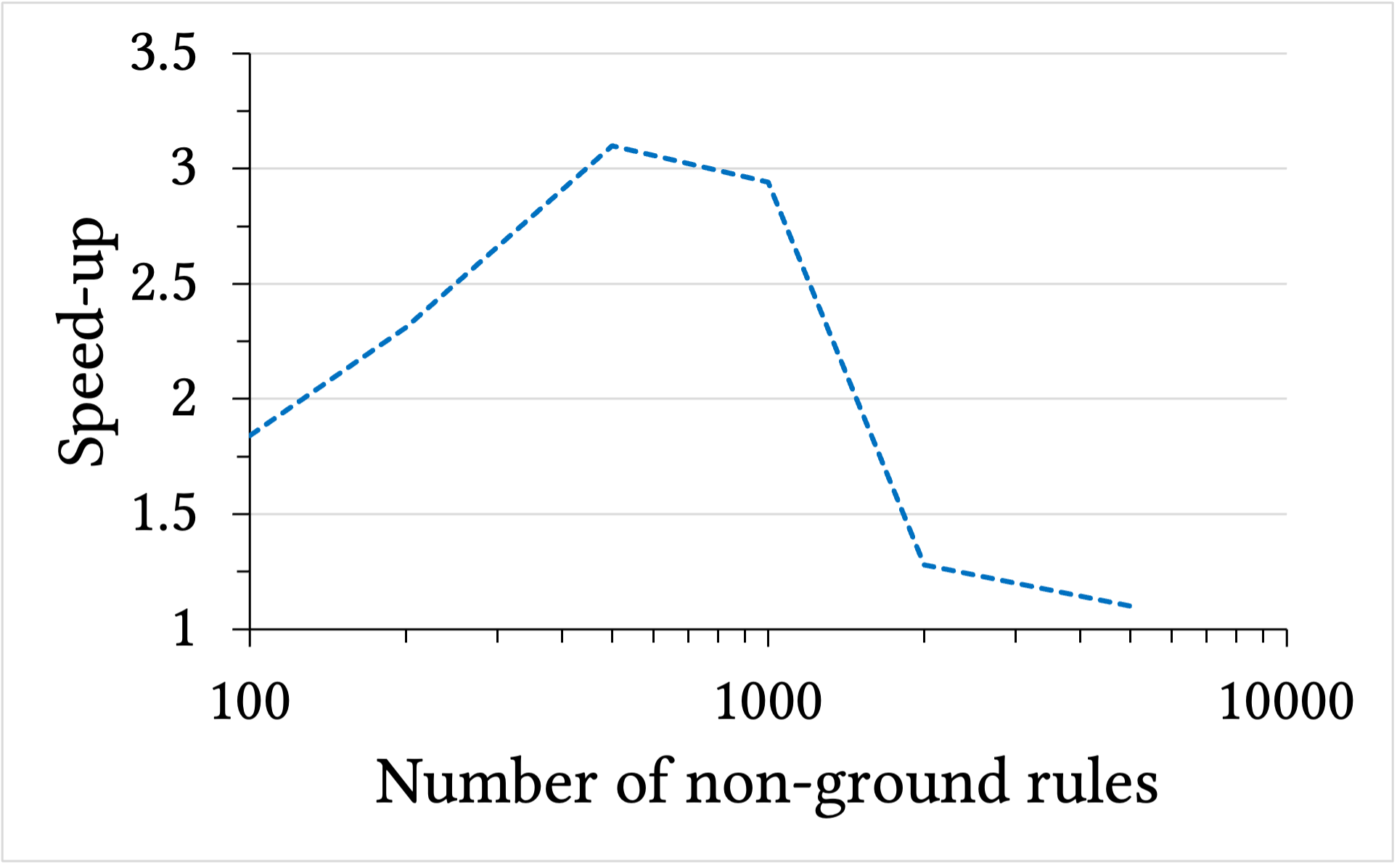}
            \end{center}
        \end{subfigure}
        \begin{subfigure}{0.4\columnwidth}
            \begin{center}
                \includegraphics[width=\linewidth, trim=3 3 3 3, clip]{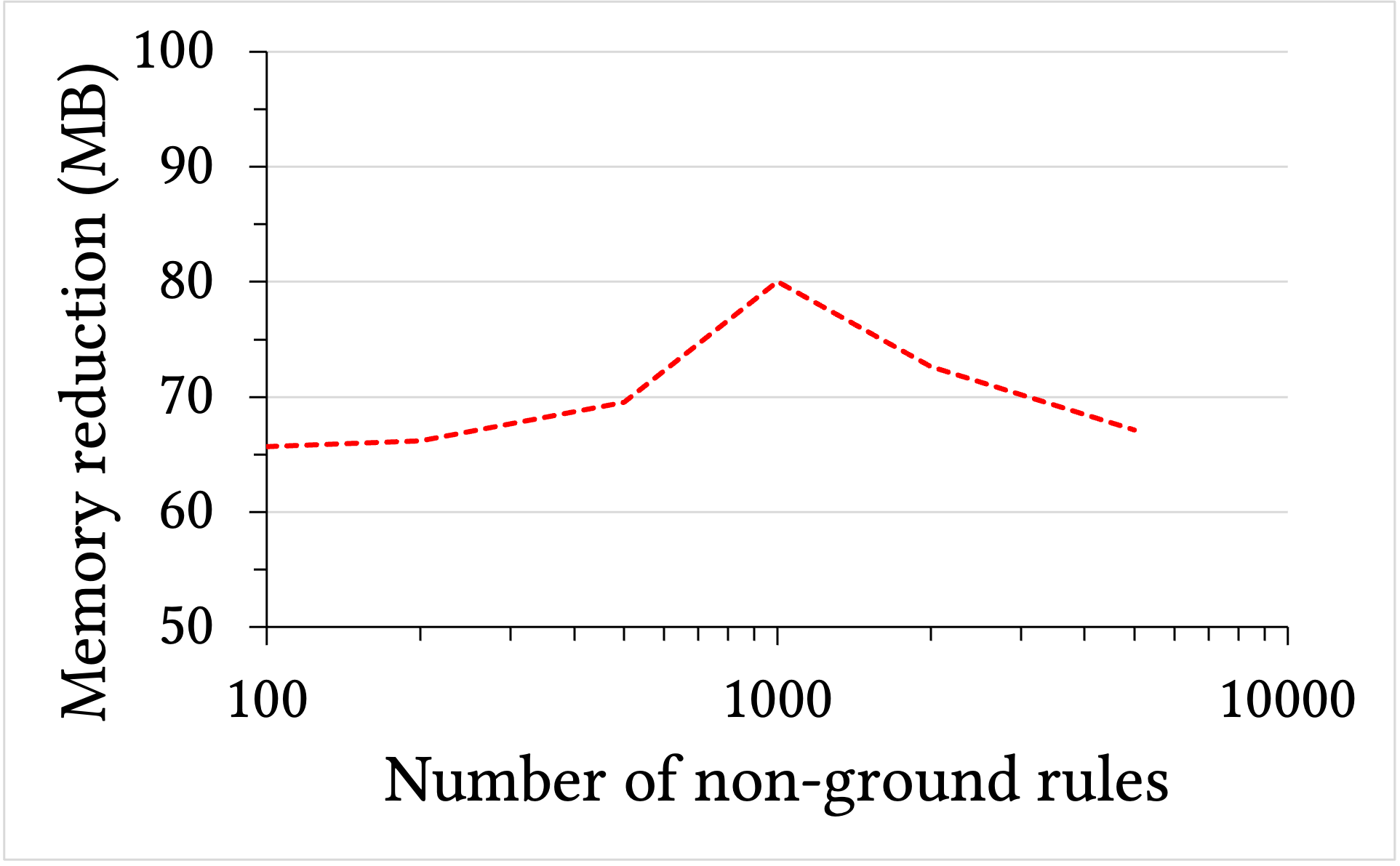}
            \end{center}
        \end{subfigure}
    \end{center}
    \caption{Speed-up and Memory reduction vs Program size for 2 (top) and 5 (bottom) fixpoints on UMLS dataset.}
    \label{fig:umls-sc}
\end{figure}

To investigate the impact of program size and reasoning steps on running time and memory footprint, we selected the UMLS knowledge graph; this choice was based on its manageable file size and reasonable experimental time requirements. 
For context, while the fully-grounded UMLS graph requires 1.2GB of memory, our subgraphs of FB15k-237, YAGO03-10, and WN18RR required 32GB, 66GB, and over 100GB of disk space, respectively. Moreover, these larger graphs needed up to 1TB of memory for reasoning operations.

Figure~\ref{fig:umls-sc} illustrates the speed-up and memory reduction (in MB) for~2 and~5 reasoning steps ($x$-axes on log scale). 
Maximum speedup was obtained between 1,000 and 10,000 rules, consistently outperforming the non-Skolemization approach. From two to five fixpoints, the peak efficiency shifted slightly towards a lower number of rules, corroborating our previous findings that increased fixpoint operations lead to more inferences and, consequently, increased running time, thus reducing speedup. 
With fewer rules, the growth of ground atoms is slower, resulting in faster running times. Similarly, memory efficiency decreases as the number of rules increases, due to the significant increase in inferences after each fixpoint.
We imposed a 12-hour time limit for each experiment, consequently obtaining results for up to 5K non-ground rules for five fixpoints, compared to 355K for two fixpoints. Experiments were conducted on 128 cores of AMD~EPYC~7413 with a maximum of 1TB allocated memory.

\subsubsection{Multi-step Inference}

\begin{table}[t]
\caption{Performance metrics with single and multi-step inference}
\label{tab:fixpoint_metrics}
\begin{tabular}{lcccccccccc}
\toprule
Dataset & $\Gamma$ & \multicolumn{3}{c}{Hits@k} & MRR & \# ground atoms & \# rules & \# nodes & \# edges & \# queries\\\cmidrule(l){3-5}
 & & 1 & 3 & 10 & & \\
\midrule
WN18RR & 1 & 0.043 & 0.065 & 0.092 & 0.058 & 725,666& 500 & 8,809 & 10,007 & 6,268 \\
& 2 & 0.043 & 0.087 & 0.127 & 0.07 & 763,636\\
\midrule
FB15k-237 & 1 & 0.101 & 0.121 & 0.202 & 0.1248 & 1,676 & 400 & 945 & 1,108 & 40,932 \\
& 2 & 0.136 & 0.161 & 0.237 & 0.162 & 4,131\\
\midrule
YAGO03-10 & 1 & 0.357 & 0.429 & 0.429 & 0.393 & 107,014 & 500 & 3,029 & 5,020 & 10,000 \\
& 2 & 0.357 & 0.464 & 0.464 & 0.413 & 115,078\\
\midrule
UMLS & 1 & 0.054 & 0.093 & 0.099 & 0.073 & 186 & 200 & 135 & 5,216 & 1,322 \\
& 2 & 0.055 & 0.097 & 0.104  & 0.076 & 267\\
\bottomrule
\end{tabular}
\end{table}

Multi-step reasoning employs sequential applications of logical inference to build upon previously established knowledge after each fixpoint application, distinguishing it from multi-hop reasoning, which refers to the process of aggregating and connecting information across multiple pieces of evidence or sources. In our approach, we perform multi-step reasoning while relaxing any closed world assumption. We hypothesize that multi-step reasoning can help uncover deeper knowledge in a variety of scenarios, and we present some of our early findings in this direction.
We conduct our experiments on subsets of all four datasets to test our hypotheses on graphs of significantly varying sizes, density, and with different number of rules. Table~\ref{tab:fixpoint_metrics} illustrates that multi-step inference, involving just two fixpoint applications, consistently enhances performance across the Hits@k and Mean Reciprocal Rank (MRR) metrics, compared to single-step inference. Notably, FB15k-237 demonstrates improvements of~7.4\%, 14.8\%, and 13.7\% on average for Hits@k. The MRR is also shown to always improve, with increases of over~25\% for the WN18RR and FB15k-237 datasets. 
These findings indicate that multi-step inference could significantly enhance result retrieval capacity with multiple applications of the fixpoint operator, representing preliminary evidence of the approach’s potential that warrants further investigation.

\subsection{Logic as a Simulator for RL Applications}
\label{sec:rl-expt}

In this section, we benchmark our approach against two popular simulators.
We begin by introdcing the simulators, then we outline the two game scenarios we use in our experiments, analyze the limitations of Markov assumptions, and discuss the RL training methodology adopted.
Then, we present experimental results comparing our approach's scalability in Starcraft~II and AFSIM, along with its ability to learn policies in PyReason and port them to other simulators. We then explore whether incorporating non-Markovian dynamics in the simulation can improve RL algorithms' ability to learn effective policies for complex games. Additionally, we demonstrate the explainability of our approach using rule traces, highlighting its potential in reward shaping during training.

\subsubsection{Benchmarks: Popular Simulators}
\label{sec:rl-benchmarks}

In order to position PyReason as an appropriate simulator, we first compare it to two established simulators in the field:
\begin{enumerate}
\item \em{Starcraft II} (SC2) is a popular real-time strategy (RTS) video game developed by Blizzard Entertainment, and has a competitive multiplayer aspect that involves managing resources, building armies, and engaging in tactical battles. 
Due to its complex gameplay and emphasis on strategic decision-making, it has been considered as a potential tool for military simulations. 
We extended Deepmind's PySC2~\cite{vinyals2017starcraft} to use the Starcraft~II environment in our experiments\footnote{Extensions to PySC2: \url{https://github.com/lab-v2/pysc2-labv2}}.

\item \em{Advanced Framework for Simulation, Integration, and Modeling software} (AFSIM)~\cite{clive2015advanced} is a powerful simulation tool used by the United States Department of Defense (DoD) for various purposes, including training, analysis, experimentation, and mission planning. 
AFSIM is developed by the Air Force Research Laboratory (AFRL) and is used primarily by the United States Air Force (USAF) as well as other branches of the military and defense organizations. AFSIM is a high-fidelity modeling and simulation software designed to provide realistic representations of aerial warfare scenarios and environments. It enables the USAF to assess and analyze the performance of various systems, strategies, and tactics in simulated combat situations.
\end{enumerate}
To compare PyReason with SC2 and AFSIM, we design the scenarios and game dynamics in all three simulators.

\subsubsection{Game Setup}
\label{sec:rl-setup}

We design a simple grid world war game as shown in Figure~\ref{fig:game-setup}. The basic scenario has two teams (red and blue) of one agent each. 
Each team has a base, and there are also a few obstacles (shown as mountains) in the environment that are impenetrable and impassable. 
For this base scenario, the objective of the game is to capture (reach) the rival base before the enemy can do the same. 
The red team follows our learned RL policy (the agent(s)), whereas the blue team follows a pre-defined base policy (the opponent(s)) described later in this section. 
Later on we build upon this basic scenario by adding more agents and then extending the action and observation spaces.

\begin{figure}[t]
    \begin{center}
        \includegraphics[width=0.45\columnwidth]{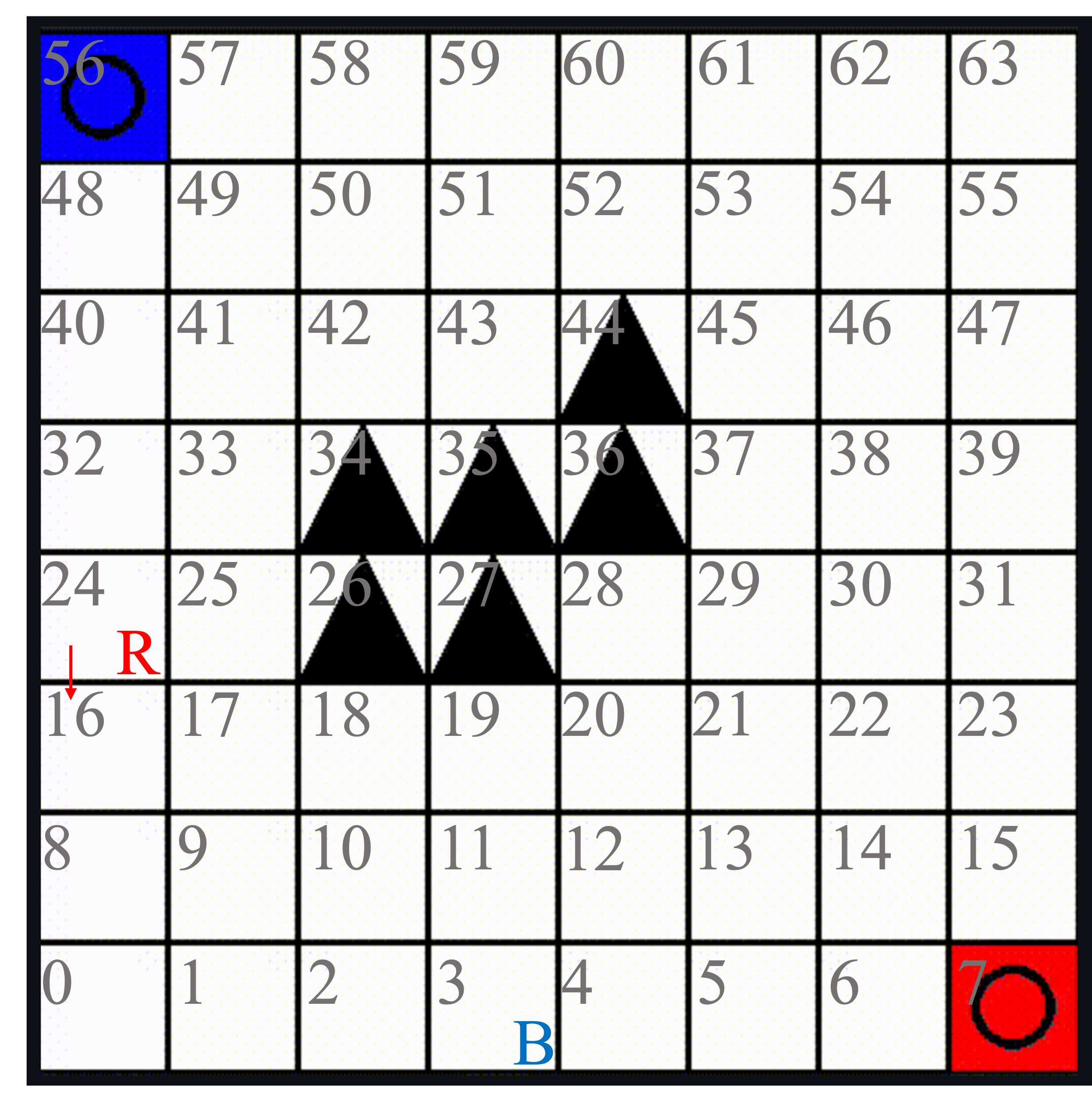}
    \end{center}
    \caption{Grid map for the scenario. Red (bottom-right) and Blue (top-left) squares are fixed base locations for each team. All agents start at their respective base locations. Obstacles (mountains) are shown with black triangles. Bottom left quadrant of the grid map is marked with indices to aid the understanding of the explainable trace in Table~\ref{tab:short_rule_trace}.}
    \label{fig:game-setup}
\end{figure}

\subsubsection{Scalability}

Allowing the agents to take random actions in the grid world, we compare the scaling capability of our software against other simulators by comparing the running time and memory footprint over a large number of actions for different number of agents per team.
The experiments were performed on an AWS EC2 container with 96~vCPUs (48~cores) and 384GB memory.

Figure~\ref{fig:runtime-mem-comp} show the scaling capability of the different simulators tested. We note that, among the two established simulation environments, AFSIM generally performed better with~5 agents per team; with~20 agents per team, AFSIM is overtaken by SC2 as the actions per agent increase.  This is expected, as AFSIM is designed as a high-fidelity simulation environment, so we can expect greater computational cost with more complex situations. PyReason consistently outperformed SC2, achieving anywhere from a one to nearly three orders of magnitude improvement. 
Though PyReason performs comparably to AFSIM for lower numbers of actions per agent (which are arguably the least important in practice), it also achieved comparable multiple order-of-magnitude improvements in terms of running time as the number of actions per agent increased. 
This suggests that PyReason will scale to large environments where the traditional use of simulators would otherwise prohibit model training.

Additionally, we examined memory consumption (Figure~\ref{fig:runtime-mem-comp}). PyReason uses considerably less memory compared with SC2 over all configurations while having sub-linear ($R^2 = .84$) growth with action and agent space. 
AFSIM's strength as a large-scale military simulator is shown here with little effect on memory consumption with change in agents or actions; however, it has a large base memory cost that was still significantly higher than that of PyReason for the largest case considered (40,000 actions in total).

\begin{figure}[t]
    \begin{center}
        \begin{subfigure}{0.45\columnwidth}
            \begin{center}
                \includegraphics[width=\linewidth, trim=3 3 3 2, clip]{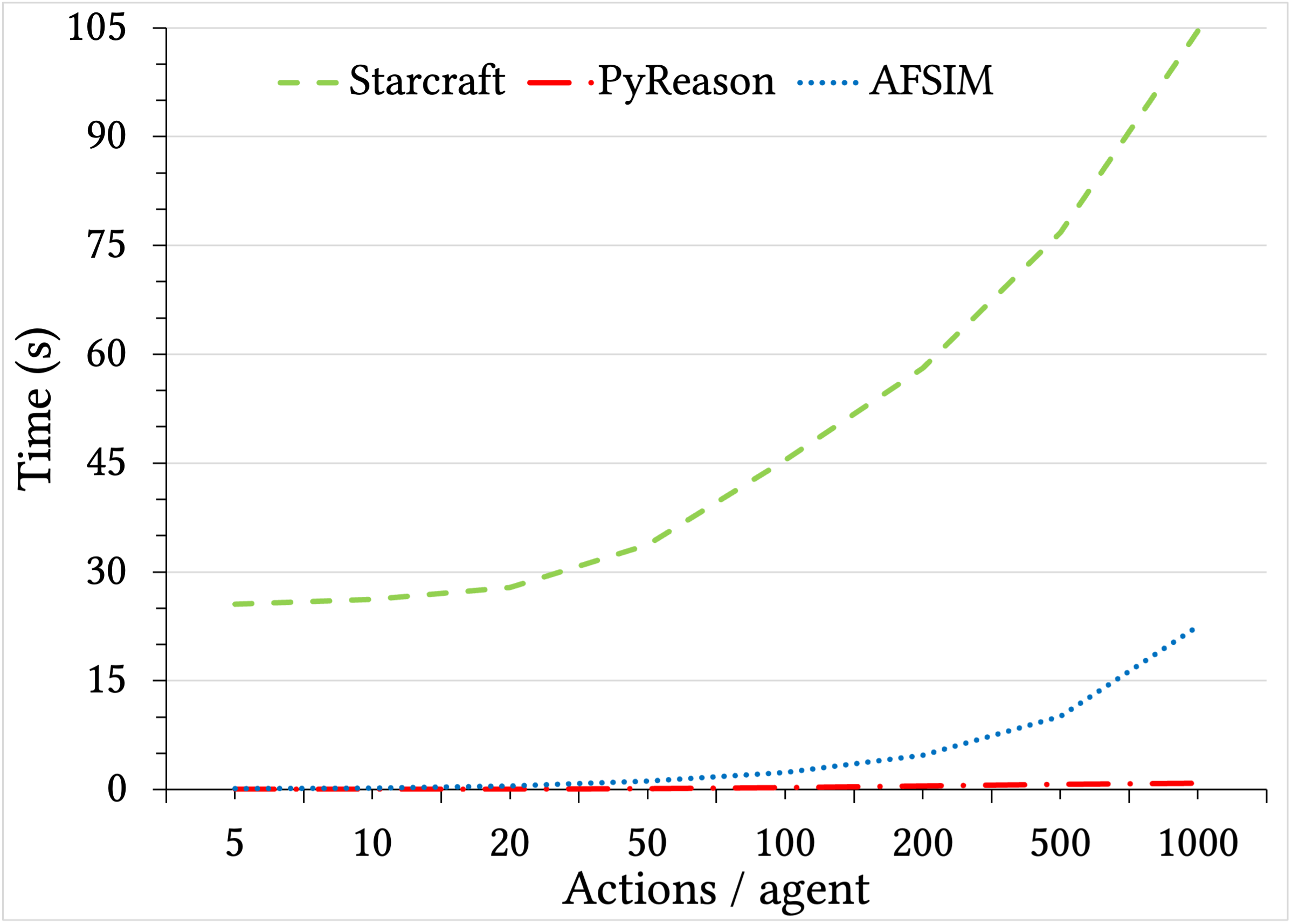}
            \end{center}
        \end{subfigure}
        \begin{subfigure}{0.45\columnwidth}
            \begin{center}
                \includegraphics[width=\linewidth, trim=2 3 3 2, clip]{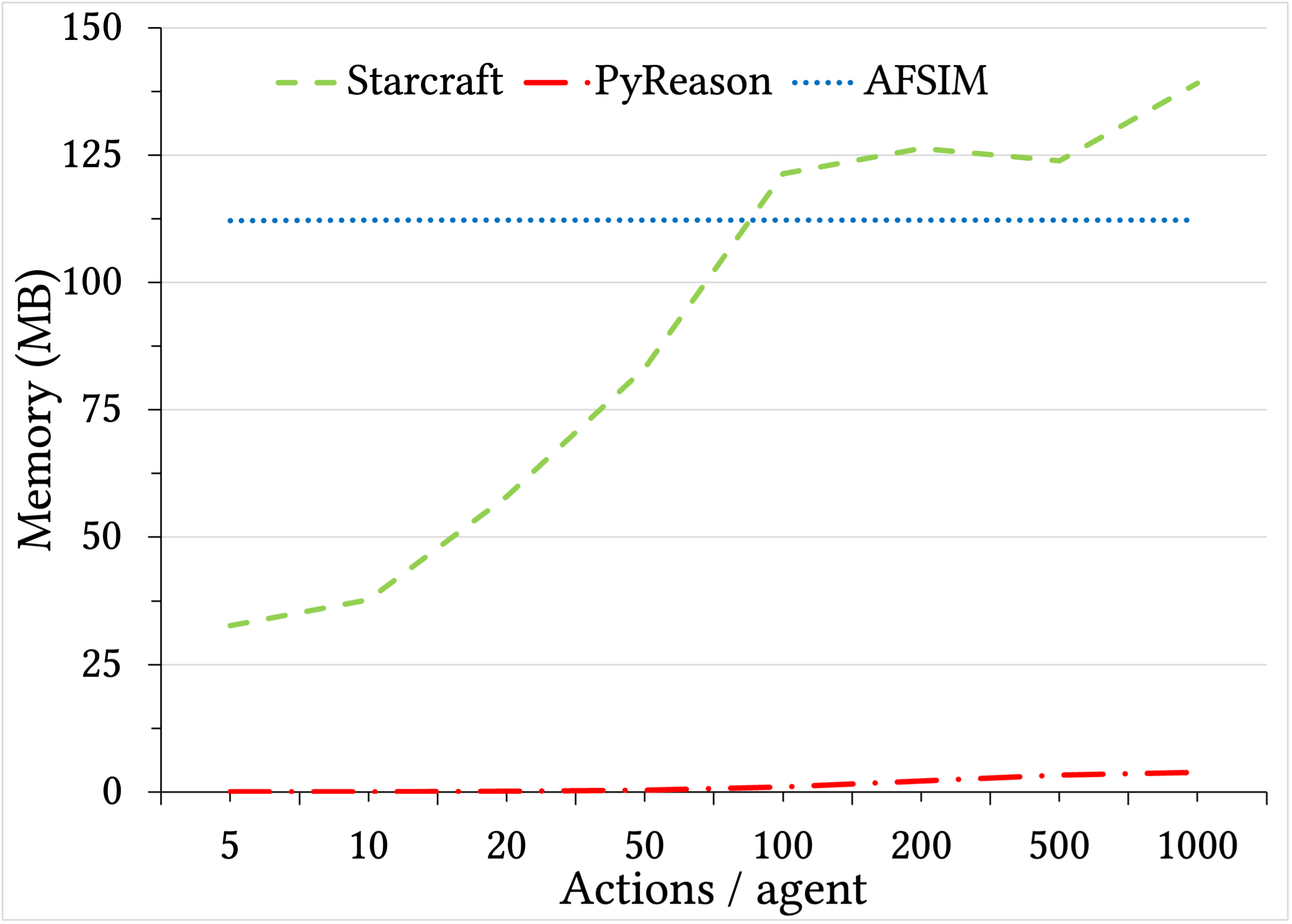}
            \end{center}
        \end{subfigure}
        \begin{subfigure}{0.45\columnwidth}
            \begin{center}
                \includegraphics[width=\linewidth, trim=2 2 2 2, clip]{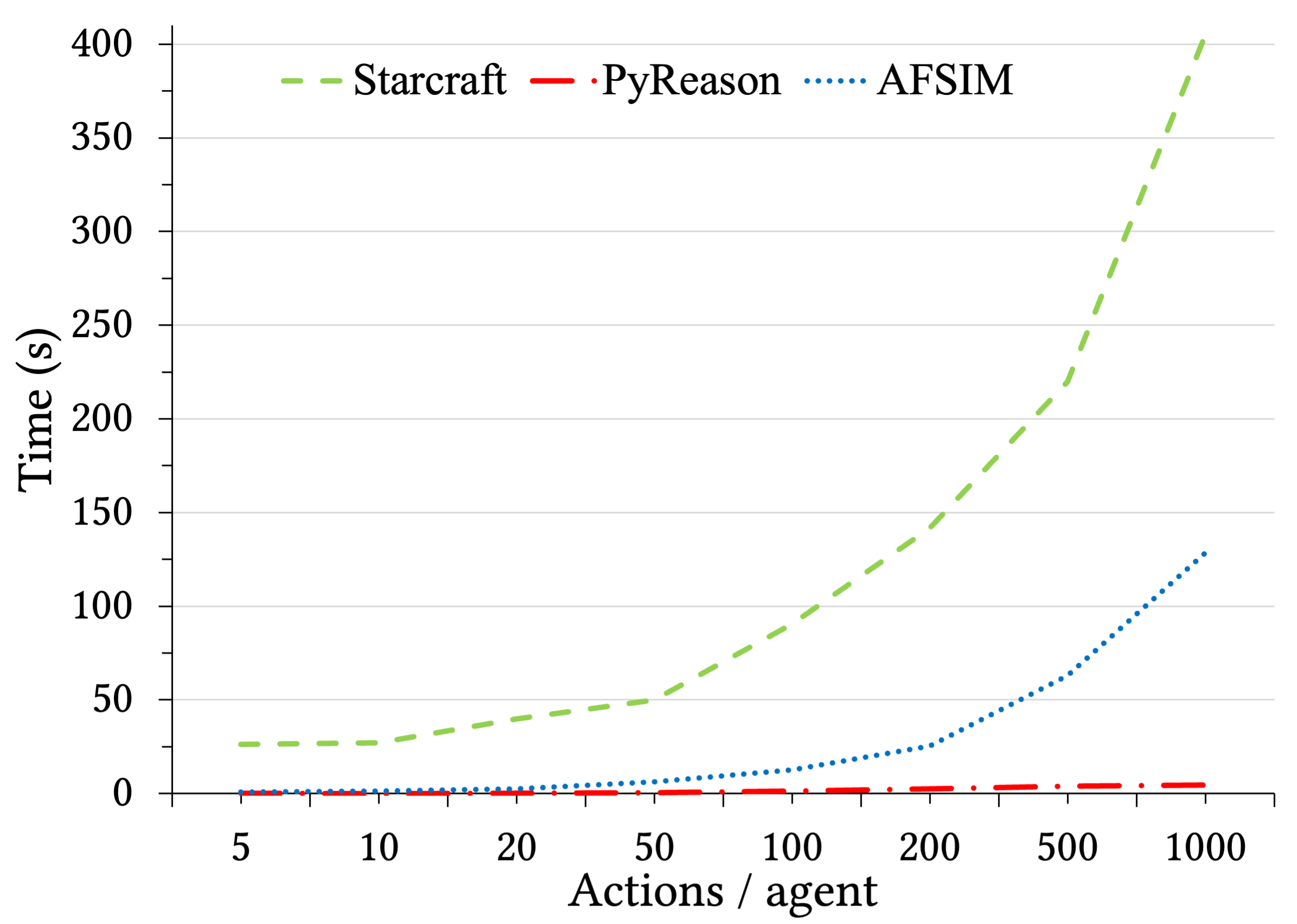}
            \end{center}
        \end{subfigure}
        \begin{subfigure}{0.45\columnwidth}
            \begin{center}
                \includegraphics[width=\linewidth, trim=2 2 2 2, clip]{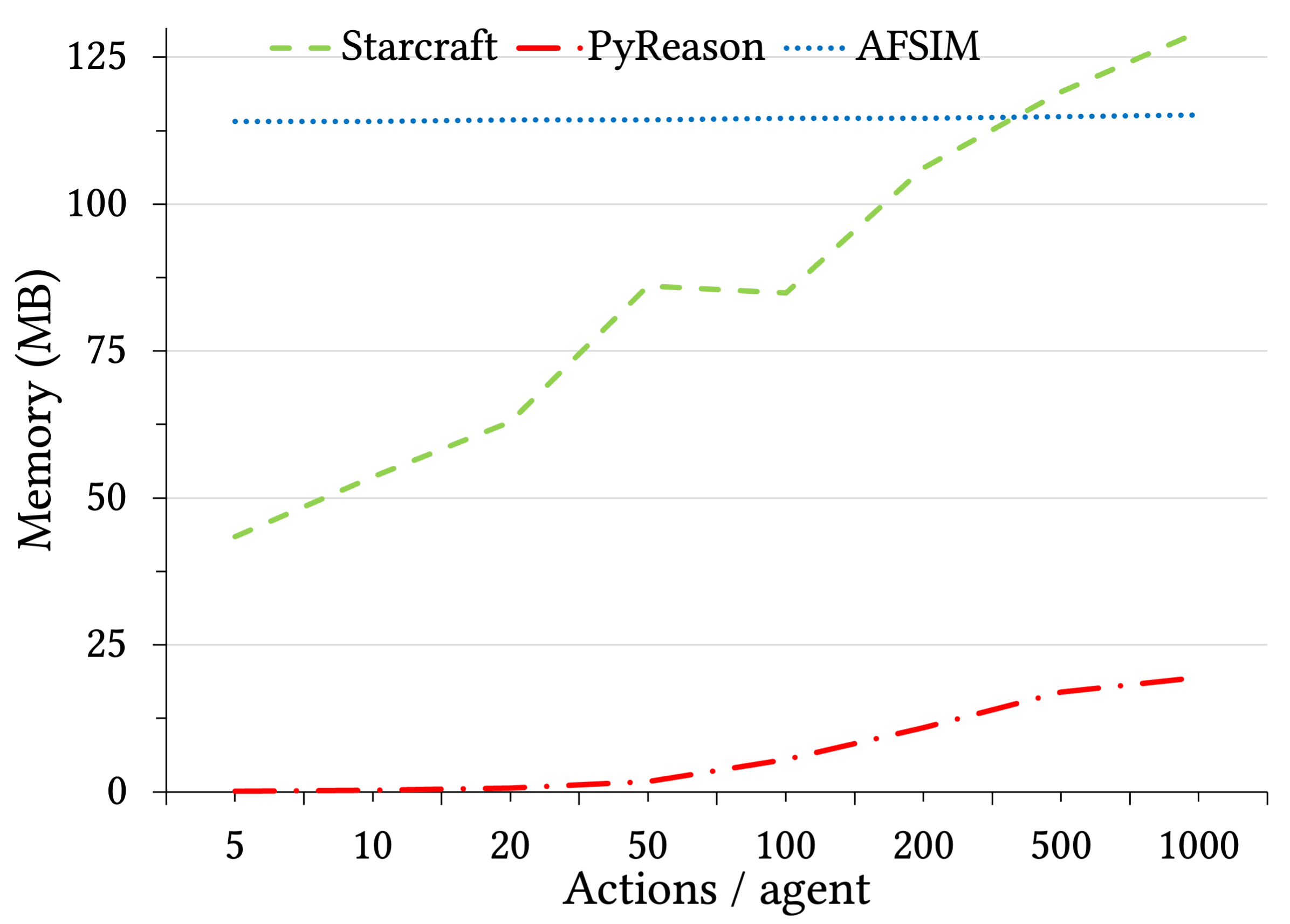}
            \end{center}
        \end{subfigure}
        \begin{subfigure}{0.45\columnwidth}
            \begin{center}
                \includegraphics[width=\linewidth, trim=2 2 2 2, clip]{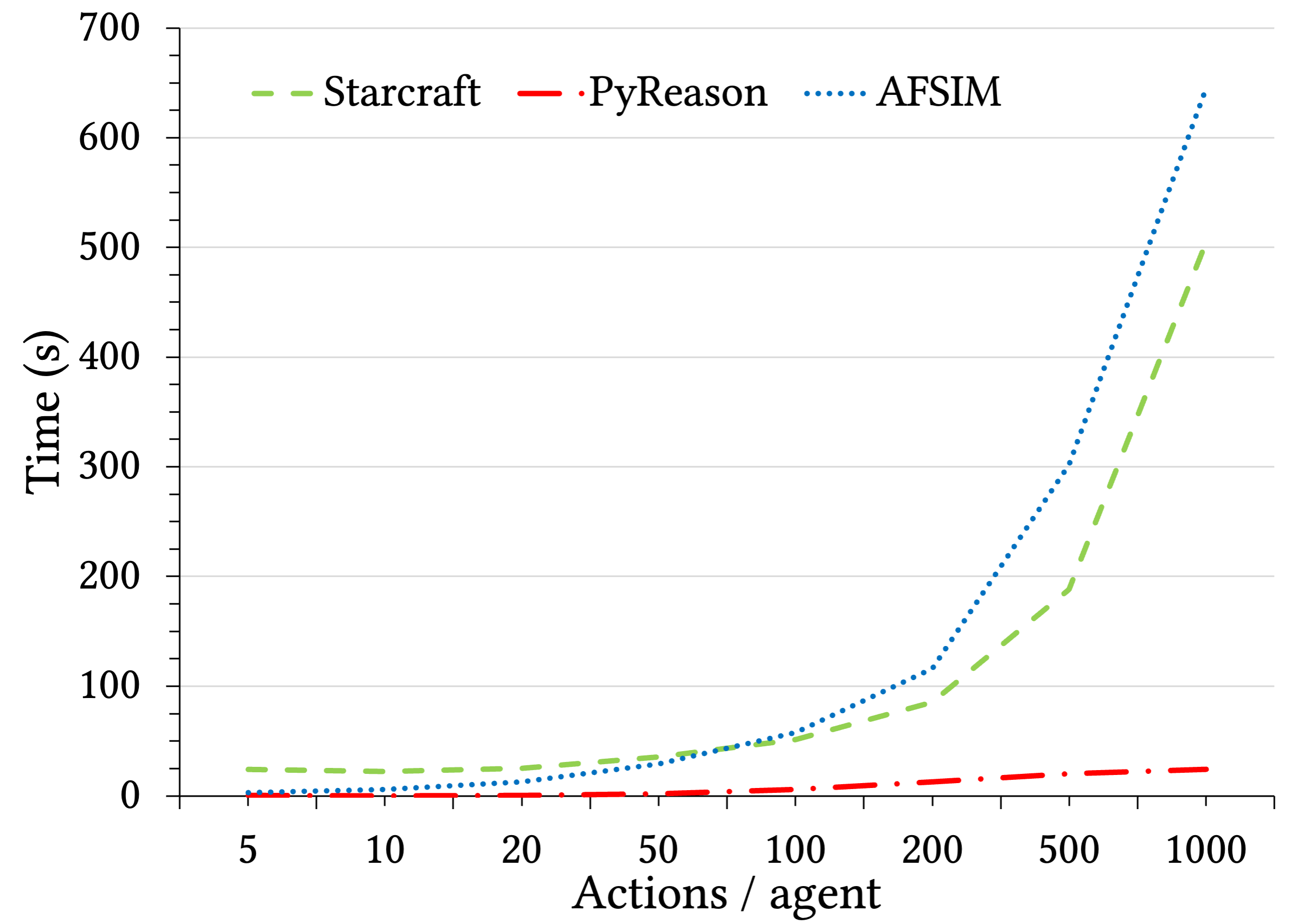}
            \end{center}
        \end{subfigure}
        \begin{subfigure}{0.45\columnwidth}
            \begin{center}
                \includegraphics[width=\linewidth, trim=2 2 2 2, clip]{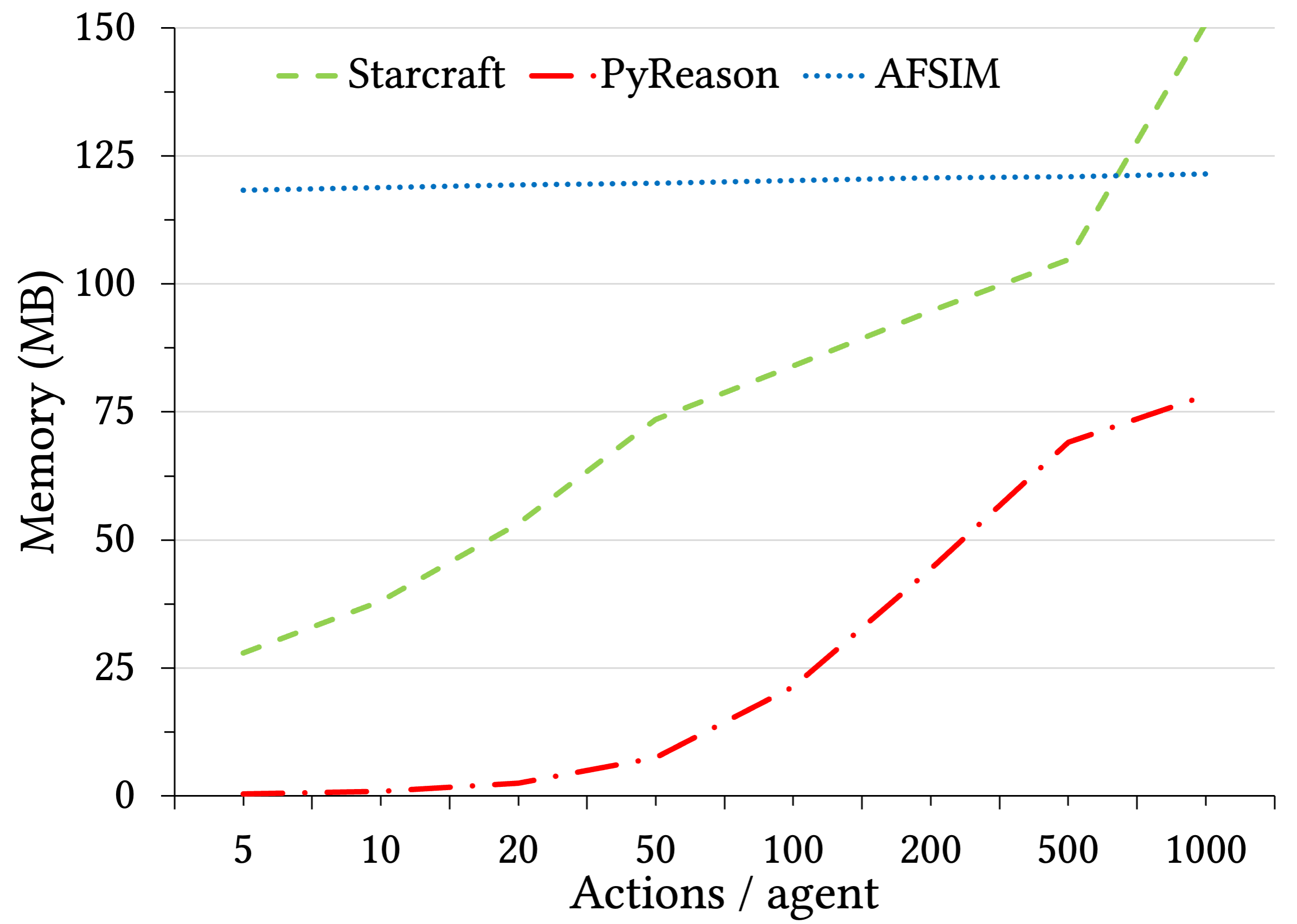}
            \end{center}
        \end{subfigure}
    \end{center}
    \caption{Runtime (left) and memory footprint (right) comparison when~1 (top), 5 (middle), 20 (bottom) agents/team take random actions in three simulation environments.}
    \label{fig:runtime-mem-comp}
\end{figure}

\subsubsection{Portability}
\label{sec:portability}

Next, we wanted to test whether a Reinforcement Learning (RL) agent trained in PyReason (PR) can provide performance comparable to AFSIM (AFS) and PySC2 (SC2). For this, we considered two cases: single agent and multi (five) agents per team. 
At certain intervals during the training process, policies were extracted and were used to play the base scenario described earlier~500 times in each of the three simulators (PyReason, AFSIM, and PySC2) and the outcomes were compared.

When policies learned in PyReason played the base scenario, comparable numbers were observed for all three simulators, as shown in Table~\ref{tab:portability}---variance can be attributed to inherent randomness in learned policies). 
These results suggest that the approach is generalizable, as an agent trained in PyReason can be ported to various simulation environments and achieve comparable reward and win percentage.

\begin{table}[t]
\caption{Performance metrics when PyReason trained policies were used to play the game on different simulators for single and multi~(5) agent scenarios---numbers in parentheses specify differences with respect to PyReason.}
\label{tab:portability}
\begin{tabular}{cccccccc}
\toprule
 \# & Epochs & \multicolumn{3}{c}{Avg. Reward} & \multicolumn{3}{c}{Win \%} \\ \cmidrule(l){3-5} \cmidrule(l){6-8} 
 & & PR & SC2 & AFS & PR & SC2 & AFS \\ \midrule
1 & 400K & -209.87 & -210.15 & -222.65 & 0.0 & 0.0 & 0.0 \\
 &  &  & (-0.13\%) & (-6.09\%) & & (0) & (0)\\
 & 544K & 162.51 & 165.64 & 168.04 & 43.0 & 42.8 & 44.0 \\
 &  &  & (+1.93\%) & (+3.40\%) &  & (-0.2) & (+1)\\
 & 760K & 482.50 & 487.00 & 473.50 & 97.6 & 100.0 & 100.0 \\
 &  &  & (+0.93\%) & (-1.87\%) &  & (+2.4) & (+2.4) \\ \midrule
5 & 112K & -913.27 & -986.88 & -880.16 & 0.0 & 0.0 & 0.0 \\
&  &  & (-8.06\%) & (+3.63\%) &  & (0) & (0) \\
& 352K & -5166.99 & -5548.18 & -5229.43 & 1.6 & 1.8 & 0.0 \\
 &  &  & (-7.38\%) & (-1.21\%) &  & (+0.2) & (-1.6) \\
 & 1536K & 1899.71 & 1860.05 & 1765.43 & 79.4 & 78.8 & 79.0 \\
 &  &  & (-2.09\%) & (-7.07\%) &  & (-0.6) & (-0.4) \\ \bottomrule
\end{tabular}
\end{table}

\subsubsection{Extending the Action Space with Shooting in PyReason}
Some simulations (e.g., Starcraft II) do not separate movement and shooting (i.e., the agent always shoots when in line of sight with an enemy). 
This, however, is clearly undesirable in any military simulator looking to emulate real battlefield scenarios. 
Strategies are often pragmatic, with shooting typically limited and highly tactical, given that practical issues such as limited ammunition and avoiding exposure are important considerations here. 
Hence, we build upon the basic scenario by integrating shooting into PyReason, independent from movement actions, allowing RL agents to learn varied and in-depth strategies, and in the process ensuring our implementation fits our eventual goal of a faithful miliary simulation. 
For this advanced scenario, each agent is provided with three bullets, and at each timepoint they may either choose to move, shoot, or to not take any action. Other than capturing the enemy base, a team can win by eliminating all enemy agents.

\begin{table*}[t]
\caption{Example rules in first order logic and descriptions in natural language.}
\label{tab:example_rules}
\begin{tabular}{p{0.14\textwidth}p{0.42\textwidth}p{0.35\textwidth}}
\toprule
{\bf Rule Identifier} & {\bf Rule} & {\bf English Description} \\
\midrule
m\_Down\_on & \(moveDown(A):[1,1] \xleftarrow[\Delta t=0]~~agent(A):[1,1] \wedge moveDir(A, down):[1,1] \wedge atLoc(A, X):[1,1] \wedge downLoc(Y, X):[1,1] \wedge blocked(Y):[0, 0]\) & If $A$ is an agent (annotated $[1,1]$) at location $X$, chooses to move in downward direction to $Y$ (which is not blocked), then  $moveDown(A)$'s label is updated to $[1,1]$.\\[4pt]
s\_Left\_on & \(shootLeftB(A):[1,1] \xleftarrow[\Delta t=0]~~agent(A):[1,1] \wedge team(A, blue):[1,1] \wedge health(A):[0.1,1] \wedge ammo(A):[0.1,1]  \wedge shootLeft(A):[1,1] \) & If $A$ is an agent on the blue team and chooses to shoot left, then $shootLeftB(A)$'s label is updated to $[1,1]$ iff $A$ has non-zero health and remaining ammo.\\
\bottomrule
\end{tabular}
\end{table*}

\subsubsection{Learning policies with RL} 

Since our approach is agnostic to any specific RL algorithm, for this work we chose to use the widely popular and versatile Deep Q learning (DQN) algorithm~\cite{mnih2015human} for all of our experiments. 
Based on a specific application or domain, a suitable algorithm can be seamlessly used in place of DQN. 
In our implementation, we combine a shallow Q-Net architecture with techniques discussed in~\cite{mnih2015human} such as experience replay, stable learning, and hard updates for the target network. In our architecture, we use one hidden layer between the input and output layers, 64~state variables (one for each grid cell), and an action space of~5 (for the base scenario) or~9 (for the advanced scenario).
The observation state space available to the agent is symbolic in nature, and its size varied between experimental setups as follows:
\begin{enumerate}
    \item \em{Four} for single agent in the base scenario: two each for the current positions of the agent and the opponent.
    \item \em{Seven} for single agent in the advanced scenario: one for the number of opponent bullets in the environment, two for the nearest bullet position, and two each for the current positions of the agent and the opponent.
\end{enumerate}
For multi-agent setups, the observation space is multiplied by the number of agents in each team. For the special non-Markovian setup described later, the observation space is doubled as observations from previous the timestep are considered. For experiments in multi-agent environments, we learn non cooperative single agent policies using multi-agent sampling. We use the widely adopted Smooth L1 loss function, instead of gradient clipping as described in the seminal DQN work.

We use the following reward function (rewards related to shooting actions are only applicable to the advanced scenario):
\label{sec:reward_fn}
\begin{enumerate}
    \item Terminal state rewards: +250 for a win, -250 for a loss, +400 for shooting an opponent, -200 for getting shot.
    \item Non-terminal state rewards: -2 for a valid action, -200 for an unsafe or illegal action, -10 for an invalid action (such as trying to shoot after exhausting ammunition).
\end{enumerate}

We define the behavior of the opponent using a stochastic base policy, which at each timestep tries to move closer to the enemy base by reducing the manhattan distance with a probability of~0.7, or chooses a random action from the action space with a probability of~0.3. In the advanced scenario, shooting is prioritized over movement until ammo is exhausted.
All RL policies described in this paper were learned on an NVIDIA A100 GPU with 80GB memory. and 40 cores of AMD EPYC 7413 with 378GB memory.

\subsubsection{Shielding in RL} 
As discussed in Section~\ref{sec:intro}, we incorporate logic shielding within the reward function, as well as the simulation environment itself. 
In the reward function, the agent is heavily penalized for taking an unsafe action, such as trying to move through the mountains or choosing an action that takes it out of bounds. While this approach encourages the agent to learn policies that avoid unsafe actions, it provides no guarantees. Adding shielding in the simulator itself ensures that even if the agent was to choose an unsafe action, our rule-based environment dynamics can detect and stop the execution of such actions in runtime. Furthermore, we can leverage these dynamics to prevent illegal actions, such as shooting when ammo has already been exhausted.

\subsubsection{Exploring Limitations of the Markov Assumption}

The Markov assumption in RL is the assumption that the next state of an agent only depends on its current state and action, and not on the history of states and actions that led to the current state. As this simplifies the problem and enables the use of techniques like Markov Decision Processes (MDPs) and the Bellman equation, many well-established simulators make this assumption. However, many real-world environments are not truly Markovian, since in some cases the current state may not contain all the relevant information for decision-making. This is especially important for simulators replicating realistic military combat environments where various key factors like logistical support, conflict history, long-term intelligence data, and patterns in surveillance reports, which go into tactical decision making, are non-Markovian in nature.

PyReason does not make a Markov assumption, and we now showcase this capability by creating a simple experiment with non-Markovian dynamics. 
We consider a two-agents per team advanced scenario as described earlier. We introduce a modification to one agent within each team, constraining its ability to execute actions to once every two timesteps, with the added stipulation that each of its movement actions require two timesteps to complete. We learn to play the game in two different ways. In the initial approach, the player adheres to a Markov assumption, leveraging solely the current state's information. Conversely, in the second approach, the player gains access not only to the present state data but also to observations from the preceding time step. We compare the success of the two methods by evaluating learned policies over~500 games after every 32,000 training epochs.

Evolution of the performance of policies learned with and without the Markov assumption is shown in Figure~\ref{fig:shoot-markov-non-markov-comp}. 
Both agents underwent training for a duration of up to~1.6 million epochs, with policy evaluations conducted at intervals of 32,000 epochs. 
Each policy was used to play the advanced scenario~500 times in order to obtain a win percentage. Evaluations were carried out on~48 cores of AMD EPYC 7413 with~378GB memory.
Markovian policies obtained a peak performance of 59\%, significantly lower than the 85\% achieved by the non-Markovian policies. 
However, we observe that policies learned in the Markovian assumption setting attained decent performance with noticeably less training, which is unsurprising given the doubling of the observation space in the non-Markovian case. When examining the most effective policy within each category, the removal of the Markov assumption resulted in an increase in the average number of actions per agent required to secure a single victory, rising from~15.51 to~18.01. 
This observation suggests the acquisition of a policy characterized by greater complexity, yet one that exhibits enhanced reliability.
Despite the relative simplicity of our experiment, a noteworthy performance enhancement was observed. This underscores the essentiality and significance of accommodating non-Markovian dynamics within simulation environments.

\begin{figure}[t]
    \begin{center}
        \includegraphics[width=0.5\columnwidth, trim=2 2 2 2, clip]{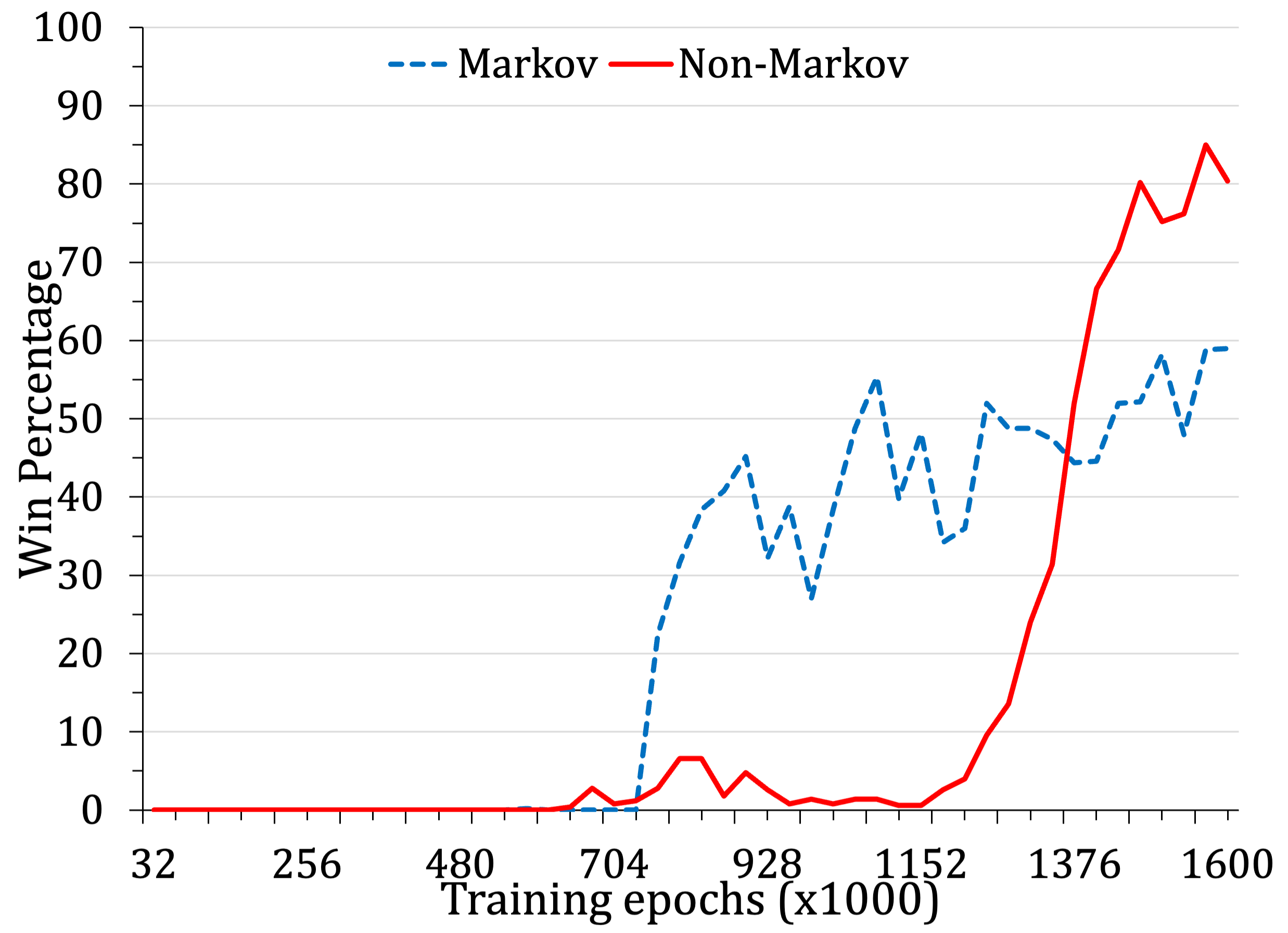}
    \end{center}
    \caption{Win percentage for policies learned with Markovian and non-Markovian dynamics.}
    \label{fig:shoot-markov-non-markov-comp}
\end{figure}

Win percentage over 500 trials, for policies learned with non-Markovian dynamics is shown in Figure~\ref{fig:shoot-markov-non-markov-comp}. Each team is made up of one fast-moving and one slow-moving agent. The action space is extended to include two timesteps.

\subsubsection{Explainability} 
One the major drawbacks of deep learning-based systems is the difficulty associated with understanding the output in terms of how it was computed. 
On the other hand, logic programs inherently support this kind of understanding. PyReason works with graphs using first order logical rules and produces an explainable trace detailing rules fired at different timesteps, constants used for grounding, and interpretation changes---an example is shown in Table~\ref{tab:example_rules}. 
The explainable trace is a direct result of the structures leveraged in computational logic; this makes our approach explainable, allowing the user to understand system behavior and debug errors. 

Two examples of how we leveraged this to improve our reward function given in Section~\ref{sec:reward_fn} are:
\begin{enumerate}
    \item Initially, we had set the penalty for getting shot at~400. However, from rule traces we observed that the agent was learning to prioritize hiding behind impenetrable mountains and take a safety first approach, instead of trying to win the game. Halving the penalty to~200 produced a more balanced policy.
    \item The penalty for trying to shoot after exhausting ammunition was set to a lower value of 10 after observing that higher values led to the agent avoiding shooting altogether.
\end{enumerate}

An excerpt of a rule trace is shown in Table~\ref{tab:short_rule_trace}; it begins at timestep~16 of one of our experiments. 
Initial conditions are as depicted in Figure~\ref{fig:game-setup}. ``R'' and ``B'' respectively show the location of the red and blue agents at the beginning of this example. As the red agent moves downward from its starting location (from ``24'' to ``0'' through ``16'' and ``8''), the blue agent decides to shoot to the left so as to intercept red (at ``0''). However, red has seemingly learned to predict the bullet path and evade it. So it backtracks (to ``16'').
Rule \textbf{m\_Down\_on} presented in Table~\ref{tab:example_rules} is fired at timestep~16 (as well as 17 and 18), and is pictorially shown with a red arrow in Figure~\ref{fig:game-setup}, and in bold in Table~\ref{tab:short_rule_trace}.

\begin{table}[tb]
\caption{An extract of a rule trace produced by the PyReason software.}
\label{tab:short_rule_trace}
\begin{tabular}{cllccl}
\toprule
t & Constant Symbols & Predicate & Old Annotation & New Annotation & Rule fired \\ \midrule
0 & 26 & blocked & {[}0.0,1.0{]} & {[}1.0,1.0{]} & -- \\
0 & 27 & blocked & {[}0.0,1.0{]} & {[}1.0,1.0{]} & -- \\ \midrule
\textbf{16} & \textbf{red-agent-1} & \textbf{moveDown} & \textbf{{[}0.0,0.0{]}} & \textbf{{[}1.0,1.0{]}} & \textbf{m\_Down\_on} \\ \midrule
17 & red-agent-1 & moveDown & {[}1.0,1.0{]} & {[}0.0,0.0{]} & m\_Down\_off \\
17 & (red-agent-1,16) & atLoc & {[}0.0,1.0{]} & {[}1.0,1.0{]} & m\_Set\_location \\
17 & (red-agent-1,24) & atLoc & {[}1.0,1.0{]} & {[}0.0,0.0{]} & m\_Rem\_location \\
17 & red-agent-1 & moveDown & {[}0.0,0.0{]} & {[}1.0,1.0{]} & m\_Down\_on \\ \midrule
18 & red-agent-1 & moveDown & {[}1.0,1.0{]} & {[}0.0,0.0{]} & m\_Down\_off \\
18 & (red-agent-1,8) & atLoc & {[}0.0,1.0{]} & {[}1.0,1.0{]} & m\_Set\_location \\
18 & (red-agent-1,16) & atLoc & {[}1.0,1.0{]} & {[}0.0,0.0{]} & m\_Rem\_location \\
18 & blue-agent-1 & shootLeftB & {[}0.0,1.0{]} & {[}1.0,1.0{]} & s\_Left\_on \\
18 & (blue-bullet-1,3) & atLoc & {[}0.0,1.0{]} & {[}1.0,1.0{]} & s\_Set\_location \\
18 & (blue-bullet-1,left) & direction & {[}0.0,1.0{]} & {[}1.0,1.0{]} & s\_Set\_dir \\
18 & red-agent-1 & moveDown & {[}0.0,0.0{]} & {[}1.0,1.0{]} & m\_Down\_on \\ \midrule
19 & red-agent-1 & moveDown & {[}1.0,1.0{]} & {[}0.0,0.0{]} & m\_Down\_off \\
19 & (red-agent-1,0) & atLoc & {[}0.0,1.0{]} & {[}1.0,1.0{]} & m\_Set\_location \\
19 & (red-agent-1,8) & atLoc & {[}1.0,1.0{]} & {[}0.0,0.0{]} & m\_Rem\_location \\
19 & blue-agent-1 & shootLeftB & {[}1.0,1.0{]} & {[}0.0,0.0{]} & s\_Left\_off \\
19 & (blue-bullet-1,3) & atLoc & {[}1.0,1.0{]} & {[}0.0,0.0{]} & s\_Rem\_location \\
19 & (blue-bullet-1,2) & atLoc & {[}0.0,1.0{]} & {[}1.0,1.0{]} & s\_Set\_location \\
19 & red-agent-1 & moveUp & {[}0.0,0.0{]} & {[}1.0,1.0{]} & m\_Up\_on \\ \midrule
20 & red-agent-1 & moveUp & {[}1.0,1.0{]} & {[}0.0,0.0{]} & m\_Up\_off \\
20 & (red-agent-1,8) & atLoc & {[}0.0,0.0{]} & {[}1.0,1.0{]} & m\_Set\_location \\
20 & (red-agent-1,0) & atLoc & {[}1.0,1.0{]} & {[}0.0,0.0{]} & m\_Rem\_location \\
20 & (blue-bullet-1,2) & atLoc & {[}1.0,1.0{]} & {[}0.0,0.0{]} & s\_Rem\_location \\
20 & (blue-bullet-1,1) & atLoc & {[}0.0,1.0{]} & {[}1.0,1.0{]} & s\_Set\_location \\
20 & red-agent-1 & moveUp & {[}0.0,0.0{]} & {[}1.0,1.0{]} & m\_Up\_on \\ \midrule
21 & red-agent-1 & moveUp & {[}1.0,1.0{]} & {[}0.0,0.0{]} & m\_Up\_off \\
21 & (red-agent-1,16) & atLoc & {[}0.0,0.0{]} & {[}1.0,1.0{]} & m\_Set\_location \\
21 & (red-agent-1,8) & atLoc & {[}1.0,1.0{]} & {[}0.0,0.0{]} & m\_Rem\_location \\
21 & (blue-bullet-1,1) & atLoc & {[}1.0,1.0{]} & {[}0.0,0.0{]} & s\_Rem\_location \\
21 & (blue-bullet-1,0) & atLoc & {[}0.0,1.0{]} & {[}1.0,1.0{]} & s\_Set\_location \\
\bottomrule
\end{tabular}
\end{table}

\section{Conclusions and Future Work}
\label{sec:conclusion}

This work introduces \logic, a logic programming framework that integrates temporal extensions with a lower lattice annotation structure to model non-Markovian temporal relationships while ensuring tractable and scalable exact reasoning. By departing from the standard Markov assumption, \logic~enables reasoning about dependencies spanning multiple past time steps, capturing complex dynamic behaviors that traditional approaches like Markov Decision Processes cannot represent. Through rigorous theoretical analysis, we prove the correctness, convergence, and inconsistency detection capability of a fixpoint operator for this logic, and demonstrate how the use of a lower lattice facilitates Skolemization that significantly reduces the amount of grounding required. 
Our implementation, called PyReason, leverages these properties to efficiently perform reasoning over large-scale, sparse domains common in real-world applications such as multi-agent geospatial simulations, knowledge graph completion, and reinforcement learning. 
Empirically, we showed that \logic~achieves multiple orders of magnitude improvements in grounding size, computational speed, and memory efficiency, making previously intractable reasoning tasks feasible. Moreover, by integrating non-Markovian dynamics into simulation environments, we show notable gains in reinforcement learning performance, highlighting the practical importance of tractable non-Markovian reasoning. This work thus bridges a critical gap by providing both a theoretically sound and practically scalable approach for modeling and reasoning about non-Markovian temporal dynamics in intelligent systems.

Looking ahead, several promising directions emerge for extending this framework. Incorporating probabilistic reasoning stands out as a natural next step, as APT logic~\cite{APTL}, which combined temporal logic with probabilistic semantics, suffers from intractability. Leveraging recent advances in tractable probabilistic circuit learning, as developed by Choi et al.~\cite{choi2020probabilistic}, could enable learning probability distributions within \logic’s efficient and tractable semantics, allowing a rich yet computationally feasible representation of uncertainty beyond deterministic intervals. 
Another potential avenue is constructing logic programs using large language models, as introduced in Logic LM~\cite{logiclm}, which demonstrated the use of LLMs with symbolic solvers for reasoning tasks---extensions to temporal logic remain unexamined in this front. 
Finally, applying Inductive Logic Programming approaches~\cite{deepMinIlp2018} to automatically learn temporal and non-Markovian rules within \logic~offers a compelling route to scale and adapt the framework to data-driven scenarios where expert knowledge is limited or unavailable, providing a pathway to fully automated, explainable temporal reasoning systems.
PyReason has already been successfully applied in several domains, including reasoning about medical triage optimization~\cite{patil2025triage}, integrating machine learning models with temporal logic for process automation~\cite{aditya2025processautomation}, as well as abductive reasoning in vision~\cite{leiva2025consistency} and geospatial applications~\cite{bavikadi2025humanmove}. A promising direction for future work is to explore abductive queries in a more general way, further enhancing the framework’s reasoning capabilities and applicability across diverse problems.

\begin{acks}
    Some of the authors were funded by Scientific Systems Company, Inc. (SSCI).
\end{acks}

\bibliographystyle{ACM-Reference-Format}
\bibliography{rl-py}

\appendix
\input{si.tex}

\end{document}

%% file: macros.tex
\newtheorem{definition}{Definition}[section]
\newtheorem{example}{Example}[section]
\newtheorem{proposition}{Proposition}[section]
\newtheorem{theorem}{Theorem}[section]

\newtheorem{corollary}{Corollary}[section]

\newtheorem{example-1}{KB Example}[section]
\newtheorem{example-2}{ST Example}[section]

\newcommand{\Algphase}[1]{%
  \Statex \hrulefill
  \Statex \textbf{\underline{#1}}
}

\def\calf{\mathcal{F}}

\def\avar{\textsf{AVar}}

\newcommand{\constantSet}{\mathcal{C}}

\newcommand{\predicateSet}{\mathcal{P}}
\newcommand{\variableSet}{\mathcal{V}}

\newcommand{\groundLiteralSet}{\mathcal{G}}
\newcommand{\interpretationSet}{\mathcal{I}}
\newcommand{\interpretation}{I}
\newcommand{\semiLattice}{\mathcal{M}}
\newcommand{\delay}{\Delta t}
\newcommand{\satisfaction}{\models}
\newcommand{\satisfactionAtTime}[1]{\models_{#1}}
\newcommand{\program}{\Pi}
\newcommand{\fixpointOperator}{\Gamma}
\newcommand{\groundedLiteral}{g}
\newcommand{\nonGroundedLiteral}{\ell}

\def\avar{\textsf{AVar}}

\newcommand{\shortlogic}{LAT}
\newcommand{\logic}{LAT logic}

\def\avar{\textsf{AVar}}

%% file: si.tex
\section{Complete Proof for Theorem~\ref{theorem:ga-theorem} \label{app:proofs}}

\begin{proof}
Let $P_{\program} \subseteq P$ be the set of predicates containing only predicates present in the head of at least one rule in $\program_{\textit{Rules}}$.
\begin{eqnarray}
    |g_i| &=& |\bigcup_{p \in P} g_i(p)| \nonumber \\
    &=& \sum_{p \in P} |g_i(p)| \nonumber \\
    &=& \sum_{p \in P_{\Pi}} |g_i(p)| + \sum_{p \notin P_{\Pi}} |g_i(p)| \nonumber \\
    &=& \sum_{p \in P_{\Pi}} |g_i(p)| + \sum_{p \notin P_{\Pi}} |g_0(p)| \label{eq:gi_main_a}
\end{eqnarray}
Let, $\fixpointOperator_{r}(g)$ denote the set of ground atoms produced when a single fixpoint operator is applied to a single rule $r$ with the set of ground atoms $g$. 
\begin{flalign}
    |g_i(p)| = & |g_{i-1}(p) \cup \bigcup_{r \in \Pi_{\textit{rules}}~\wedge~ \textit{pred}(\textit{head}(r))= p} \fixpointOperator_{r}(g_{i-1})| \notag \\
    = & |g_{i-1}(p)| + \textit{newF}_{p,i} \times |\bigcup_{r \in \Pi_{\textit{rules}}~\wedge~ \textit{pred}(\textit{head}(r))= p} \fixpointOperator_{r}(g_{i-1})| \notag\\
    = & |g_{i-1}(p)| + \textit{newF}_{p,i} \times \textit{uniqueF}_{p,i} \times \sum_{r \in \Pi_{\textit{rules}}~\wedge~ \textit{pred}(\textit{head}(r))= p} |\fixpointOperator_{r}(g_{i-1})| \label{eq:gip_main_a}
\end{flalign}
Here, $\textit{newF}_{p,i} \in [0,1]$ denotes the fraction of ground atoms produced, with predicate $p$ and at the $i^{\textit{th}}$ $\Gamma$ application, which did not exist after the $(i-1)^{\textit{th}}$ application. 
Similarly, $\textit{uniqueF}_{p,i} \in [0,1]$ is the fraction of ground atoms produced across rules, with predicate $p$ in the head, which are unique.
\begin{align}
    |\fixpointOperator_{r}(g_{i-1})| &\leq&  \prod_j |g_{i-1}(\textit{pred}(\textit{body}(r), j))|  \nonumber \\
    \label{eq:gamma_r_main_a}
    |\fixpointOperator_{r}(g_{i-1})| &=& \textit{validF}_{r,i} \times \prod_j |g_{i-1}(\textit{pred}(\textit{body}(r), j))|
\end{align}
Here, $\textit{validF}_{r,i} \in [0,1]$ denotes the fraction of valid groundings that leads to firing of non-ground rule $r$, within the cross-product of possible groundings for each body clause.

\noindent
From Eqs.~\eqref{eq:gip_main_a} and~\eqref{eq:gamma_r_main_a} we get:
\begin{align}
\label{eq:gip_second_a}
|g_i(p)| = |g_{i-1}(p)| + \textit{newF}_{p,i} \times \textit{uniqueF}_{p,i} \sum_{\substack{r \in \Pi_{\textit{rules}} \\
    \textit{pred}(\textit{head}(r))= p}} \textit{validF}_{r,i} 
    \times \prod_j |g_{i-1}(\textit{pred}(\textit{body}(r), j))|
\end{align}

\begin{align}
    |g_i(p)| - |g_{i-1}(p)| = & \textit{newF}_{p,i} \times \textit{uniqueF}_{p,i} \sum_{\substack{r \in \Pi_{\textit{rules}} \\
    \textit{pred}(\textit{head}(r))= p}} \textit{validF}_{r,i} \times \prod_j |g_{i-1}(\textit{pred}(\textit{body}(r), j))|
\end{align}
\begin{align}
\label{eq:delta_gip_a}
    \Delta |g_i(p)| = \textit{newF}_{p,i} \times \textit{uniqueF}_{p,i} \sum_{\substack{r \in \Pi_{\textit{rules}} \\ \textit{pred}(\textit{head}(r))= p}} \textit{validF}_{r,i} 
    \times \prod_j |g_{i-1}(\textit{pred}(\textit{body}(r), j))| \nonumber
\end{align}

\noindent
Considering the maximum value (= 1) for all three fractions:
\begin{equation*}
    \Delta |g_i(p)| \leq \sum_{\substack{r \in \Pi_{\textit{rules}} \\ \textit{pred}(\textit{head}(r))= p}} ~ \prod_j |g_{i-1}(\textit{pred}(\textit{body}(r), j))|
\end{equation*}

\noindent
Substituting Equation~\eqref{eq:gip_second_a} into Equation~\eqref{eq:gi_main_a} we obtain:
\begin{align}
    |g_i| =& \sum_{p \in P_{\Pi}} \Bigg[ |g_{i-1}(p)| + \textit{newF}_{p,i} \times \textit{uniqueF}_{p,i} \sum_{\substack{r \in \Pi_{\textit{rules}} \\ \textit{pred}(\textit{head}(r))= p}} \textit{validF}_{r,i} \prod_j |g_{i-1}(\textit{pred}(\textit{body}(r), j))| \Bigg] + \sum_{p \notin P_{\Pi}} |g_0(p)| \nonumber
\end{align}
\begin{align}
    |g_i| =& \sum_{p \in P_{\Pi}} |g_{i-1}(p)| + \sum_{p \notin P_{\Pi}} |g_0(p)| \nonumber \\
    &+ \sum_{p \in P_{\Pi}} \Bigg[ \textit{newF}_{p,i} \times \textit{uniqueF}_{p,i} \sum_{\substack{r \in \Pi_{\textit{rules}} \\ \textit{pred}(\textit{head}(r))= p}} \textit{validF}_{r,i} \times \prod_j |g_{i-1}(\textit{pred}(\textit{body}(r), j))| \Bigg]
\end{align}
\begin{align}
    & |g_i| = |g_{i-1}| + \sum_{p \in P_{\Pi}} \textit{newF}_{p,i} \times \textit{uniqueF}_{p,i} \sum_{\substack{r \in \Pi_{\textit{rules}} \\ \textit{pred}(\textit{head}(r))= p}} \textit{validF}_{r,i} \times \prod_j |g_{i-1}(\textit{pred}(\textit{body}(r), j))| \nonumber
\end{align}
\begin{align}
\label{eq:delta_gi_final_a}
    \Delta |g_i| =& \sum_{p \in P_{\Pi}} \textit{newF}_{p,i} \times \textit{uniqueF}_{p,i} \sum_{\substack{r \in \Pi_{\textit{rules}} \\ \textit{pred}(\textit{head}(r))= p}} \textit{validF}_{r,i} \prod_j |g_{i-1}(\textit{pred}(\textit{body}(r), j))|
\end{align}

\noindent
Considering the maximum value for all three fractions:
\begin{equation*}
    \Delta |g_i| \leq \sum_{p \in P_{\Pi}}  ~ \sum_{\substack{r \in \Pi_{\textit{rules}} \\ \textit{pred}(\textit{head}(r))= p}}~ \prod_j |g_{i-1}(\textit{pred}(\textit{body}(r), j))|
\end{equation*}
which further simplifies to:
\begin{equation}
    \label{eq:delta_gi_bound_a}
    \Delta |g_i| \leq \sum_{r \in \Pi_{\textit{rules}}} ~ \prod_j |g_{i-1}(\textit{pred}(\textit{body}(r), j))|
\end{equation}
\end{proof}

\section{Reproducibility guide \label{app:reproducibility}}

All experiments carried out to obtain the results shown in Section \ref{sec:experiments} use the PyReason framework with specific configurations. Geospatial experiments use the PyReason configuration settings shown in Table~\ref{tab:pyreason_config}. The parameter $ad\_hoc\_grounding$ acts as the key change between the Skolemization-enabled approach and the traditional full grounding approach. For knowledge graph completion experiments, we use the same PyReason configuration settings as presented in Table~\ref{tab:pyreason_config} with the exception that we only use $ad\_hoc\_grounding = False$ since this refers to node-level Skolemization which we do not need for knowledge graph completion experiments. Thus, the $resolution\_levels$ parameter is not applicable for Knowledge graph completion tasks. Additionally, knowledge graph completion experiments used AnyBURL for learning rules with rule snapshot at twenty minutes for all datasets. The following are the commands to execute Python scripts for geospatial and knowledge graph completion experiments.

\begin{table}[tb]
\caption{PyReason configuration settings for geospatial experiments.}
\label{tab:pyreason_config}
\begin{tabular}{lll}
\toprule
Setting & Value & Description\\
\midrule
\textit{verbose} & True & Print all info to screen during reasoning.\\
\textit{atom\_trace} & False & Groundings untracked, reducing overhead for large graphs.\\
\textit{persistent} & False & Interpretations are not reset to bottom of the lattice after every timestep.\\
\textit{static\_graph\_facts} & False & Interpretations in the input graph are allowed to change during reasoning.\\
\textit{parallel\_computing} & True & Use parallel processing.\\
\textit{ad\_hoc\_grounding} & True & Use skolemization.\\
\textit{resolution\_levels} & {2,3,4,5,6,7,8,9} & Grid Size $= (2^{resolution\_levels})^2$.\\

\bottomrule
\end{tabular}
\end{table}

\subsection{Skolemization approach}
\begin{verbatim}
    python3 play_random_game.py --resolution 5 --field_soldiers_per_team 10 
  --border_soldiers_per_team 10 --actions_per_soldier 100 --ad_hoc
\end{verbatim}

\subsection{Full grounding approach}
\begin{verbatim}
    python3 play_random_game.py --resolution 5 --field_soldiers_per_team 10 
  --border_soldiers_per_team 10 --actions_per_soldier 100
\end{verbatim}

\subsection{Knowledge Graph completion}
\begin{verbatim}
    python3 anyBurl_multistep_multirule.py -rf yago_1200_99_100_ann 
    -s 1 -e 1000 -ts 10 -g anyBurl_graphs/YAGO3-10/knowledge_graph_train.graphml
\end{verbatim}

\section{Example pipeline using PyReason}

\subsection{Knowledge Graph completion} 
Both the train and test set of any knowledge graph are set of triples. We show the example triple from the train set of the YAGO03-10 dataset that we need for our example as follows.

\begin{verbatim}
Chelsy_Davy	playsFor	Panathinaikos_F.C.
\end{verbatim}

The next step is to convert the training triples into graphs that can be inputted into PyReason. An example of such a graph representation is shown in Figure~\ref{fig:ex_kgc}.

\begin{figure}[bt]
    \begin{center}
        \includegraphics[width=0.6\linewidth]{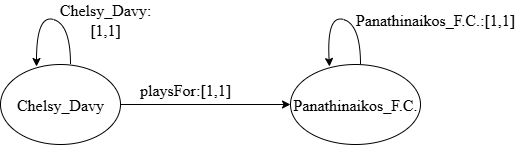}
    \end{center}
    \caption{Graph representation of example triple to work with PyReason.}
    \label{fig:ex_kgc}
\end{figure}

We then convert the learned AnyBURL rules into PyReason rules. The sample AnyBURL rule used in our example is as follows:

\begin{verbatim}
    0.934	isAffiliatedTo(X,Panathinaikos_F.C.) <= playsFor(X,Panathinaikos_F.C.)
\end{verbatim}

The first number is the confidence value of the rule, which becomes the lower bound in the converted rule. Further, partially ground rules are made fully non-grounded to make grounding faster. The converted non-ground PyReason rule is then,

\begin{verbatim}
    isAffiliatedTo(X,X_0):[0.934,1] <-1 playsFor(X,X_0):[0.1,1],
                                        Panathinaikos_F.C.(X_0):[1,1]
\end{verbatim}

After inference, an additional edge is added to the graph. The updated graph after inference is shown in Figure~\ref{fig:ex_kgc_updated}.
\begin{figure}[bt]
    \begin{center}
        \includegraphics[width=0.6\linewidth]{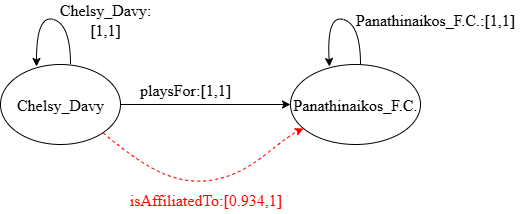}
    \end{center}
    \caption{Graph representation of example triple after Inference.}
    \label{fig:ex_kgc_updated}
\end{figure}

\subsection{Geospatial Application}

\begin{figure}[tb]
    \begin{center}
        \includegraphics[width=0.35\linewidth]{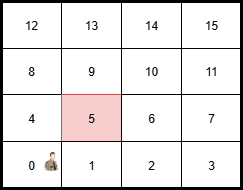}
    \end{center}
    \caption{Example grid for agent movement in Geospatial application. In this case, note that cell \#5  has an obstacle, hence the agent is not allowed to move there.}
    \label{fig:ex_geo_grid}
\end{figure}

For the geospatial skolemization experiment, we convert a grid map shown in Figure~\ref{fig:ex_geo_grid} into a graph structure shown in Figure~\ref{fig:geo_graphs} (a). Note that here we consider a grid map with the least grid size, i.e., 16, with the initial graph passed to PyReason in the case of Skolemization. For the non-Skolemization case, the complete graph with all grid points needs to be passed to PyReason. The following two non-ground rules are fired when the policy wants the agent to move in the right direction on the border of the grid:

\begin{verbatim}    
atLoc(AGENT, NEWLOC):[1,1] <-2 moveRight(AGENT):[1,1], borderAgent(AGENT):[1,1],
                        atLoc(AGENT, OLDLOC):[1,1], right(OLDLOC, NEWLOC):[1,1],
                        borderLoc(NEWLOC):[1,1], blocked(NEWLOC):[0,0]

atLoc(AGENT, OLDLOC):[0,0] <-2 moveRight(AGENT):[1,1], borderAgent(AGENT):[1,1],
                        atLoc(AGENT, OLDLOC):[1,1], right(OLDLOC, NEWLOC):[1,1],
                        borderLoc(NEWLOC):[1,1], blocked(NEWLOC):[0,0]
\end{verbatim}

Figure~\ref{fig:geo_graphs} (b) shows the resulting graph after two timesteps of reasoning.

\begin{figure}[hbt]
    \begin{center}
        \begin{subfigure}[t]{0.45\columnwidth}
            \vskip 0pt
            \includegraphics[width=\linewidth]{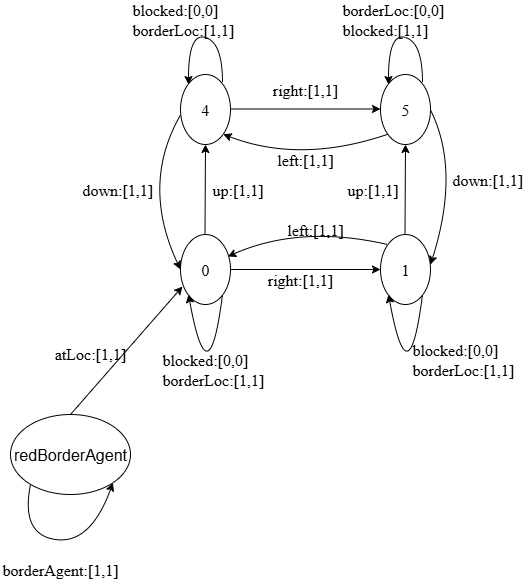}
            \caption{}
        \end{subfigure}
        \begin{subfigure}[t]{0.45\columnwidth}
            \vskip 0pt
            \includegraphics[width=\linewidth]{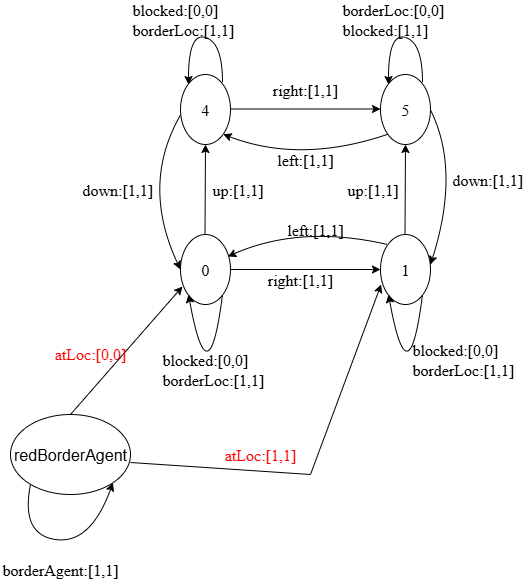}
            \caption{}
        \end{subfigure}
    \end{center}
    \caption{(a) Initial PyReason graph representation of grid and agent location. (b) Updated PyReason graph after two timesteps of inference.}
    \label{fig:geo_graphs}
\end{figure}